\documentclass[12pt]{article}

\usepackage[margin=1in]{geometry}
\usepackage{amssymb,amsmath,amsthm,amsfonts}
\usepackage{appendix}
\usepackage{arydshln}
\usepackage{array}
\usepackage[ruled,vlined,linesnumbered]{algorithm2e}
\usepackage{booktabs}
\usepackage{colortbl}
\usepackage{color}
\usepackage[small, margin=1cm]{caption}
\usepackage{enumitem}
\usepackage{hhline}
\usepackage{multirow}
\usepackage{mathrsfs}
\usepackage{multicol}
\usepackage{bbm}
\usepackage[switch]{lineno}
\usepackage{subfigure}
\usepackage{tikz}
\usepackage{xcolor}
\usepackage{makecell}
\usepackage{natbib}
\usepackage[colorlinks=true,urlcolor=blue,citecolor=blue,linkcolor=red,bookmarks=true]{hyperref}
\usepackage{xr}

\makeatletter
\newcommand*{\addFileDependency}[1]{%
  \typeout{(#1)}%
  \@addtofilelist{#1}%
  \IfFileExists{#1}{}{\typeout{No file #1.}}%
}
\makeatother

\newcommand*{\myexternaldocument}[1]{%
  \externaldocument{#1}%
  \addFileDependency{#1.tex}%
  \addFileDependency{#1.aux}%
}
\myexternaldocument{appendix}

\newtheorem{theorem}{Theorem}
\newtheorem{lemma}[theorem]{Lemma}

\newtheorem{proposition}{Proposition}
\newtheorem{definition}{Definition}
\newtheorem{corollary}{Corollary}

\title{Approximate Revenue Maximization \\for Diffusion Auctions}
\author{
  Yifan Huang\thanks{Equal contribution.} \and
  Dong Hao\thanks{Equal contribution. Corresponding author: \href{mailto:haodongpost@gmail.com}{haodongpost@gmail.com}} \and
  Zhiyi Fan \and
  Yuhang Guo \and
  Bin Li
}
\date{}

\begin{document}
\maketitle
\vspace{-30pt}
\begin{center}
\textit{ University of Electronic Science and Technology of China \\
  University of New South Wales \\
  Nanjing University of Science and Technology
}
\end{center}

\vspace{20pt}
\begin{abstract}
Reserve prices are widely used in practice. The problem of designing revenue-optimal auctions based on reserve price has drawn much attention in the auction design community. Although they have been extensively studied, most developments rely on the significant assumption that the target audience of the sale is directly reachable by the auctioneer, while a large portion of bidders in the economic network unaware of the sale are omitted. This work follows the diffusion auction design, which aims to extend the target audience of optimal auction theory to all entities in economic networks. We investigate the design of simple and provably near-optimal network auctions via reserve price. Using Bayesian approximation analysis, we provide a simple and explicit form of the reserve price function tailored to the most representative network auction. We aim to balance setting a sufficiently high reserve price to induce high revenue in a successful sale, and attracting more buyers from the network to increase the probability of a successful sale. This reserve price function preserves incentive compatibility for network auctions, allowing the seller to extract additional revenue beyond that achieved by the Myerson optimal auction. Specifically, if the seller has $\rho$ direct neighbours in a network of size $n$, this reserve price guarantees a $1-{1 \over \rho}$ approximation to the theoretical upper bound, i.e., the maximum possible revenue from any network of size $n$. This result holds for any size and any structure of the networked market.
\end{abstract}


\section{Introduction}
\noindent 



 
%

Designing revenue optimal auctions to sell an item to symmetric bidders is one of the most fundamental problems in mechanism design. 
The optimal auction is the one with the highest expected revenue.
In auction theory, attracting together a larger group of customers in an auction has the potential to enhance not only the seller's revenue but also overall social welfare \citep{Bulow1996Auction}.  
With symmetric i.i.d. assumption, the \textit{number of bidders} and the \textit{reserve price} jointly decide the seller's revenue in a mechanism of allocation rule and payment rule. The seminal result by  \citep{Bulow1996Auction} demonstrates the power of competition for extracting revenue: 
the revenue from SPA with  $\rho+1$ bidders, is higher
than the revenue from the optimal auction with $\rho$ bidders. The gist of this theorem is that besides optimizing the mechanism (i.e., allocation rule, payment rule and reserve price), attracting
participation is more important. However,  most works in optimal auction design are under the basic assumption that the number of bidders is fixed, while how to attract bidders is not under consideration. Despite the evident importance of social networks in almost all human-activated systems \citep{jackson2008social, borgatti2009network}, standard auction models only concentrate on realizing preferred allocations and payments within a fixed group of bidders that are directly reachable by the seller. This limits the audience of the auction to a very limited group, and does not align with real markets, which operate in the context of significant social and economic relationships between entities, playing a pivotal role in economic dynamics. 


Our work is based on the idea that if the auction mechanism can attract more bidders, additional high-value bidders have the opportunity to participate in the auction. Consequently, \textit{setting a reasonably higher reserve price for this enlarged audience} can increase the seller's revenue. 
In contrast to classical optimal auction which does not consider the underlying social network structure among the agents, we investigate the reserve price for diffusion auctions,  that view the  audience for optimal auction design as unfixed, and incentivize each participated bidder to invite extra bidders from their neighbourhood to join the auction.
This essentially relaxes the commonly overlooked assumption in classical optimal auctions, which presumes a fixed number, denoted as $\rho$, of bidders. By doing so, it unleashes the potential of competition from a more extensive networked audience. This shift broadens the target market from the `thin market' of classical optimal auctions to a much larger economic network, which can afford a higher reserve price. In addition to the $\rho$ bidders initially accessible to the auctioneer, there exists a substantial segment of potential buyers who are originally hidden to the auctioneer  through whom the higher reserve price can be set.  

More formally, we consider an economic network  $G=\left(N, E\right)$, where $N$ encompasses all potential bidders in the network and only $\rho$ of them are direct neighbors of the seller, and $E$ are the social links between all potential bidders.
That is, in addition to the $\rho$ bidders which are considered in the classical optimal auction design, there exists a substantial portion of potential buyers in $G$ who are hidden from the auctioneer but can be located via a some mechanism utilizing the graph structure.
The strategy of each bidder in network auctions is $\left(b_i, e_i\right)$, where $b_i$ is the bid while $e_i\subseteq E_i$ is any subset of $i$'s neighbours, meaning which neighbour $i$ invites. 


With this lens, the problem of seller's revenue optimization boils down to how to incentive each participant to do a spontaneous invitation, and then how to choose the the good reserve price. To this objective, we design  a novel network auction mechanism that can simultaneously enhance  the two factors in Bulow-Klemperer theorem: the number of bidders and  revenue optimization. 
This newly proposed optimal auction is proven to be dominant-strategy incentive-compatible in the sense that, on one hand, this mechanism incentivizes each current buyer to attract other buyers from the economic network into the sale, i.e., each agent $i$ truthfully invite $e_i'=e_i$. On the other hand, it incentivizes every participants to truthfully report her valuation $b_i=v_i$.

Built upon the DSIC property, this new mechanism optimizes and adapts the reserve price to the growing number of bidders attracted from the economic network. We prove that, for any economic network, it generates higher revenue compared to  Myerson optimal auction. 
For this new network auction mechanism whose revenue varies with the network structure, we derive the reserve price $\gamma$ which is given in a simple and explicit form. Its revenue approximates very closely  to the theoretical  upper-bound, that is the maximum possible revenue extracted from any possible economic networks. More precisely, its approximation ratio is at least $1-{1\over \rho}$, where  $\rho$ is the number of bidders considered in classical optimal auction.
Given the richness of the (truthful) network auction design space, it is interesting  that such a simple yet nearly optimal reserve price works very well.

\subsection{Reserve Price in Auctions}
Setting a reserve price is a typical way to improve the seller's revenue. 
The theory of auction with reserve price was originally given by  \citet{myerson1981optimal} and \citet{riley1981optimal}: A single seller faces a fixed number of  buyers, if the buyers’ valuations for the item are i.i.d. drawn from a known distribution, then the optimal auction that maximizes the seller's revenue is just the second price auction with a reserve price $r$. The allocation rule awards the item to the highest bidder, unless all bids are less than $r$, in which case no one gets the item. The corresponding payment rule charges the winner (if any) the second-highest bid or $r$, whichever is larger.  This result extends beyond the i.i.d. setting: when agents are independent but not identical, the optimal auction entails sorting the agents based on a function of their bid (referred to as the virtual value) and allocating the item to the agent with the highest non-negative virtual value.

%

Let's illuminate what is the core inside this simple optimal auction from Myerson and R\&S.
Regarding the revenue-maximizing reserve price $r$, it is surprising that it is independent of the number of buyers. Essentially, it is the bid which makes the inverse of the virtual valuation equal $0$, and this important virtual valuation
of a bidder depends on her own valuation and distribution, and not on those of the others.  
On the other hand, however, the expected revenue extracted by the optimal auction is decided by the number of bidders, the value distribution, and the reserve price. More precisely, it is the number of bidders times the expected payment, which  is further decided by the virtual valuations and  the reserve price-based allocation \citep{myerson1981optimal}. 
 The expected payment can also be interpreted in another form, by
 an integral with   the reserve price and the number of bidders as parameters \citep{riley1981optimal}.

The discussion of optimal reserve price regarding the number of bidders in standard auctions can also be seen in works by \citep{engelbrecht1987optimal, levin1996optimal, wolfstetter1996auctions, menicucci2021basic}. Recent progresses also investigate how many more additional bidders need to join the auction in order to surpass the revenue of the classical optimal mechanism \citep{eden2016competition, esfandiari2019online, cai202199, beyhaghi2019optimal, bei2023bidder}, however, formal methods about how to  attract bidders are not a key point in these works. 
The discussion of optimal reserve price regarding the number of bidders in standard auctions can also be seen in works by \citep{engelbrecht1987optimal, levin1996optimal, wolfstetter1996auctions, menicucci2021basic, allouah2020prior}. Recent progresses also investigate how many more additional bidders need to join the auction in order to surpass the revenue of the classical optimal mechanism \citep{eden2016competition, esfandiari2019online, cai202199, beyhaghi2019optimal, bei2023bidder}, however, formal methods about how to  attract bidders are not a key point in these works. 
The literature has also attempted to utilize the knowledge of the exact number of buyers. For example, \citep{mcafee1987auctions} characterize optimal auctions when the seller possesses some prior knowledge about the number of buyers. \citep{levin2004auctions} investigate the seller's optimal response when the number of buyers is intentionally chosen from a set of ambiguities.

Researchers also devoted to relaxing the symmetric i.i.d. assumption \citep{krishna2009auction, fu2015randomization, koccyiugit2020distributionally}. 
Many works on optimal auction view the valuation distribution as non i.i.d., and 
 tries to find a suitable auction that generates as much revenue as possible. Many of them 
pursue to  find  near-optimal mechanisms that  are simple, practical, and robust \citep{hartline2013mechanism, roughgarden2016optimal, shah2019semi, castiglioni2023increasing}. 
In literature, there are also evidences from realistic Internet auctions showing that the higher the reserve price, the fewer the number of qualified bidders, thus reducing the success rate of auction \citep{bajari2003winner,lucking2007pennies,barrymore2009effect}. These evidences  support our idea that bring in more buyers and then increase the reserve price.

\subsection{Diffusion Auctions}

Our work belongs to the research direction of \textit{diffusion auction design} \citep{li2017mechanism, lee2017mechanisms, guo2021emerging, li2022diffusion, jeong2024groupwise}, an emerging area at the intersection of auction theory and network economics. In diffusion auctions, the auction mechanism not only incentivizes truthful bidding but also encourages each participant to actively diffuse auction information to their neighbors in a social network. The central idea is to stimulate information propagation so that buyers who are already aware of the auction help spread the word to others, thereby increasing participation and potential competition. As a result, the seller benefits from both improved revenue and enhanced overall welfare \citep{li2022diffusion}.

To ensure incentive compatibility in this dual setting of bidding and information diffusion, formal characterizations have been established \citep{li2020incentive}. A variety of diffusion auction mechanisms have since been proposed \citep{li2017mechanism, jeong2024groupwise, Lee2017MechanismsWR, Li2018CustomerSI, li2019graph, kawasaki2021mechanism, liu2021budget, guo2022combinatorial, shi2022social, li2024diffusion}. However, despite significant progress on mechanism design for diffusion and participation incentives, a general framework for revenue optimization in networked auctions remains largely undeveloped.

In the context of single-item diffusion auctions, it has been shown that no mechanism can simultaneously guarantee efficiency, incentive compatibility, and individual rationality without incurring potential losses in seller revenue. Although the VCG mechanism is naturally extensible to networks, it frequently results in deficits for the seller due to externality payments. Addressing these shortcomings, the foundational work by \citet{li2017mechanism} introduced the first DSIC information diffusion mechanism (IDM). Building on IDM, \citet{Lee2017MechanismsWR} proposed a multi-level mechanism (MLM) that improves efficiency and maintains the same revenue as IDM but sacrifices truthfulness. Further extensions include the customer sharing mechanism (CSM) \citep{Li2018CustomerSI}, which adapts IDM to sharing economies; the critical diffusion mechanism (CDM) \citep{li2019diffusion} for unweighted networks; and the weighted diffusion mechanism (WDM) \citep{li2019graph} for weighted graphs. Mechanisms such as FDM \citep{zhang2020redistribution} and NRM \citep{zhangECAIincentivize} focus on fairer reward distributions and improved social welfare without harming seller revenue.

\subsection{Revenue Maximization Problems in Diffusion Auctions}

Despite growing interest in diffusion auctions, the specific problem of \textit{revenue maximization} remains underexplored. Only a handful of works have explicitly addressed how to design diffusion mechanisms that improve seller revenue.

A recent attempt by \citet{bhattacharyya2023optimal} introduces the \texttt{maxViVa} mechanism, which transforms the network into a referral tree using diffusion timestamps and applies Myerson's virtual valuation function to maximize revenue. However, \texttt{maxViVa} fails to satisfy DSIC due to its reliance on timestamp-based referral paths and decreasing virtual valuations as the number of downstream buyers increases. This creates perverse incentives: direct neighbors of the seller with larger subtrees face higher reserve prices, discouraging propagation in dense subgraphs. Additionally, this mechanism assumes the seller has complete knowledge of diffusion timelines—an unrealistic requirement in practice.

Another work by \citet{zhang2023optimal} proposes the CWM mechanism, which applies a node-wise reserve price based on Myerson's virtual value function directly to each agent. However, this method lacks a rigorous justification for why the classical reserve price remains appropriate in networked diffusion settings. The reserve price for each agent is fixed under the assumption that buyer valuations are i.i.d., and no structural adjustment is made for their position in the network. Moreover, CWM assumes that losing bidders continue to propagate information truthfully, an assumption that introduces weak incentives for propagation and risks inefficiencies in participation.

Recent contributions also include the fixed-price diffusion mechanism studied by \citet{yu2024diffusion}, which provides a lower bound on the revenue and analyzes its average-case approximation. \citet{jeong2023groupwise} proposed the GPR mechanism based on groupwise externalities, which is efficient, weakly budget-balanced, and propagation strategy-proof. Similarly, the TRDM mechanism by \citet{zhang2023truthful} introduces sybil-proof properties to guard against sybil attacks and group collusion, achieving efficient allocations with high ex-post revenue.

A solid foundation for revenue-optimal payments under fixed allocations has been provided by \citet{li2020incentive}, and more recently, \citet{bhattacharyya2023optimal} explored broader payment characterizations for randomized diffusion auctions. Nonetheless, the vast majority of existing work focuses on maximizing social welfare, diffusion reach, or characterizing strategic equilibria, with only limited attention paid to optimizing the seller's expected revenue in general network settings.

We introduce a class of reserve price strategies that adapt to the underlying network structure—an element often overlooked in prior work. Our motivation stems from the observation that traditional reserve pricing assumes full participation, which is often violated in diffusion settings. For instance, a buyer with a high valuation might be excluded from the auction due to a lack of direct or indirect connections with the seller. Without incentives for intermediaries to propagate the auction information, the seller’s market becomes “thin,” leading to underperformance in both participation and revenue.

By tailoring reserve prices to account for network positions and propagation paths, our mechanism closes the gap between participation incentives and revenue optimization. This also offers a practical design strategy in settings where full social reach cannot be guaranteed. While the role of network structure in pricing has been studied in other domains such as operations, marketing, and economics \citep{alizamir2022impact}, it remains rare in the context of \textit{diffusion auction design}. Our work contributes to filling this gap.

\subsection{Our Method}
To the best of our knowledge, no existing work has systematically explored how to incorporate reserve prices into diffusion auction mechanisms while preserving critical properties such as individual rationality and incentive compatibility in both the bidding and diffusion dimensions. This omission marks a significant gap in the literature. In classical auction theory, the use of reserve prices is a standard and powerful tool for enhancing the seller’s expected revenue. However, in networked settings—where participation depends not only on valuation but also on information diffusion—naïvely applying reserve prices risks violating the incentive structures that underpin diffusion auction design.

Recognizing this, we propose the first general framework for integrating reserve price optimization into existing diffusion auction mechanisms. Our approach builds upon the Information Diffusion Mechanism (IDM) \citep{li2017mechanism}, which serves as a canonical model in diffusion auction design, and extends it within a Bayesian framework to optimize the seller's revenue. We analyze the expected revenue under Bayesian assumptions and derive the optimal global reserve price. While this theoretically optimal reserve price improves revenue, it undermines incentive compatibility—highlighting a core trade-off in the design of revenue-maximizing diffusion auctions.

To reconcile revenue optimization with strategic tractability, we further propose a practical mechanism—denoted as \texttt{APX-R}—that introduces a simple and explicit reserve pricing rule tailored for diffusion settings. Our design achieves approximately optimal revenue while maintaining the propagation and bidding incentives necessary for DSIC in many realistic cases. Importantly, we show that \texttt{APX-R} can extract more revenue than the classical Myerson benchmark in diffusion environments, where market access is constrained by network structure and participation is endogenous.

By situating our contribution within the broader diffusion auction literature and focusing on a key but underexplored aspect—revenue enhancement through reserve pricing—we offer a new and complementary direction to existing work. To facilitate comparison and contextualize our contribution, we summarize and contrast the properties of \texttt{APX-R} with those of existing single-item diffusion mechanisms in the table below.
\linespread{1.25}
{ 
	\begin{table}[htbp]\label{table_com}
		\centering
		\begin{tabular}{l|l|r}
			\toprule
			\textbf{Mechanisms} & \textbf{Revenue} & \textbf{Propagation Incentive} \\
		  \hline
            \texttt{APX-R} &  $\max\{v^*_{-1}, \gamma\}$ & strong incentive \\
            \hline
            IDM & $v^*_{-1}$ & strong incentive \\
            \hline
            \texttt{maxViVa} & $w_{T_1}^{-1}(w_{T_2})$ & no  propagation incentive \\ 
            \hline
            CWM & $\max\{v^{*}_{-w},rp\}$ & weak incentive \\
            \hline
            GPR & $v^*_{-1}$ & strong incentive \\
            \hline
            TRDM & $v^*_{-p_w}$ & strong incentive \\
        
			\bottomrule
		\end{tabular}%
		\label{tab:addlabel}%
         \caption{Revenue Comparisons}
	\end{table}%
}


Let $v^*_{(\cdot)}$ denote the highest bid among the corresponding set of reachable buyers. $1$ represents the closest critical diffusion buyer for the winner, and $\gamma$  is the reserve price defined in this article. In the \texttt{maxViVa} mechanism, the graph is transformed into a tree structure based on the timestamp at which each buyer propagates the sale info to others. The virtual valuation for each sub-tree rooted at $\ell$, a direct neighbor of the seller is defined as $w_{T_{\ell}}(v_{T_{\ell}})=v_{T_{\ell}}-\frac{1-F_{T_{\ell}}(v_{T_{\ell}})}{f_{T_{\ell}}}$ where $v_{T_{\ell}}$ is the valuation of buyers in sub-tree $T_{\ell}$, $F_{T_{\ell}}(\cdot)$ is the cdf of the highest valuation of the buyers in sub-tree $T_{\ell}$, and $f_{T_{\ell}}$ is the corresponding pdf. Let $T_1$ and $T_2$ represent the highest and second-highest virtual valuations, respectively, among the seller’s neighbors. The winner of each mechanism is denoted as $w$ while $p_w$ represents the shortest path from the seller to the winner. The notation  $-w$ refers to the set of buyers excluding the winner, while $-p_w$ represents the set of buyers excluding those in $p_w$. $rp$ is the reserve price defined in myerson optimal auction, the solution to  $v-\frac{1-F(v)}{f(v} = 0$.


The GPR mechanism is an efficient diffusion auction framework that allocates the item to the buyer with the highest bid among all buyers. It leverages groupwise externality to calculate the payment for both the winner and the critical buyers on the winner's path. While GPR maintains incentive compatibility in the dimension of propagation, it fails to do so in the bidding dimension, as losers on the winning path can strategically overbid to get higher rewards. Notably, GPR shares the same payment rule as IDM but differs in its allocation rule within single-item diffusion auctions. In our work, we propose a method to set an appropriate reserve price for IDM while preserving its incentive compatibility. This approach is also extended to GPR since the payment rule is identical to IDM in single-item networked auctions.


The TRDM mechanism allocates the item to the buyer with the highest bid, ensuring the efficiency of the auction process. Beyond this, TRDM guarantees incentive compatibility in diffusion dimensions and effectively prevents sybil attacks and collusion. Although IDM is not sybil-proof, sybil bids have no impact on the seller's revenue under IDM. In fact, the revenue generated by IDM is not lower than that of TRDM, while the total reward distributed under TRDM is not less than that of IDM. In contrast, for other non-sybil-proof mechanisms, such as CDM, FDM, and NRM, the introduction of sybil bids negatively impacts revenue, the more sybil bids are employed, the lower the resulting revenue for the seller.

\section{Problem Formulation}

Consider a network $G=(N\cup{s}, E)$, where $N\cup{s}$ constitutes the set of all agents. The seller $s$ has an item for sale, while agents $i \in N= \{1, \ldots, n\}$ are potential buyers within this network. An edge $(i, j)\in E$ signifies that $i$ is adjacent to $j$. Let $e_i$ denote the direct neighbors of agent $i$ with power set $\mathcal{P}(e_i)$, and the seller $s$ only maintains $\rho$ direct neighbours $e_s$.  It's important to highlight that this network could contain  a significantly larger number of agents compared to  classical optimal auction, i.e., $n \gg \rho$. 

In this network, a \textit{diffusion path} is a route through which the sale info travels from one agent to another. Agent $i$ can participate in the sale only if there exists a diffusion path from $s$ to $i$. Once  agent $i$ receives the sale info, she can place a bid and choose to diffuse the sale info to her neighbors. Classical auctions only focus on buyers who have an edge to the seller $s$, while a large portion of agents who are connected to the seller via a diffusion path and the information flow among these agents are omitted. Diffusion auctions aim to incentivize the buyers not only to truthfully bid, but also to  diffuse the information to their neighbors.

%


Every agent $i $ has a private type $t_i = (v_i, e_i)$, where $v_i$ is her true valuation. We assume that all $v_i$ are i.i.d. over interval $V_i=[0, \bar{v}]$, with cumulative distribution function (cdf) $F$ and   probability density function (pdf) $f$. 
Based on her true type $t_i$, agent $i$ takes action $t_i' = (v_i', e_i') \in T_i = V_i \times \mathcal{P}(e_i)$, where $v_i'$ is her bid and $e_i' \subseteq e_i$ is a subset of neighbors to whom she diffuses the sale info. 
Let $\mathbf{T}$ be the  space of type profiles, and let $\mathbf{t'} = (t_i', \mathbf{t}_{-i}')$ denote the action profile.


\subsection{Seller's Problem} 

A network auction mechanism $\mathcal{M} = (\pi, p)$ takes $\mathbf{t'} $ as input and determines  allocation $\pi=\{\pi_i\}_{i \in N}$ and  payment $p=\{p_i\}_{i \in N}$, where $\pi_i:\mathbf{T} \rightarrow\{0.1\}$, $ p_i:\mathbf{T} \rightarrow \mathcal{R}$.
While this definition is consistent with classical auctions, the set of potential bidders has been extended to include the entire network (i.e., the set $N$ and its underlying edges), as opposed to classical auctions that consider only $e_s$.

Under $\mathcal{M}=(\pi, p)$, given the action  profile $\mathbf{t'} $,  if the allocation rule satisfies $\sum_{i \in N} \pi_i\left(\mathbf{t}^{\prime}\right) \leq 1$  and $\pi_i\left(\mathbf{t}^{\prime}\right)=1$ when $t^{\prime}_{ i} \neq \text{nil}$, then this allocation rule is \textit{feasible}. We assume that the set of all feasible allocation rules is $\Pi$. The \textit{social welfare} is $\mathcal{W}(\pi, \mathbf{t'})=\sum_{i \in N} \pi_i(\mathbf{t'}) v_i$, where $\pi_i\left(\mathbf{t}^{\prime}\right)=1$ means  agent $i$ gets the item.  $\mathcal{M}$  is \textit{efficient} if its $\pi$  maximizes the social welfare. 



The utility for $i$ is $u_i(t_i, \mathbf{t}^{\prime},\pi,p) = \pi_i(\mathbf{t}^{\prime}) v_i - p_i(\mathbf{t}^{\prime})$. We focus on individually rational (IR) and dominant-strategy incentive-compatible (DSIC) auction mechanisms for the network setting. We aim to design new $\pi$ and $p$ with a reserve price $r$ that are IR and DSIC, and \textit{maximise the seller's revenue} and surpass the classical Myerson optimal revenue as much as possible.
\begin{definition}
Auction $\mathcal{M}$ is individually rational if $u_i\left(t_i,\left(\left(v_i, r_i^{\prime}\right), \mathbf{t}_{-i}^{\prime}\right)\right)=\pi_i\left(\mathbf{t}^{\prime}\right) v_i-p_i\left(\mathbf{t}^{\prime}\right) \geq 0$,  for all  $i \in N$, $r_i  \in \mathcal{P}\left(r_i\right)$ and $\mathbf{t}^{\prime}_{-i} \in T_{-i}$.
\end{definition}
\begin{definition}\label{definition_IC}
Auction $\mathcal{M}=(\pi, p)$ is dominant strategy incentive compatible if $u_i\left(t_i,\left(t_i, \mathbf{t}^{\prime}_{-i}\right)\right) \geq u_i\left(t_i,\left(t_i^{\prime}, \mathbf{t}^{\prime \prime}_{-i}\right)\right)$, for all $ i \in N$, $ t_i \in T_i$ and $ \mathbf{t}^{\prime}_{-i} \in T_{-i}$.
\end{definition}

The \textit{seller's revenue} is the total payment from all buyers: 
$
    R e v^{\mathcal{M}}\left(\mathbf{t}^{\prime}\right)=\sum_{i \in N} p_i\left(\mathbf{t}^{\prime}\right).
$
In network auctions, it's plausible that certain losing buyers who contribute significantly to diffusion get rewards, so it's possible that $p_i\left(\mathbf{t'}\right)<0$ for some $i$.
If $p_i\left(\mathbf{t}^{\prime}\right)\geq 0$, this means that the buyer $i$ should pay $p_i\left(\mathbf{t}^{\prime}\right)$ to the seller, while $p_i\left(\mathbf{t}^{\prime}\right)<0$ indicates that buyer $i$ should receive $\left|p_i\left(\mathbf{t}^{\prime}\right)\right|$ from the seller. $\mathcal{M}$ is \textit{weakly balanced} (WBB) if $R e v^{\mathcal{M}}\left(\mathbf{t}^{\prime}\right) \geq 0$ for any $\mathbf{t}^{\prime} $. 
$\mathcal{M}$ is \textit{optimal} if it maximises the seller's revenue. 
To compare the revenue of two different auctions, the general approach is to use \textit{average-case} or \textit{Bayesian analysis}.

 \section{Theoretical Benchmark}
To set a benchmark for our revenue-maximising network auction design, in this section we first characterize the upper-bound and lower-bound of the revenue optimization.

\subsection{Theoretical Upper-Bound}
The action profile $\mathbf{t'}$ forms a diffusion graph upon which the allocation, payments, and the seller's revenue are determined. The space of possible diffusion graphs is extremely large. For our benchmark, we target the most ideal network where the seller has omniscient knowledge, allowing her to directly call together all agents within the  network. That is, in this most ideal network for the seller, all buyers are direct neighbors of the seller, i.e., $e_s=N$. 

\begin{proposition} \label{proposition_OPT_benchmark}
       Consider an extremely ideal case  of the network, where all nodes are direct neighbors of the seller, i.e., $e_s=N$. Running the classical Myerson optimal auction within $N$ provides  the upper-bound for revenue.
{ \begin{equation}\label{eq_upperbound}
        \operatorname{OPT}(N) = \bar{v} - \hat{r}F^n(\hat{r})- \int_{\hat{r}}^{\bar{v}} \left[nF^{n-1}(x) - (n-1)F^n(x) \right] \, dx
    \end{equation}}
    where $\hat{r}$ is the solution to $x - (1-F(x))/f(x)=0$.
\end{proposition}

Specifically, when $F$ is the uniform distribution $U[0,\bar{v}]$, the reserve price set for buyers in Myerson optimal auction should be $\hat{r} = \bar{v}/2$, then $\operatorname{OPT}(N)$ can be calculated as:
\begin{equation}\label{eq_upperbound_uniform}
    \operatorname{OPT}(N)=\frac{\bar{v}(n-1)}{n+1}+\frac{\bar{v}}{n+1} \cdot\left(\frac{1}{2}\right)^n.
\end{equation}

Such a revenue is purely hypothetical, as it assumes the seller possesses omniscient information about all the potential buyers in the network.
Given any $e_s$, any $\bf t$ and any $\bf t'$, the seller can never attract more buyers than $N$, thus $\operatorname{OPT}(N)$ is the maximum possible revenue for \textit{any mechanism}  \textit{and any network}. This upper-bound serves as our benchmark, setting the highest standard for our mechanism design. 

\subsection{Lower-Bound}
Conversely, there is a worst-case scenario for the seller, where  the mechanism does not provide diffusion incentives and the seller can only call together buyers in $e_s$. This is  just the classical auction setting where it is a thin market. Let $\rho$ denote the size of $e_s$, and we utilize the revenue from the Myerson optimal auction for this classical scenario as our lower-bound. Use the same reasoning as for Proposition \ref{proposition_OPT_benchmark} but only consider the seller's direct neighbours in $e_s$, we have the revenue of the classical Myerson optimal auction as follows.
\begin{equation}\label{eq_bottomline}
    \operatorname{MYS}(e_s) = \bar{v} - \hat{r}F^{\rho}(\hat{r})- \int_{\hat{r}}^{\bar{v}} \left[\rho F^{\rho-1}(x) - (\rho-1)F^\rho(x) \right] \, dx
\end{equation}
Similarly, when $F$ is $U[0,\bar{v}]$, the lower-bound can be derived as:
{ {\begin{equation}
\operatorname{MYS}(e_s)   =\frac{\bar{v}(\rho-1)}{\rho+1}+\frac{\bar{v}}{\rho+1} \cdot\left(\frac{1}{2}\right)^\rho. 
\label{eq_bottomline_lowerbound}
\end{equation}}}

It is worth noting that, even this lower-bound is not a low standard for mechanism design. For notation simplicity, we didn't include $\mathbf{t}'$ in these benchmarks.
The goal of the upcoming content is to design IR,  DSIC network auction that always generates higher revenue than  $\operatorname{MYS}(e_s)$, while also approximating  $\operatorname{OPT}(N)$.

\section{Approximation Mechanism}
In this section, we design a network auction which has optimized revenue based on reserve price. Given its objective is to approximate the maximum possible revenue in any network market, i.e., eq.\eqref{eq_upperbound}, we call it the  approximation mechanism with reserve price (\texttt{APX-R}) and denote its expected revenue  as function $\operatorname{APX}(\cdot)$. 



\subsection{Data Structure}
With the action profile $\mathbf{t'}$, a \textit{diffusion graph} can be formed by aggregating all buyers' $e'_i$.  This graph contains all   diffusion paths.  If all diffusion paths from seller $s$ to $j$ must pass through   $i$, then  $i$ is called a \textit{diffusion critical node (DCN)} for $j$.
Downstream of $i$, we denote the group of nodes that share $i$ as a common DCN (including $i$) as $d_i$ and call it   \textit{diffusion downstream group (DDG)} of $i$. Denote $-d_i=N - d_i$, then $\mathbf{t'}_{-d_i}$ refers to the action profile of the remaining nodes excluding $d_i$. Upstream of $i$, we let $\left\{s, z_1, \ldots, z_k, \ldots,i\right\}$ denote the sequence of all DCNs for $i$ and  call it    \textit{diffusion critical sequence (DCS)} for $i$.

\begin{definition}\label{Def_POT}
The union of all agents' DCS forms a tree, which captures the partial ordering of  agents' importance in the diffusion. Call it the partial ordering tree (POT).
\end{definition}

POT ranks the buyers' \textit{topological importance} for the sale in networks. Note that POT is a partial ordering, thus only buyers on a same branch can be compared. 




\begin{definition}\label{Def_POST}
Let $T_i=$ $\left(d_i, E_i\right)$ be the Partial Ordering Sub-Tree (POST) rooted at node $i$, where $d_i$ is the set of sub-tree nodes, and let the set of edges be denoted as \\$E_i=\left\{(j, k) \mid j, k \in d_i \land (j, k) \in E\right\}$. Denote $k_i=|T_i|=|d_i|$ the size of  $T_i$.  
\end{definition}
POT is the partial ordering of all nodes, while POST is that on set $d_i$. 
The concept of DCS will play a central role in the design of \texttt{APX-R}, whereas POST will be utilized in the analysis of the approximation ratio.

\subsection{Mechanism \texttt{APX-R}}
Now we explain the workflow of \texttt{APX-R}. It is inherited from the simplest and most representative network auction, i.e., the information diffusion mechanism (\texttt{IDM}) \citep{li2017mechanism,li2022diffusion}. However, the verification of allocation now has the reserve price as an additional criterion.

\begin{algorithm}[tb]
\SetAlgoLined
\caption{\texttt{APX-R}: Approximation Mechanism with Reserve Price}
\label{alg_Mechanism_APX}
\KwIn{action profile $\mathbf{t'}$ and reserve price $r \in [0,1]$.}
\KwOut{allocation rule $\pi^{apx}(\mathbf{t'})$, payment rule $p^{apx}(\mathbf{t'})$.}

\textbf{Preprocessing:} construct diffusion graph and convert it into POT; \label{alg_POST}

\textbf{Initialization:$\forall i \in N, \pi_i^{apx}(\mathbf{t'})\leftarrow 0, p_i^{apx}(\mathbf{t'})\leftarrow 0 $}; \label{alg_Init}
Find  buyer $h$ with  highest bid $v'_h$\; \label{alg_HighestBidder}
\If {$v'_h< r$}  
{
    auction concludes without transaction; \label{alg_CheckFailure2}
}
\Else
{
    Initialize $h$ as winner: $w \leftarrow h$, and find $h$'s DCS $c_h$;\label{alg_FindPPS}
    
    \For {$j \in c_h \backslash\{h\}$} 
    {
        \If {$ v'_j \geq r$ \textbf{and} $v'_j=v^*_{-d_{j+1}}$}   
        {
            $w \leftarrow j$; 
            break; \label{alg_CriteriaEnd}
        }
    }
    
    $\pi_w^{apx}(\mathbf{t'}) \leftarrow 1$; \label{alg_DecideWinnder}
    
    $p_w^{apx}(\mathbf{t'})\leftarrow \max \{v^*_{-d_w}, r\}$;\label{alg_WinnerPayment}
    
    \For {$j \in c_w\backslash\{w\} $}
    {
        $p_j^{apx}(\mathbf{t'})\leftarrow \max \{v^*_{-d_j}, r\}-\max \{v^*_{-d_{j+1}}, r\}$;\label{alg_PPNsPayment}
    }
}
\end{algorithm}

In \texttt{APX-R},
lines \ref{alg_POST} to \ref{alg_DecideWinnder} is the allocation. 
Line \ref{alg_POST} constructs the diffusion graph according to the action profile, and then convert it into the corresponding POT. Lines \ref{alg_HighestBidder}  to \ref{alg_CheckFailure2} check whether there is a bid that meets the reserve price and if not, the sale fails. If the auction does not fail, then line \ref{alg_FindPPS} identifies the highest bidder's DCS. Denote the highest bid from group $d_i$ as $v^*_{-d_i}$. 
Along  this DCS, the algorithm will select the final winner. The intuition is that after removing all successors $d_{j+1}$ of $j$, if (1) $j$ stands as the highest bidder (i.e., $v'_j=v^*_{-d_{j+1}}$) and (2) her bid surpasses the reserve price $r$, then she wins. 
The algorithm then checks the above two criteria.
The algorithm only selects the first $j$ that satisfies both criteria as the winner, as in  line \ref{alg_DecideWinnder}. 
The  allocation rule classifies the agents in the network into the following partition.
After the allocation by Algorithm \ref{alg_Mechanism_APX}, buyers are divided into three categories: (i) the winner $w$; (ii) DCNs of the winner, i.e., $c_w-\{w\}$ ; (iii) other buyers, i.e., $N- c_w$.

From line \ref{alg_WinnerPayment} on is about payments. The payment to each buyer is based on the critical winning bid, that is, the minimum bid if a buyer wants to win. 
The payment of the winner $w$ is just her  critical winning bid $\max \{v^*_{-d_w}, r\}$. 
Any $ j $ along  the highest bidder's DCS is rewarded by line \ref{alg_PPNsPayment}, 
\begin{equation}\label{eq_PPNpayment}
     p_j^{apx}(\mathbf{t'})=\max \left\{v^*_{-d_j}, r\right\}-\max  \left\{v^*_{-d_{j+1}}, r\right\}. 
\end{equation}
This follows a similar form to the pioneering work of \texttt{IDM}. It is interpreted by the ``social externality'' such that the removal of some critical nodes can exclude other nodes from the market, so that an agent's externality comes not only from her bid but also from her social position.
The sharp difference with the original \texttt{IDM}, however, is that this externality in eq.\eqref{eq_PPNpayment} is now a difference between two maxima, where the reserve price $r$ is a function we don't yet know. 

\subsection{Basic Properties} 
Before analyzing the reserve price, we first prove that \texttt{APX-R} using any reserve price is WBB and IR. 
Denote ${\operatorname{APX}}({\mathbf t}';r)$ the seller's revenue under \texttt{APX-R} with reserve price $r$. Summing up the payments $p_j^{apx}$  of all  $j \in c_w$ gives us the seller's revenue. The subtrahend of $p_j^{apx}$ and the minuend of $p_{j+1}^{apx}$ are identical and offset each other. 



\begin{proposition}\label{proposition_WBB}
The mechanism \texttt{APX-R} with any reserve price is WBB with revenue ${\operatorname{APX}}({\mathbf t}';r)= \max \{v^*_{-d_1}, r\}$.
\end{proposition}

WBB is not obtained for free. The  work by \citep{li2017mechanism,li2022diffusion} proved that VCG  in networks is not WBB. In contrast,  the mechanism \texttt{APX-R} with any reserve price is WBB. It is also easily proved that \texttt{APX-R} is  IR.
\begin{proposition}\label{proposition_IR}
The mechanism \texttt{APX-R} with any reserve price is IR.
\end{proposition}

\section{Approximate Revenue Maximization}
A significant aspect   is how to set the reserve price so that \texttt{APX-R} is DSIC and has a good revenue. 
Now we will derive the  expected revenue of it and explore the condition that the reserve price must satisfy in order for the mechanism to be DSIC and yield high revenue. We start by analyzing the most common assumption that all buyers' valuations follow a uniform distribution 
$U[0,\bar{v}]$. In later sections, we extend our theoretical framework to other distribution models.

\subsection{Computation of Expected Revenue} Based on Definitions \ref{Def_POT} and \ref{Def_POST}, we can derive the expected revenue of \texttt{APX-R}  with any $r$, assuming that  buyers act truthfully. We will ensure the DSIC  in the following sections. This is a common procedure in DSIC mechanism design.
 
\begin{definition}\label{def_Tx}
If in the POT, there are $m$  neighbors  of $s$, then denote the POSTs directly branching from $s$ by set $\{T_1, \ldots, T_x, \ldots, T_m\}$, where the size of $T_x$ is $|T_x|=k_x$ and the highest bid in $T_x$ is $v_x$. We call $v_x$ the \textit{sub-tree value}.
\end{definition} 
The sub-tree  $T_x$ is essentially a sub-market in $G$, and the seller has $m$ sub-markets.
The cdf of $v_x$ is denoted as $H_x(v_x)=F^{k_x}(v_x)$, which is the probability that the valuations of all $k_x$ buyers in sub-tree $T_x$ are not greater than $v_x$. Denot its pdf as $h_x(v_x)$.

The condition for the winner to be from sub-tree $T_x$ is  that 
$v_x$ exceeds  all  $n-k_x$ bids from other sub-trees. 
Let    $Y$ be the highest bid from the sub-trees excluding $T_x$. Then the probability that the winner is from $T_x$ is $\operatorname{Prob}(Y \leq v_x)=F^{n-k_x}(v_x)$.  Then we can  get the cdf and pdf of $Y$, which are denoted as $G_x$ and $g_x$, respectively.
The  expected revenue  extracted from $T_x$ is calculated as:
{ {\begin{equation}\label{eq_APX_Tx}
  \operatorname{APX}_{T_x}(\mathbf{t'};r)= \int_r^{\bar{v}}\Big [r  G_x(r)+\int_r^{v_x} y g_x(y) d y \Big]  h_x(v_x) d v_x.
\end{equation}  }}  

We call it the \textit{sub-tree revenue of $T_x$}. There are two terms in the outer integral: the first term means that the sub-tree revenue is $r$ if the second highest sub-tree value $y$ is lower than $r$, which occurs  with probability $G_x(r)$. The second term is the expected revenue when $y$ meets $r$. 
According to eq.\eqref{eq_APX_Tx}, the sub-tree revenue from $T_x$ is computed as follows.
{ \small{
\begin{displaymath}
\begin{aligned}
&\operatorname{APX}_{T_x}(\mathbf{t'};r)= \int_r^{\bar{v}}\left [r  G_x(r)+\int_r^{v_x} y g_x(y) d y \right]  h_x(v_x) d v_x\\
& =rG_x(r)(1-H(r))+\int_r^{\bar{v}}\Big(\int_r^{v_x} y dG_x(y)\Big) h_x(v_x) d v_x\\
& =rF^{n-k_x}(r)-rF^n(r)+\int_r^{\bar{v}}\left[v_x F^{n-k_x}(v_x)-rF^{n-k_x}(r)-\int_r^{v_x} F^{n-k_x}(y) d y\right]dH_x(v_x) \\
& =\frac{k_x}{n}(\bar{v}-rF^n(r))-\frac{k_x}{n}\int_r^{\bar{v}}F^n(v_x)dv_x-k_x \int_r^{\bar{v}}\Big(\int_r^{v_x} F^{n-k_x}(y) d y\Big) \cdot F^{k_x-1}(v_x) \cdot f(v_x) d v_x.\\
\end{aligned}
\end{displaymath}}}
Exchanging the order of integration of the second integral, we get:
{ {
\begin{displaymath}
\begin{aligned}
& k_x \int_r^{\bar{v}}\Big(\int_r^{v_x} F^{n-k_x}(y) d y\Big) \cdot F^{k_x-1}(v_x) \cdot f(v_x) d v_x \\
& =\int_r^{\bar{v}}\Big(\int_y^{\bar{v}} k_x F^{k_x-1}(v_x) \cdot f(v_x) d v_x\Big) \cdot F^{n-k_x}(y) d y \\
& =\int_r^{\bar{v}}\Big(\int_y^{\bar{v}} d F^{k_x}(v_x)\Big) \cdot F^{n-k_x}(y) d y \\
&=\int_r^{\bar{v}}\big(1-F^{k_x}(y)\big) \cdot F^{n-k_x}(y) d y.
\end{aligned}
\end{displaymath}}}
By using the same variable $v_x$ for integration, we have:
{\small {
\begin{equation}
\begin{aligned}\label{eq_APX_Tx_r}
& \operatorname{APX}_{T_x}(\mathbf{t'};r)=\frac{k_x}{n}(\bar{v}-rF^n(r))- \int_r^{\bar{v}} \left[ \frac{k_x}{n}F^n(v_x)-F^n(v_x)+F^{n-k_x}(v_x)\right]dv_x.
\end{aligned}
\end{equation}}}
Thus, we can get  $\operatorname{APX}(\mathbf{t'};r)$ by summing up the expected revenue  from all $T_x$:
\begin{equation}\label{eq_APX_r}
\begin{aligned}
& \operatorname{APX}(\mathbf{t'};r) =\sum_{x=1}^m  \operatorname{APX}_{T_x}(\mathbf{t'};r)\\
& =\bar{v}-rF^n(r)-\sum_{x=1}^m \int_r^{\bar{v}}\left[F^{n-k_x}(v_x)-\frac{n-k_x}{n} F^n(v_x)\right] d v_x.
\end{aligned}
\end{equation}
If $F$ is $U[0,\bar{v}]$, the expected revenue of \texttt{APX-R}  extracted from $T_x$ is
\begin{equation}\label{eq_uniform_APX_Tx_r}
\operatorname{APX}_{T_x}(\mathbf{t'};r) = \frac{\bar{v}(1+k_x)}{n+1}\left[1-\left(\frac{r}{\bar{v}}\right)^{n+1}\right]-\bar{v}\cdot\frac{1-\left(\frac{r}{\bar{v}}\right)^{n-k_x+1}}{n-k_x+1}.
\end{equation}
Accordingly,  summing up this sub-tree revenue over all sub-trees, 
the \textit{seller's  expected revenue} from mechanism \texttt{APX-R} is
{ {\begin{equation}\label{eq_revenue_graph}
\begin{aligned}
 & \operatorname{APX}(\mathbf{t'};r)=\sum_{x=1}^m \operatorname{APX}_{T_x}(\mathbf{t'};r)\\&=\frac{\bar{v}(n+m)}{n+1}\left[1-\left(\frac{r}{\bar{v}}\right)^{n+1}\right]-\bar{v} \cdot \sum_{x=1}^m \frac{1-\left(\frac{r}{\bar{v}}\right)^{n-k_x+1}}{n-k_x+1}.
\end{aligned}
\end{equation} }}

Specifically, when we set $r=0$, the approximation mechanism degenerates into the Information Diffusion Mechanism (IDM) \citep{li2022diffusion}, which is the pioneering auction mechanism in social networks. Consequently, the expected revenue of the IDM mechanism can be computed as follows:
\begin{equation}
\operatorname{IDM}=\bar{v}-\sum_{x=1}^m \int_0^{\bar{v}}\left[F^{n-k_x}(v_x)-\frac{n-k_x}{n} F^n(v_x)\right] d v_x.
\end{equation}

\subsection{Impossibility of Optimal Reserve Price}\label{section_opt}

Building upon the derivation of the expected revenue of \texttt{APX-R} in the previous section, when buyers' valuations are independently and identically drawn from a regular distribution $F$, as given in eq.\eqref{eq_APX_r}, the first-order necessary condition for maximizing the seller's expected  with respect to the reserve price $r$ is:
\begin{equation}
    \frac{\partial \operatorname{APX}(\mathbf{t}';r)}{\partial r} =F^n(r) \sum_{x=1}^m\left[\frac{-k_x r f(r)}{F(r)}+\frac{1-F^{k_x}(r)}{F^{k_x}(r)}\right]=0
\end{equation}
After rearranging the equation, the optimal reserve price $r_{opt}$ should satisfy the following condition:
\begin{equation}\label{eq_reserve_opt}
    \sum_{x=1}^m \left[{F^{k_x}(r_{opt})}\right]^{-1}=m+n \cdot \frac{r_{opt} \cdot f(r_{opt})}{F(r_{opt})}
\end{equation}
This condition can also be expressed in the following form:
\begin{equation}\label{eq_hazard}
    \sum_{x=1}^m \frac{h_x(r_{opt})\left[r_{opt}-\frac{1-H_x(r_{opt})}{h_x(r_{opt})}\right]}{H_x(r_{opt})}=0
\end{equation} 
$H_x(\cdot)$ is the cdf of the highest bid in the subtree $T_x$, and $h_x(\cdot)$ is its corresponding pdf. Notably, the term $\frac{1-H_x(\cdot)}{h_x(\cdot)}$ represents the reciprocal of the hazard rate associated with the distribution $H_x(\cdot)$ \citep{myerson1981optimal} can be interpreted as the unavoidable revenue loss incurred by the seller due to the uncertainty in the highest valuation 
$v_x$ within the subtree $T_x$. The hazard rate of the value distribution $v_x$ is given by $\operatorname{W}_x(\cdot)=\frac{h_x(\cdot)}{1-H_x(\cdot)}$. Notice that $H_x(r_{opt})$ represents the probability that the bids of all buyers in the subtree $T_x$ are no higher than $r_{opt}$. $\operatorname{Z}_x(r_{opt})=\frac{1-H_x(r_{opt})}{H_x(r_{opt})}$ is essentially the \textit{odds ratio} in probability theory, representing the probability that there exists a bid higher than $r_{opt}$ in the subtree $T_x$. Thus, eq.\eqref{eq_hazard} can be reformulated as:
\begin{equation}\label{eq_hazard_rewritten}
    \sum_{x=1}^m \operatorname{W}_x(r_{opt})\operatorname{Z}_x(r_{opt})\left[r_{opt}-\frac{1}{\operatorname{W}_x(r_{opt})}\right]=0
\end{equation} 

Without loss of generality, we adopt a uniform distribution over $U[0, \bar{v}]$ for computation, which is a common assumption in auction theory. To maximize the seller's  revenue in eq.\eqref{eq_revenue_graph}, the first-order necessary condition with respect to the optimal reserve price $r_{opt}$ is
{ {\begin{equation}\label{eq_reserve_opt_uniform_distribution}
  \sum_{x=1}^m \left(\frac{\bar{v}}{r_{opt}}\right)^{k_x}=m+n.
\end{equation}}}
Eq.\eqref{eq_reserve_opt_uniform_distribution} is a polynomial function. Each term within the summation is monotonic with respect to $r_{opt}$. Consequently, under  uniform distribution, the solution for $r_{opt}$ is unique.

Unluckily, $r_{opt}$ depends on $m$, $n$, and each $k_x$ which is jointly determined by the diffusion actions of all buyers. This makes \texttt{APX-R} with $r_{opt}$ not DSIC.
\begin{theorem}[Non-Truthfulness of $r_{opt}$]\label{theorem_r_opt_hard_and_no_IC}
\texttt{APX-R} with its optimal reserve price $r_{opt}$ is not DSIC. 
\end{theorem}
\begin{proof}
We prove it by giving a counterexample.
$r_{opt}$ is dependent on $m$, $n$, and each $k_x$, but buyers' actions have the potential to alter these values, thus undermining the DSIC property. See the following example.
\begin{figure}[ht]
 \centering
\subfigure[social network]{
        \centering
        \includegraphics[height=3cm]{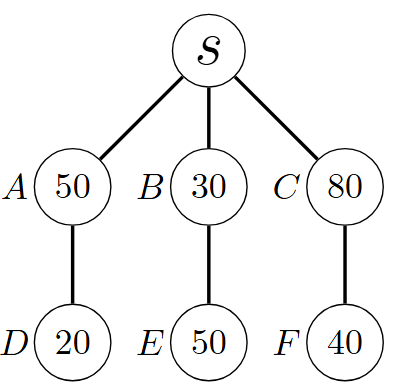}
        \label{fig_no_ic1}
}
\subfigure[truthful propagation]{
        \centering
        \includegraphics[height=3cm]{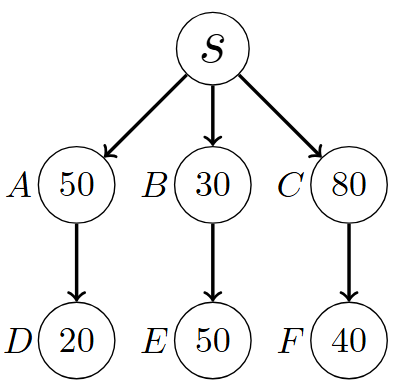}
        \label{fig_no_ic2}
}
\subfigure[no propagation of $C$]{
        \centering
        \includegraphics[height=3cm]{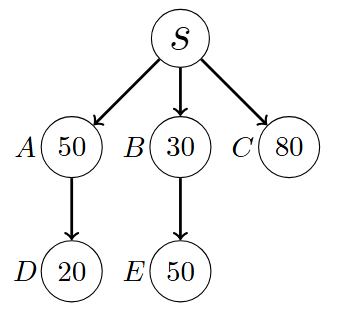}
        \label{fig_no_ic3}
}
\caption{DSIC fails under optimal reserve price.}
\label{counterexample_no_C}
\end{figure}

Suppose subfigure \ref{fig_no_ic2} is the corresponding POT when all buyers fully propagate sale info to all neighbors, while subfigure \ref{fig_no_ic3} is that when buyer C does not propagate. Buyers' $v_i$ are distributed in $U[0,1]$ and are shown in the circles.
In subfigure (a), $n=6, m=3, k_1=k_2=k_3=2$. By eq.\eqref{eq_reserve_opt_uniform_distribution}, the optimal reserve price is $r_{opt}=0.5773$. In subfigure (b), $n=5, m=3, k_1=k_2=2$ and $k_3=1$ which lead to  $r_{opt}=0.5664$. In both cases, the winner is C.  In this example, if C does not propagate, she would end up paying  less amount and having increased utility.   
This means \texttt{APX-R}  with the optimal reserve price $r_{opt}$ fails to be DSIC. 
\end{proof}

Guided by theorem \ref{theorem_r_opt_hard_and_no_IC}, we should abandon the use of the optimal reserve price for \texttt{APX-R}. A more practical approach is to relax the revenue optimization requirement and instead find a simple reserve price that approximate the optimal reserve price yet still guarantees DSIC.  We emphasize that there is a very large search space for the near-optimal $\gamma$.  Our goal is to find a simple but effective way. Therefore, we focus on reserve prices $r$
that are independent of the type of agent. Within such $r$, the DSIC property can be preserved, but the difficulty lies in revenue optimization over such reserve prices $r$.

\begin{proposition}\label{proposition_APX_DSIC}
The \texttt{APX-R} mechanism whose reserve price function 
$r$ is independent of both buyers' bidding strategy and propagation strategy is DSIC.
\end{proposition}
The proof technique is similar to that for the original information diffusion mechanism (\texttt{IDM}) \citep{li2022diffusion}. The process of our proof is divided into two steps, we keep one parameter in the action $(v_j', e_j')$ fixed while varying the other one. The first step is to demonstrate  when the propagation action $e'_j$ is fixed, the payment rule outlined in \texttt{APX-R} ensures that for any buyer $j$, submitting her true valuation $v_j$ maximizes their utility.
\begin{lemma}\label{lemma_IC_fixed_ei}
Under the mechanism \texttt{APX-R}, when applying type-independent reserve price function $r$, for any given propagation action, truthful bidding results in the highest utility for every buyer. 
\end{lemma}
 
\begin{proof}	 
Under true bidding, there are only three case of $j$'s categories. We prove this lemma holds for each category of $j$.

Case $1$: $j \notin c_w$. Under true bidding $v_j'=v_j$,  her payment is $0$. Now consider false bidding. If $j$ is in the winner's PDG, meaning $j \notin c_w$ and $ j \in d_w$, her utility is always $0$ regardless of her bid. If $j \notin c_w$ and $j \notin d_w$, overbidding with $v'_j > v_w$ will make her the winner with a payment of $v_w$, resulting in a utility of $v_j - v_w \leq 0$.

Case $2$: $j\in c_w \backslash \{w\}$. Under bidding $v_j'=v_j$, her utility is $\max \{v^*_{-d_{j+1}}, r\}-\max \{v^*_{-d_j}, r\}$ which is constrained by $v_j \leq \max \{v^*_{-d_{j+1}},r\}$. Overbidding could make her win, resulting in a utility of $v_j - \max \{v^*_{-d_j}, r\}$, which is less than her true bidding utility. Additionally, lower bidding does not change  utility. 

Case $3$: $j=w$. Her utility is $ v_j-\max \{v^*_{-d_j}, r\}$. If  overbids, she still wins with no utility change. If she lower bids and becomes a PPN of the winner, her utility becomes $\max \{v^*_{-\left(d_{j+1}\cup\{j\}\right)}, r\} - \max \{v^*_{-d_j}, r\}$, where the minuend represents the second-highest bid in $-d_{j+1}$, leading to a lower utility than  true bidding utility. If she lower    bids and becomes a non-PPN of the new winner, her utility becomes $0$.

For any buyer, given any  propagation action, bidding her true valuation  maximizes her utility.
\end{proof}

The second step parallels with Lemma \ref{lemma_IC_fixed_ei}. We analyze the utilities resulting from full propagation by buyers in each category. It shows that, for any buyer $j$,  while keeping her bid $v'_j$ fixed, full propagation maximizes her utility. The proof follows a similar way to that of Lemma \ref{lemma_IC_fixed_ei}. 
\begin{lemma}\label{lemma_IC_fixed_vi}
Under the approximation mechanism with type-independent reserve price function $r$, for any fixed bid, fully propagating the sale info to all neighbours results in the highest utility for every buyer. 
\end{lemma}
\begin{proof}	 
We employ a similar  approach as used in Lemma 1: examining the utility of each buyer class.

Case 1: When $j\notin c_w$ by full propagation, it holds that $j\in N-c_w$ even when she only propagate the sale info  to $e'_j \subseteq e_j$.

Case $2$: $j\in c_w \backslash {w}$. Then with  $e_j'=e_j$, her utility is a difference $\max\{ v^*_{-d_{j+1}}, r\}-\max \{v^*_{-d_j}, r\} $ which is constained by $v_j \leq \max\{ v^*_{-d_{j+1}}, r\}$. If she strategically propagates to $e'_j\subseteq e_j$ and fails to be a PPN of the winner, then her utility reduces to $0$. If under this strategic $e'_j$ she is still in a PPN of some winner $w$, her utility remains in the form of a difference $\max\{v^*_{-d_{(j+1)'}}, r\}-\max \{v^*_{-d_j}, r\} $. However, now $-d_{(j+1)'}$ is only a subset of $-d_{j+1}$, indicating that this new difference is less than the previous one. If this strategic $e'_j$ leads her to win, then her utility as the winner is $v_j-\max \{v^*_{-d_j}, r\}$, which is also less than the utility achieved through full propagation.

Case $3$: $j=w$. Her utility by full propagation is $v_j-\max \{v^*_{-d_j}, r\} $. According the allocation rule in \texttt{APX-R}, she is the highest bidder in $-d_{j+1}$ and ranks first in $c_h$. Since $e'_j$ can only change $d_{j+1}$, then with $e'_j \subseteq e_j$ she must still remain the first highest bidder in $-d_{j+1}$. Thus, her position remains unchanged as a winner, resulting in no impact on her utility.

In a conclusion, for any fixed bid, propagating the sale info to all neighbours results in the highest utility for every buyer. 
\end{proof}

Lastly, in the scenario where no buyer's valuation meets $r$, if buyer $j$ overbids and surpasses $r$, her utility becomes negative as she should pay $r$. Thus there's no advantage for her in misreporting her valuation. 
Truthfully bidding and fully propagating still optimize her utility in this scenario. Combining Lemmas \ref{lemma_IC_fixed_ei} and \ref{lemma_IC_fixed_vi} with the aforementioned analysis, we can conclude that \texttt{APX-R} with reserve price function $r$ which is independent of buyers' types is DSIC. Proposition \ref{proposition_APX_DSIC} is proved.

\subsection{Secure Revenue Constraint}
We   first guarantee that the  reserve price function $\gamma$ does not result in lower revenue for the seller compared to the scenario without any reserve price. That is, the secure revenue constraint is
\begin{equation}
    \operatorname{APX}(\mathbf{t'};\gamma) \geq \operatorname{APX}(\mathbf{t'};0).
\label{eq_rc_better_than_0}
\end{equation}
According to the definition of function $\operatorname{APX}(\mathbf{t'};r)$ in inequality \eqref{eq_revenue_graph},
the above constraint is equivalent to 
$
\sum_{x=1}^m \operatorname{APX}_{T_x}(\mathbf{t'};\gamma)-\operatorname{APX}_{T_x}(\mathbf{t'};0)\geq 0. 
$
This inequality holds if each subtraction within the summation is non-negative, that is,
$$ 
\forall T_x, \operatorname{APX}_{T_x}(\mathbf{t'};\gamma)-\operatorname{APX}_{T_x}(\mathbf{t'};0) \geq 0. 
$$
Using  eq.\eqref{eq_APX_Tx} under uniform distribution, 
the condition eq.\eqref{eq_rc_better_than_0}  holds via the following proposition.
\begin{proposition}[Sufficient Condition for Secure Revenue]\label{proposition_gamma_every_subtree}
 For any network, under \texttt{APX-R}, a sufficient condition for a secure revenue is
{ { \begin{equation}\label{eq_condition1_of_rc}
\forall T_x, \gamma_x\leq \bar{v}\cdot \sqrt[\leftroot{0}\uproot{10} k_x]{\frac{(n+1)}{(k_x+1)(n-k_x+1)}},
\end{equation}}}
where $\gamma_x$ is the  reserve price for $T_x$. 
\end{proposition}
As an interpretation, $n+1$ is the size of the entire networked market, including the seller and all bidders. $k_x+1$ is the size of the $x$-th submarket, specifically referring to  $T_x$ and the seller.  $n-k_x+1$ is the size of the remaining market after removing $T_x$ from the networked market.
Inequality \eqref{eq_condition1_of_rc} shows that different submarkets give different constraints on the secure revenue. 
Since there are $m$ subtrees in the whole partial order tree, the secure revenue requires a total of $m$ inequality constraints.

Given that the search space of reserve prices for the networked market could be very large, can we find a simple reserve price that still satisfies the inequality \eqref{eq_condition1_of_rc} and still achieves near-optimal revenue? In the following, we focus on a global reserve price $\gamma$ for all agents.
Denote the upper-bound of $\gamma_x$  by $\operatorname{\sup} \gamma_x$.
Using differential analysis shows that the monotonicity  of $\operatorname{\sup}\gamma_x$.
\begin{proposition}[Upper-Bounds of Reserve Prices]\label{proposition:monotone_upperbound_with_kx}
The upper-bound $\operatorname{\sup}\gamma_x$ of the reserve price set for each sub-market  is  increasing on $k_x$.
\end{proposition}
\begin{proof}
Recall the upper-bound of every $\gamma_x$ is $\operatorname{\sup}\gamma_x=\bar{v}\cdot \sqrt[\leftroot{0}\uproot{10} k_x]{\frac{(n+1)}{(k_x+1)(n-k_x+1)}}$.
To investigate the monotonicity of $\operatorname{\sup} \gamma_x$ with respect to $k_x$, we take the first-order partial derivative with respect to $k_x$ of it:
\begin{displaymath}
\frac{\partial{(\operatorname{\sup}\gamma_x)}}{\partial{k_x}}=\frac{\operatorname{\sup}\gamma_x}{k_x^2}\left[\frac{k_x(2k_x-n)}{(k_x+1)(n-k_x+1)}+\ln{\frac{(k_x+1)(n-k_x+1)}{n+1}}\right].
\end{displaymath}
We discuss whether $\frac{\partial{(\operatorname{\sup}\gamma_x)}}{\partial{k_x}}$ is non-negative in the following three cases:

Case 1: $k_x=1$,
{ \small{
\begin{displaymath}
\frac{\partial{(\operatorname{\sup}\gamma_x)}}{\partial{k_x}}=\gamma(1)\cdot\left(\ln{\frac{2n}{n+1}}-\frac{n-2}{2n}\right)=\gamma(1)\cdot\left(\ln{\frac{2n}{n+1}}-\frac{n+1}{2n}+\frac{3}{2n}\right)>0.
\end{displaymath} }}

Case 2: $1<k_x<n$, thus the following inequality holds:
{ {
 \begin{displaymath}
     \frac{(k_x+1)(n-k_x+1)}{n+1}\geq \frac{k_x+1}{k_x}.
 \end{displaymath}}}
 Moreover, due to that $\frac{(2k_x-n)}{(n-k_x+1)} >-1$,  we can get:
 { {
 \begin{displaymath}
 \frac{k_x(2k_x-n)}{(k_x+1)(n-k_x+1)} > -\frac{k_x}{k_x+1}.
 \end{displaymath}}}
 From these two inequalities above, we can derive:
 { {
\begin{displaymath}
\frac{k_x(2k_x-n)}{(k_x+1)(n-k_x+1)}+\ln{\frac{(k_x+1)(n-k_x+1)}{n+1}}> \ln{\frac{k_x+1}{k_x}}-\frac{k_x}{k_x+1},
\end{displaymath} }}
where the RHS is positive due to that the basic inequality in mathematics
{ {$\left(\ln{z}-\frac{1}{z}\right)> 0$}} holds for $\forall z> 1$.
Then each term of the product in the $\frac{\partial{(\operatorname{\sup}\gamma_x)}}{\partial{k_x}}$ is positive for $k_x \in (1,n)$, which indicates that $\frac{\partial{(\operatorname{\sup}\gamma_x)}}{\partial{k_x}}>0$ holds for $k_x \in (1,n)$.

Case $3$: $k_x=n$, $\frac{\partial{(\operatorname{\sup}\gamma_x)}}{\partial{k_x}}=0$ holds.

Taking into account these three cases above, we can conclude that
\begin{displaymath}
\operatorname{\sup}\gamma_x=\bar{v}\cdot \sqrt[\leftroot{0}\uproot{10}  k_x]{\frac{(n+1)}{(k_x+1)(n-k_x+1)}}
\end{displaymath}
is monotonically increasing on $k_x \in [1,n]$.   
\end{proof}


Now  based on this monotonicity, a sufficient condition for a global reserve price $\gamma$ to generate higher revenue than no reserve price is that, it should not surpass the sub-market reserve price upper-bounds $\operatorname{\sup}\gamma_x$ of every sub-market $T_x$.  Moreover, by Proposition \ref{proposition:monotone_upperbound_with_kx}, a smaller sub-tree identifies a smaller $\operatorname{\sup}\gamma_x$. Thus, our focus should be on the tight upper-bound induced by the smallest sub-tree.

Denote $\underline{k}$  the size of the minimum possible sub-tree in POT generating from the underlying network in historical auctions.  Note that it is a prior estimation based on historical auction data for this network and is independent from the current auction process, thus not affected by any agent's diffusion behavior. If this network has not conducted an auction before and the seller lacks sufficient prior knowledge about $\underline{k}$, then she can set the reserve price for the auction by making a conservative estimate of $\underline{k}$ (e.g., use Myerson reserve price by setting $\underline{k}=1$) and conduct multiple auctions to learn more about $\underline{k}$ accurately. In the following section of approximation analysis, we will show in eq.(\ref{eq_APX_SPA}) that even using $\underline{k} = 1$ is preferable to having no reserve price.
Due to DSIC, the network structure can be discovered, and the seller can eventually obtain the accurate $\underline{k}$ for a better revenue.

Substituting $\underline{k}$ for every $k_x$ in equation \eqref{eq_condition1_of_rc}, we get the global constraint for all $\gamma_x$. This is a tight upper-bound and is a refined sufficient condition for the secure revenue.



\begin{corollary}[Bounded Global Reserve Price]\label{corollary_secure_revenue}
A  reserve price function
 $\gamma$ for any agent can induce a higher revenue than no reserve price, if 
{ {\begin{equation}\label{eq_condition2_of_rc}
    \gamma \leq \bar{v} \cdot \sqrt[\leftroot{0}\uproot{10} \underline k]{\frac{(n+1)}{(\underline{k}+1)(n-\underline{k}+1)}} 
\end{equation}}}

\end{corollary}




Any reserve price function $\gamma$ should be bounded by the RHS of this inequality.  
Note that $n$, $\bar{v}$ and $\underline{k}$  are all known constants. 
The seller does not need to know any network strcuture in advance. 




\subsection{Near-Optimal Reserve Price}
Moreover, we want the  reserve price $\gamma$ to yield a  high revenue and approximate the theoretical upperbound $\operatorname{OPT}(N) $ as much as possible. To optimize it, revisit the  sub-tree revenue $\operatorname{APX}_{T_x}(\mathbf{t'};r)$ in eq.\eqref{eq_APX_Tx}.
For analytical convenience, we utilize the setting of uniform distribution of the valuations, but the method can be extended to other distributions. 
Taking the derivative of $\operatorname{APX}_{T_x}(\mathbf{t'};r)$  with respect to $r$ and setting it equal to zero leads to the optimal value of $r$, which leads to the the following result.
\begin{lemma}[Revenue Monotonicity]\label{lemma_monotonicity_max}
The submarket revenue $\operatorname{APX}_{T_x}(\mathbf{t'};r)$ is increasing on $r \in (0, \frac{\bar{v}}{\sqrt[k_x]{k_x+1}})$ and decreasing on $ r \in (\frac{\bar{v}}{\sqrt[k_x]{k_x+1}}, \bar{v})$ . It is maximized  at $r^*_x=\frac{\bar{v}}{\sqrt[k_x]{k_x+1}}$.
\end{lemma}

The optimal reserve price $r^*_x$ for each sub-tree $T_x$ might vary based on its size $k_x$. The monotonicity of $\operatorname{APX}_{T_x}(\mathbf{t'};r)$, as defined in Lemma \ref{lemma_monotonicity_max}, plays a crucial role in setting a near optimal global  reserve price $\gamma$.

\begin{theorem}[Near-Optimal Reserve Price $\gamma$]\label{theorem_reserve_price}
The reserve price function
\begin{equation}\label{eq_constant_reserve_price}
\gamma=\frac{\bar{v}}{\sqrt[\leftroot{0}\uproot{5} \underline{k}]{\underline{k}+1}}    
\end{equation}
for the mechanism \texttt{APX-R} is secure, and approaches the optimal reserve price from the left.
  
\end{theorem}
\begin{proof}
First, this $\gamma$ satisfies  inequality \eqref{eq_condition2_of_rc}. This is because the RHS of  inequality \eqref{eq_condition2_of_rc} can be expressed as $\sqrt[\leftroot{0}\uproot{4}\underline k]{\frac{n+1}{n-\underline{k}+1}} \cdot \frac{\bar{v}}{\sqrt[\leftroot{0}\uproot{4} \underline{k}]{\underline{k}+1}}$, which simplifies to $\sqrt[\leftroot{0}\uproot{4} \underline k]{\frac{n+1}{n-\underline{k}+1}}\cdot \gamma$, and the root term is no less than $1$. Second,  according to Lemma \ref{lemma_monotonicity_max}, the sub-tree optimal reserve price  $r^*_x=\frac{\bar{v}}{\sqrt[k_x]{k_x+1}}$ is increasing on $k_x$. Since for each sub-tree, $\underline{k}\leq k_x$, the revenue of all sub-trees increases  on $(0,\frac{\bar{v}}{\sqrt[\leftroot{0}\uproot{5} \underline{k}]{\underline{k}+1}})$. Therefore, $\gamma$ does not exceed the optimal reserve $r_{opt}$, but approaches $r_{opt}$ from the left. 
\end{proof} 

It is also easy to verify that the approximation mechanism \texttt{APX-R} with reserve price function $\gamma$ is IR and DSIC according to . We didn't make any assumption about the size or structure of the network, thus such a simple and explicit form for a DISC and approximately revenue-maximizing reserve price is convenient for application. 
If the seller knows some basic statistics of the network, such as the average degree or the minimum degree (which is usually easy to know), then she can set a more accurate $\underline{k}$. Even in cases where $\underline{k}$ is not available, the auctioneer can still conservatively underestimate its value.  
In the worst-case  for the seller, where no network can be exploited, Theorem \ref{theorem_reserve_price} further shows that $\gamma$ includes the classical Myerson optimal reserve price as a degenerate case.
\begin{corollary}[Relation to Myersonian Reserve Price]
When $n=\rho$, we have  $\gamma=\frac{\bar{v}}{\sqrt[\leftroot{0}\uproot{5} 1]{1+1}}=\frac{\bar{v}}{2}$. That is, for classical auction  with no network,  \texttt{APX-R} degenerates to Vickrey auction and  $\gamma$ degenerates to the Myerson optimal reserve price. 
\end{corollary} 

The above procedure for computing the reserve price in a diffusion auction can also be adapted to other value distributions.
The workflow for computing the reserve price is:
(1). Utilize DSIC to construct the diffusion graph and generate the subtrees $T_x$. (2). Calculate the seller's subtree revenue and total revenue, which are functions of the sizes of all $T_x$ (eqs. (\ref{eq_uniform_APX_Tx_r})-(\ref{eq_revenue_graph})). (3). Decompose the secure revenue constraint into multiple sub-constraints for subtree reserve prices $\gamma_x$ (eq.(\ref{eq_condition1_of_rc})). (4). Find the monotonicity of subtree revenues concerning $\gamma_x$ by taking derivatives. Identify the overlapping interval of different $\gamma_x$ where all subtree revenues are increasing. (5). Set the right boundary of this overlapping interval as the constant global reserve price $\gamma$.
For other i.i.d. distribution functions. Steps 1, 3, and 5 still work in the same way. Regarding step 2, eqs. (\ref{eq_uniform_APX_Tx_r})-(\ref{eq_revenue_graph}) still can   be  computed. The challenge arises in step 4, as complex distributions may lead to multiple ''overlapping intervals'' for different $\gamma_x$. Thus other techniques cuould be needed to find the global $\gamma$.

\section{Approximation Guarantees}
We will finally demonstrate the extent to which \texttt{APX-R}, utilizing the reserve price provided in Theorem \ref{theorem_reserve_price}, is effective in extracting revenue from  networks. Recall that regarding the approximation benchmark, we have an upper-bound and a lower-bound.
The theoretical upper-bound $\operatorname{OPT}(N)$ is the maximum possible revenue from any network while the lower-bound $\operatorname{MYS}(e_s)$ is the revenue of classical Myerson optimal auction. We will show that $\operatorname{APX}(\mathbf{t'};\gamma) $ always surpasses $\operatorname{MYS}(e_s)$ and approximates $\operatorname{OPT}(N)$.

\subsection{Compare with Myerson's Revenue}
According to the Myerson optimal auction, $\operatorname{MYS}(e_s)$ is obtained by using a reserve price of $\frac{\bar v}{2}$ in an auction with buyers $e_s$.
To establish a connection between $\operatorname{APX}(\mathbf{t'};\gamma)$ and $\operatorname{MYS}(e_s)$, we introduce the intermediary $\operatorname{APX}(\mathbf{t'};\frac{\bar v}{2})$. It is the revenue of our \texttt{APX-R} when naively using the optimal reserve price for the classical $\rho$ bidder market, which is $\frac{\bar v}{2}$.
Rewriting $\frac{\bar v}{2}$ and applying  Theorem \ref{theorem_reserve_price} shows
\begin{displaymath}
0<\frac{\bar{v}}{2}=\frac{\bar{v}}{\sqrt[1]{1+1}}\leq \frac{\bar{v}}{\sqrt[\underline{k}]{\underline{k}+1}}=\gamma\leq \frac{\bar{v}}{\sqrt[k_x]{k_x+1}} .
\end{displaymath}
By Lemma \ref{lemma_monotonicity_max},  $\operatorname{APX}(\cdot)$ should be increasing on interval $ (0, \frac{\bar{v}}{\sqrt[k_x]{k_x+1}})$, this results in the comparison that
\begin{equation}\label{eq_APX_SPA}
\operatorname{APX}(\mathbf{t'};\gamma)\geq \operatorname{APX}(\mathbf{t'};\frac{\bar{v}}{2}).  
\end{equation}


The value of the RHS can be easily computed by eqs. \eqref{eq_APX_Tx} and \eqref{eq_revenue_graph}, and the value of $\operatorname{MYS}(e_s)$ is already  in eq.\eqref{eq_bottomline}. By applying basic scaling techniques (see Appendix), we have
\begin{equation}\label{eq_APX_MYS}
 \operatorname{APX}(\mathbf{t'};\frac{\bar{v}}{2})\geq \operatorname{MYS}(e_s).   
\end{equation}
Eqs. \eqref{eq_APX_SPA} and \eqref{eq_APX_MYS} together lead to Theorem \ref{proposition_better_tha_Myerson}, and the intermediate result in eq.\eqref{eq_APX_SPA} gives Corollary \ref{corollary_APX_SPA}.


\begin{theorem}\label{proposition_better_tha_Myerson}
The revenue of mechanism \texttt{APX-R} with reserve price function in eq.\eqref{eq_constant_reserve_price} surpasses that of classical Myerson optimal auction. 
\end{theorem}
\begin{corollary}\label{corollary_APX_SPA}
The mechanism \texttt{APX-R} using reserve price function in eq.\eqref{eq_constant_reserve_price}  generates larger revenue than using classical optimal reserve price (i.e., that for SPA).
\end{corollary}
\begin{proof}
To prove Theorem \ref{proposition_better_tha_Myerson}, we introduce an intermediary fabricated function $\operatorname{APX}(\mathbf{t'};\frac{\bar v}{2})$, which represents the revenue of our mechanism \texttt{APX-R} when employing the classical optimal reserve price. By eq.\eqref{eq_APX_SPA}, we know $\operatorname{APX}(\mathbf{t'};\gamma)$ is lower-bounded by  $\operatorname{APX}(\mathbf{t'};\frac{\bar{v}}{2})$.

Therefore, to prove Theorem \ref{proposition_better_tha_Myerson}, we only need to show that $\frac{\operatorname{APX}(\mathbf{t'};\frac{\bar{v}}{2})}{\operatorname{MYS}(e_s)}$ is lower-bounded by a value that isn't less than $1$.
When $\rho=1$, since $n \geq$ $\rho=1$, it is trivial that 
{ {\begin{displaymath}
\operatorname{APX}(\mathbf{t'};\gamma) \geq \operatorname{APX}(\mathbf{t'};\frac{\bar{v}}{2})=\frac{\bar{v}}{2}-\frac{\bar{v}}{2^{n+1}} \geq \frac{ \bar{v}}{4} = \operatorname{MYS}(e_s).
\end{displaymath}}}
When $\rho \geq 2$, $\operatorname{APX}(\mathbf{t'};\frac{\bar{v}}{2})$ can be derived as following:
{ \small{
\begin{equation}
\begin{aligned}
\operatorname{APX}(\mathbf{t'};\frac{\bar{v}}{2})
&=\bar{v}\left[\frac{n+m}{n+1}-\sum_{x=1}^m \frac{1}{n-k_x+1}+\sum_{x=1}^m\left(\frac{(\frac{1}{2})^{n-k_x+1}}{n-k_x+1}-\frac{(k_x+1)(\frac{1}{2})^{n+1}}{n+1}\right)\right] \\
& \geq \bar{v}\left[\frac{n+m}{n+1}-\sum_{x=1}^m \frac{1}{n-k_x+1}+\sum_{x=1}^m\left(\frac{\left(\frac{1}{2}\right)^{n-1+1}}{n-1+1}-\frac{(1+1)\left(\frac{1}{2}\right)^{n+1}}{n+1}\right)\right]\\
& =\bar{v}\left[\frac{n+m}{n+1}-\sum_{x=1}^m \frac{1}{n-k_x+1}+\frac{m \cdot\left(\frac{1}{2}\right)^n}{n(n+1)}\right]\\
& \geq \bar{v}\left[ \frac{n+m}{n+1}-\frac{m-1}{n}-\frac{1}{m}+\frac{m \cdot\left(\frac{1}{2}\right)^n}{n(n+1)}\right]\\
&=\bar{v}\left[1+\frac{m-1}{n+1}-\frac{m-1}{n}-\frac{1}{m}+\frac{m \cdot\left(\frac{1}{2}\right)^n}{n(n+1)}\right]
\end{aligned}
\end{equation}}}
The first inequality holds because { {$\left[\frac{\left(\frac{1}{2}\right)^{n-k_x+1}}{n-k_x+1}-\frac{(k_x+1)\left(\frac{1}{2}\right)^{n+1}}{n+1}\right]$}}
is monotonically increasing on $k_x$. The second inequality arises from splitting the summation and bounding it by its maximum value. This maximum value is is $\frac{m-1}{n}+\frac{1}{m}$, when there are $(m-1)$ sub-trees with size of $k_x=1$ and one sub-tree with size of $k_x=n-(m-1)$.
Denote the function with respect to $n$ as $\epsilon(n)$: \begin{displaymath}
\epsilon(n)= 1+\frac{m-1}{n+1}-\frac{m-1}{n}-\frac{1}{m}+\frac{m \cdot\big(\frac{1}{2}\big)^n}{n(n+1)}.
\end{displaymath}
By calculating the first-order derivative of this function, we have
{ {
\begin{displaymath}
\begin{aligned}
    \epsilon'(n) &=\frac{m\left((2n+1)\left(1-\left(\frac{1}{2}\right)^n\right)\right)-\ln{2}\cdot \frac{mn(n+1)}{2^n}-2n-1}{n^2(n+1)^2}\\
    & \geq \frac{2\left((2n+1)\left(1-\left(\frac{1}{2}\right)^n\right)-1.5\ln{2}\right)-2n-1}{n^2(n+1)^2}\\
    & \geq \frac{2(0.75(2n+1)-1.5\ln{2})-2n-1}{n^2(n+1)^2}\\
    & = \frac{(n+0.5)-3\ln{2}}{n^2(n+1)^2}>0.
\end{aligned}
\end{displaymath}}}
The first inequality holds becuase $\frac{n(n+1)}{2^n} \leq \frac{2 \times 3}{2^2}=1.5$ when $\rho \geq 2$.
Therefore, $\epsilon(n)$ is monotonically increasing on $n \in [2,+\infty]$. Based on the monotonicity of $\epsilon(n)$, we can further derive that

\begin{displaymath}
\begin{aligned}
\frac{\operatorname{APX}(\mathbf{t'};\frac{\bar{v}}{2})}{\operatorname{MYS}(e_s)} & \geq \frac{1+\frac{m-1}{n+1}-\frac{m-1}{n}-\frac{1}{m}+\frac{m \cdot\left(\frac{1}{2}\right)^n}{n(n+1)}}{\frac{\rho-1}{\rho+1}+\frac{1}{\rho+1} \cdot\left(\frac{1}{2}\right)^\rho} \\
& =\frac{\epsilon(n)}{\frac{\rho-1}{\rho+1}+\frac{1}{\rho+1} \cdot\big(\frac{1}{2}\big)^\rho} \\
& \geq \frac{\epsilon(m)}{\frac{m-1}{m+1}+\frac{1}{m+1} \cdot\big(\frac{1}{2}\big)^m}\\
& =\frac{1+\frac{m-1}{m+1}-\frac{m-1}{m}-\frac{1}{m}+\frac{m \cdot\left(\frac{1}{2}\right)^m}{m(m+1)}}{\frac{\rho-1}{\rho+1}+\frac{1}{\rho+1} \cdot\left(\frac{1}{2}\right)^\rho} \\
& =\frac{\frac{m-1}{m+1}+\frac{1}{m+1} \cdot\left(\frac{1}{2}\right)^m}{\frac{\rho-1}{\rho+1}+\frac{1}{\rho+1} \cdot\left(\frac{1}{2}\right)^\rho} \\
& \geq 1.
\end{aligned}
\end{displaymath}
The last inequality holds because $ \bar{v}\cdot \left[\frac{m-1}{m+1}+\frac{1}{m+1} \cdot\left(\frac{1}{2}\right)^m\right]$ is the expected revenue of Myerson optimal auction among $m$ buyers, while $\bar{v}\cdot \left[\frac{\rho-1}{\rho+1}+\frac{1}{\rho+1} \cdot\left(\frac{1}{2}\right)^\rho \right]$ is the expected revenue of Myerson optimal auction among $\rho$ buyers. According to the the process of converting the propagation auction graph into the corresponding POT, the direct neighbors of $s$ in POT are no less than those in the propagation auction graph, i.e., $m\geq \rho$. Thus
{ {
\begin{displaymath}
\bar{v}\cdot \left[\frac{m-1}{m+1}+\frac{1}{m+1} \cdot\Big(\frac{1}{2}\Big)^m\right] \geq \bar{v}\cdot \left[\frac{\rho-1}{\rho+1}+\frac{1}{\rho+1} \cdot\Big(\frac{1}{2}\Big)^\rho \right].
\end{displaymath}}}
Since the expected revenue of Myerson optimal auction increases as more buyers participate.
Taking into account these two cases, we can conclude that 
\begin{displaymath}
\operatorname{APX}(\mathbf{t'};\gamma)\geq \operatorname{APX}(\mathbf{t'};\frac{\bar{v}}{2})\geq \operatorname{MYS}(e_s).
\end{displaymath}

Note that $\operatorname{APX}(\mathbf{t'};\frac{\bar{v}}{2})>\operatorname{MYS}(e_s)$ for any $n>\rho$.
That is to say, for a sale in any economic network that has more buyers than the classical auction without propagation,

\begin{displaymath}
\operatorname{APX}(\mathbf{t'};\gamma)\geq \operatorname{APX}(\mathbf{t'};\frac{\bar{v}}{2})> \operatorname{MYS}(e_s).
\end{displaymath}
In conclusion, in terms of expected revenue, running \texttt{APX-R} with reserve price $\gamma=\frac{\bar{v}}{\sqrt[\underline{k}]{\underline{k}+1}}$ always outperforms running classical Myerson optimal auction in the set of $s$'s direct neighbors without any propagation. This completes the formal proof of Theorem  \ref{proposition_better_tha_Myerson} and Corollary \ref{corollary_APX_SPA}.    
\end{proof}



\subsection{Compare with Maximum Possible Revenue}
By Proposition \ref{proposition_OPT_benchmark}, 
the theoretical upper-bound $\operatorname{OPT}(N)$ is the highest standard for evaluating the revenue of auction in networks. In fact, in the context of symmetric i.i.d. valuations, it is the highest attainable revenue across any network through any auction mechanism. 
We will show that $\operatorname{APX}(\mathbf{t'};\gamma)$ provides a very close approximation ratio to $\operatorname{OPT}(N)$.

\begin{theorem}\label{theorem_ratio_optn}
The approximation ratio of $\operatorname{APX}(\mathbf{t'};\gamma)$ to $\operatorname{OPT}(N)$ is lower-bounded by $1-\frac{1}{\rho\underline{k}-\underline{k}+1}.
$
\end{theorem}

\begin{proof}
Recall  $\rho=|e_s|$ and $\underline{k}=\min \{k_1,\dots, k_x, \dots, k_m\}$.
Let $z=n+1$ be the number of all agents including the seller in the graph, then $z-k_x$ is the number of all agents in the graph excluding agents in sub-tree $T_x$.
With eq.\eqref{eq_revenue_graph}, $\operatorname{APX}(\mathbf{t'};\gamma)$ is expressed as 
{ {\begin{displaymath} 
\begin{aligned}\label{eq_APX_in_approximation_ratio}
\operatorname{APX}(\mathbf{t'};\gamma)&= \frac{n\bar{v}}{z} -\bar{v} \cdot \sum_{x=1}^m \frac{k_x}{z(z-k_x)}  + \bar{v} \cdot \sum_{x=1}^m \left[\frac{\big(\frac{\gamma}{\bar{v}}\big)^{z-k_x}}{z-k_x}-\frac{\left(k_x+1\right)\big(\frac{\gamma}{\bar{v}}\big)^{z}}{z}\right].\\
\end{aligned}
\end{displaymath} }}By taking derivative of the third term, it's found that it monotonically increases on $k_x$ and on $ {\gamma \in \Big[0, \frac{\bar{v}}{\sqrt[\underline{k}]{\underline{k}+1}}\Big]}$.
Given that $\gamma\geq \frac{\bar{v}}{2}$ and $k_x \geq 1$, we can deduce the following inequality:
{ {\begin{displaymath} 
\begin{aligned}\frac{\big(\frac{\gamma}{\bar{v}}\big)^{z-k_x}}{z-k_x}-\frac{\left(k_x+1\right)\big(\frac{\gamma}{\bar{v}}\big)^{z}}{z}  &\geq \frac{\big(\frac{1}{2}\big)^{z-1}}{z-1}-\frac{(1+1)\big(\frac{1}{2}\big)^{z}}{z}.
\end{aligned}
\end{displaymath} }}Replacing the third term of $\operatorname{APX}(\mathbf{t'};\gamma)$ by the RHS,  the approximation ratio is guaranteed as 
\begin{displaymath} 
\begin{aligned}
    &\frac{\operatorname{APX}(\mathbf{t'};\gamma)}{\operatorname{OPT}(N)}\geq \frac{n-\sum_{x=1}^m \frac{k_x}{z-k_x} +\big(\frac{1}{2}\big)^{n} \cdot \frac{m}{n}}{n-1+\big(\frac{1}{2}\big)^n}.
\end{aligned}
\end{displaymath} 
The summation $\sum_{x=1}^m \frac{k_x}{z-k_x} $ reaches its maximum when there are $m=\rho$ sub-trees, out of which $(\rho-1)$ sub-trees have a size of $k_x=\underline{k}$ and one sub-tree has a size of $k_x=n-(\rho-1)\underline{k}$. The  details are in Appendix. Consequently, the approximation ratio is lower-bounded as
\begin{align}\label{eq_APX_OPT_ratio_bounding}
\nonumber
& \frac{\operatorname{APX}(\mathbf{t'};\gamma)}{\operatorname{OPT}(N)}\geq \frac{n-\sum_{x=1}^m \frac{k_x}{n-k_x+1}+\big(\frac{1}{2}\big)^n \cdot \frac{\rho}{n}}{n-1+\big(\frac{1}{2}\big)^n}\\ \nonumber
& \geq \frac{n-\big(\frac{1}{2}\big)^n-\frac{(\rho-1)\underline{k}}{n-\underline{k}+1}-\frac{n-\rho\underline{k}+\underline{k}}{\rho\underline{k}-\underline{k}+1}+\big(\frac{1}{2}\big)^n \cdot\Big(1+\frac{\rho}{n}\Big)}{n-1+\big(\frac{1}{2}\big)^n} \\
\nonumber
& > \frac{(\rho-1)\underline{k}\big(\frac{n+1}{\rho\underline{k}-\underline{k}+1}-\frac{1}{n-\underline{k}+1}\big)-\big(\frac{1}{2}\big)^n}{n-1}.
\end{align}
The first inequality is just replacing $m$ with its minimum value $\rho$. The second inequality is by splitting the summation and bounding it by its maximum value.
The third inequality is based on the fact that if $a<c$ and $b>d$ then $\frac{a+b}{c+d} > \frac{a}{c}$.

Let $\Phi(n,\rho,k)$ denote the last line, i.e., the performance guarantee.
Give a limit analysis of $\Phi(n,\rho,k)$ as follows. 
\begin{displaymath}
\begin{aligned}
&\lim\limits_{n \to +\infty}\frac{(\rho-1)\underline{k}\big(\frac{n+1}{\rho\underline{k}-\underline{k}+1}-\frac{1}{n-\underline{k}+1}\big)-\big(\frac{1}{2}\big)^n}{n-1}\\
&=\frac{(\rho-1)\underline{k}}{\rho\underline{k}-\underline{k}+1}\lim\limits_{n \to +\infty}\frac{\big(n+1-\frac{\rho\underline{k}-\underline{k}+1}{n-\underline{k}+1}-\frac{\rho\underline{k}-\underline{k}+1}{2^n(\rho-1)\underline{k}}\big)}{n-1}\\
&=\frac{(\rho-1)\underline{k}}{\rho\underline{k}-\underline{k}+1}\\&=1-\frac{1}{\rho\underline{k}-\underline{k}+1},
\end{aligned}
\end{displaymath} 
which is equivalent to $\lim\limits_{n \to +\infty}\Phi(n,\rho,k)=1-\frac{1}{\rho\underline{k}-\underline{k}+1}$.
Since $\Phi(n,\rho,k)$ deceases on $n$, this finally leads to that the approximation ratio is lower-bounded by  $1-\frac{1}{\rho\underline{k}-\underline{k}+1}$.
\end{proof}



The approximation ratio increases with both $\rho$ and $\underline{k}$, and these two parameters are uncorrelated. 
Therefore, in networks where $\rho$ and $\underline{k}$ are sufficiently large, the revenue collected by \texttt{APX-R} with $\gamma$ could closely approach the theoretical upper bound. 
We can further simplify this ratio by investigating the worst-case scenario where some of the seller's direct neighbors do not have any social contact (i.e., $\underline{k}=1$).

\begin{corollary}\label{corollary_ratio}
The approximation ratio of $\operatorname{APX}(\mathbf{t'};\gamma)$ to $\operatorname{OPT}(N)$ is  is lower-bounded by $1-{1\over\rho}$.  
\end{corollary}
From this ratio, we can see that \texttt{APX-R} in combination with the reserve price in Theorem \ref{theorem_reserve_price} is capable of extracting very high revenue. It works even better for sellers who originally possess  a substantial  audience on their own. In this case, it will surpass Myerson optimal auction even more. 


Finally, under the symmetric i.i.d. assumption, we list the revenue extracted by all the aforementioned mechanisms in the context of selling an item within networks.
\begin{theorem}
For a sale in any network that has   more agents than a classical auction (i.e., $n\neq \rho$), the different revenues extracted from the network follow that
 \begin{displaymath}
   \operatorname{MYS}(e_s) 
 <
 \operatorname{APX}(\mathbf{t'};\frac{\bar{v}}{2}) 
 \leq  
 \operatorname{APX}(\mathbf{t'};\gamma)   
 \lessapprox  
 \operatorname{OPT}(N)
\end{displaymath}   
\end{theorem}

%


%


\section{Extentions to Other Regular Distributions}
In the previous section, the approximation mechanism \texttt{APX-R} was introduced, and under the assumption that all buyers' valuations are drawn form uniform distribution independently, a simple yet effective reserve price function was designed. This function ensures the incentive compatibility of the mechanism, compelling buyers to bid truthfully based on their true valuations while also incentivizing them to propagate the auction info to all of their neighbors, thereby expanding the market to its fullest potential. Through rigorous analysis, it was demonstrated that under this setting, the \texttt{APX-R} mechanism optimizes the seller’s expected revenue and showcases its theoretical superiority. However, real-world auction scenarios often deviate significantly from the idealized uniform distribution. Buyers’ valuations are influenced by various factors and may follow diverse probability distributions. Consequently, a key challenge is designing a reserve price function that accommodates a broader class of distribution types while preserving the incentive compatibility of the mechanism. This section delves into the applicability and performance of the \texttt{APX-R} mechanism under valuations of different regular distributions. Additionally, it explores how to refine the reserve price function to optimize the seller’s expected revenue while maintaining the mechanism’s incentive compatibility.  A regular distribution refers to any valuation distribution 
$F$ that ensures the Myerson virtual valuation function $\psi(v)=v-\frac{1-F(v)}{f(v)}$ is monotonically increasing with respect to valuation $v$ in traditional auctions. This section is based on the assumption that buyers' valuations are independently and identically distributed drawn from a regular distribution $F$.

\subsection{Derivation of Optimal Reserve Prices}

\subsubsection{Subtree-wise Optimal Reserve Prices}
Notice that if a subtree-wise reserve price is set, meaning that after all buyers report their types $\mathbf{t'}=(\mathbf{v'},\mathbf{e'})$ and the auction network is transformed into the POT tree, the same reserve price is applied to all buyers within each POST. The expected revenue from each subtree $T_x$ is given by eq.\eqref{eq_APX_Tx_r}, differentiating it with respect to $r$ and setting it to zero gives the optimal subtree-wise reserve price for subtree $T_x$.
\begin{equation}\label{eq_optimal_r_Tx}
    \frac{k_xf(r_x^*)\left[r_x^*-\frac{1-F_x^{k_x}(r_x^*)}{k_xf(r_x^*)F^{k_x-1}(r_x^*)}\right]}{F(r_x^*)}=0.\nonumber
\end{equation}
Notice that the cdf of the subtree value $v_x$ of $T_x$ is given by $H_x = F^{k_x}(v_x)$, with the corresponding pdf given by $h_x(v_x) = k_x f(v_x) F^{k_x-1}(v_x)$, where $k_x$ represents the size of the subtree $T_x$. Thus, the above equation can be rewritten as:
\begin{equation}\label{eq_optimal_r_Tx_1}
     \frac{h_x(r_x^*)\left[r_x^*-\frac{1-H_x(r_x^*)}{h_x(r_x^*)}\right]}{H_x(r_x^*)}=0. \nonumber
\end{equation}
The hazard rate of the distribution $H_x$ for the subtree value $v_x$ of  $T_x$ is given by $\lambda_x = \frac{h_x}{1 - H_x}$. Thus, the optimal tree-wise reserve price can be determined as the solution to the following equation:
\begin{equation}\label{eq_optimal_r_Tx_2}
     r_x^*-\frac{1}{\lambda_x(r^*_x)}=0. \nonumber
\end{equation}

\subsubsection{Subtree Critical Value}

For the subtree-wise reserve price, the optimal reserve price is determined as the solution to $r - \frac{1 - H_x(r)}{h_x(r)} = 0$, which follows exactly the same form as the personalized reserve price derived in Myerson's optimal auction \citep{myerson1981optimal}. This observation leads to further exploration of the \textit{virtual valuation} function for each subtree.
\begin{definition}\label{definition_subtree_critical_value}  For any POST within a POT, the \textit{subtree critical value} is defined as:  
\begin{equation}  
    \phi(v_x; k_x) = v_x - \frac{1 - H_x(v_x)}{h_x(v_x)} = v_x - \frac{1 - F^{k_x}(v_x)}{k_x f(v_x) F^{k_x-1}(v_x)} .  
\end{equation}  
where $v_x$ represents the subtree value of $T_x$, and $k_x$ denotes the size of the subtree $T_x$.  
\end{definition}
The subtree critical value is essentially an extension of the traditional \textit{virtual valuation function} \citep{myerson1981optimal} in the context of social network-based auctions. Within this framework, each POST is treated as a whole, and the seller should ideally allocate the item to the POST subtree with the highest subtree critical value. Although all buyers' valuations are i.i.d., the distribution of the subtree value $v_x$ varies due to differences in the sizes of subtrees. Consequently, the subtree-wise reserve price dynamically adjusts to different subtree value distributions, enabling a structure-aware personalized price discriminationmbased on the subtree topology.

The monotonicity of the subtree critical value function $\phi(v_x; k_x)$ with respect to the subtree value $v_x$ and subtree size $k_x$ is crucial for the subsequent selection of the reserve price function and the feasibility analysis of its solution. First, we prove that when $F$ is a regular distribution, the function  $\phi(v_x; k_x) = v_x - \frac{1 - H_x(v_x)}{h_x(v_x)}$
is monotonically increasing with respect to the subtree value $v_x$.
\begin{lemma}
    When the buyers' valuation distribution $F$ satisfies regularity, the distribution $H$ of subtree values also satisfies regularity. That is, the subtree critical value function $\phi(v_x;k_x) = v_x - \frac{1-H_x(v_x)}{h_x(v_x)}$ is monotonically increasing with respect to $v_x$.
\end{lemma}

\begin{proof}
    In traditional optimal auction mechanism design, the buyers' valuation distribution satisfies regularity, which implies that the Myerson virtual valuation function $\psi(v) = v - \frac{1-F(v)}{f(v)}$ is monotonically increasing on $v$:
    {\small \begin{displaymath}
        \psi'(v) = 2 + \frac{f'(v)\big(1-F(v)\big)}{f^{2}(v)} > 0.
    \end{displaymath}}
    Taking the first-order partial derivative of the subtree critical value function $\phi(v_x;k_x)$ with respect to the subtree value $v_x$:
    {\small \begin{equation}\label{positive_derivative_DVL}
    \begin{aligned}
        &\frac{\partial \phi(v;k)}{\partial v_x}=2+\frac{\frac{\partial h_x}{\partial v_x} \cdot \left(1-H_x(v_x)\right)}{h_x^2(v_x)}\\
        =&2 + \frac{f'(v_x)\left(1-F(v_x)\right)}{f^{2}(v_x)} \times \frac{\left[F(v_x)+\frac{(k_x-1)}{f'(v_x)}f^{2}(v_x)\right]\cdot \left(1+F(v_x)+ ... + F^{k_x-1}(v_x)\right)}{k_x\cdot F^{k_x}(v_x)}. \\
    \end{aligned}
    \end{equation}} 
    Now, we classify the analysis based on the sign of $f'(v_x)$:
    \begin{enumerate}[leftmargin=*, labelsep=0.5em, label=\arabic*)]
    \item $f'(v_x) \geq 0$: Clearly, every term in the first-order derivative of $\phi(v_x;k_x)$ with respect to $v_x$ is non-negative, implying that $\phi(v_x;k_x)$ is monotonically increasing in $v_x$.
    \item $f'(v_x) < 0$: Then $\frac{f'(v_x)(1-F(v_x))}{f^{2}(v_x)} \leq 0$ holds.
    If $F(v_x)+\frac{(k_x-1)}{f'(v_x)}f^{2}(v_x) \leq 0$, then it is evident that the first-order derivative of $\phi(v_x;k_x)$ with respect to $v_x$ remains non-negative.
    If $F(v_x)+\frac{(k_x-1)}{f'(v_x)}f^{2}(v_x) > 0$, then given that $\psi(v) = v - \frac{1-F(v)}{f(v)}$ is a monotonically increasing function with respect to $v$, it follows from previous derivations that $\frac{f^{2}(v_x)}{f'(v_x)}<-\frac{1-F(x)}{2}$ holds, thus $F(v_x)>\frac{k_x-1}{k_x+1}$. Thus, we can derive as following:
        {\small \begin{displaymath}
         \begin{aligned}
             &\left[F(v_x)+\frac{(k_x-1)}{f'(v_x)}f^{2}(v_x)\right]\cdot \left[1+F(v_x)+ ... + F^{k_x-1}(v_x)\right]-k_x\cdot F^{k_x}(v_x) \\
             < & \left[F(v_x)-\frac{(k_x-1)(1-F(x))}{2}\right]\cdot \left[1+F(v_x)+ ... + F^{k_x-1}(v_x)\right]-k_x\cdot F^{k_x}(v_x) \\
             = & \frac{1}{2} \cdot \left[\frac{2\left(1-F^{k_x+1}(v_x)\right)}{1-F(v_x)}-(k_x+1)(1+F^{k_x}(v_x))\right] \\
             = & \frac{1}{2} \cdot [(k_x-1)F^{k_x+1}(v_x)-(k_x+1)F^{k_x}(v_x)+(k_x+1)F(v_x)-(k_x-1)] < 0.
         \end{aligned}
        \end{displaymath}}
        The final inequality follows from differentiating the function $\eta(x) = (k_x-1)x^{k_x+1}-(k_x+1)x^{k_x}+(k_x+1)x-(k_x-1)$ with respect to $x$, which is clearly monotonically increasing over $x \in [0,1]$. Therefore, for all subtree values $v_x$, we have $\eta(F(v_x)) \leq \eta(1) = 0$. Consequently, the second factor in the final line of eq.\eqref{positive_derivative_DVL} is at most $1$, allowing further simplification:
        {\small \begin{displaymath}
         \frac{\partial \phi(v;k)}{\partial v_x} \geq \psi'(v) >0
     \end{displaymath}}
    \end{enumerate}
    In conclusion, when the Myerson virtual valuation function $\psi(v)$ is monotonically increasing with respect to the valuation $v$ in traditional auctions, the subtree critical value function $\phi(v_x;k_x)$ is also monotonically increasing with respect to the subtree value $v_x$.
\end{proof}

This property of the subtree critical value can be intuitively understood as follows: when treating a POST subtree as a whole, the higher the highest bid among buyers within the subtree, the greater the probability that the item will be allocated to that subtree. This incentivizes buyers to bid truthfully and actively propagate the auction info to attract high-valuation buyers, thereby increasing the chances of their subtree winning the auction. Consequently, subtree with higher subtree value also has higher subtree critical value, making it more likely to win the auction.

\begin{corollary}\label{col_mono_v}  
    If the buyers' valuation distribution $F$ has a monotonically increasing hazard rate, the subtree critical value function $\phi(v_x;k_x) = v_x - \frac{1-H_x(v_x)}{h_x(v_x)}$ is monotonically increasing with respect to the subtree value $v_x$.  
\end{corollary}

Next, we investigate the monotonicity of the subtree critical value function $\phi(v_x; k_x)$ with respect to the subtree size $k_x$.

\begin{lemma}
    The subtree critical value function $\phi(v_x;k_x) = v_x - \frac{1-H_x(v_x)}{h_x(v_x)}$ is monotonically decreasing with respect to the subtree size $k_x$.
\end{lemma}

\begin{proof}
    First, we compute the first-order partial derivative of $\phi(v_x; k_x)$ with respect to the subtree size $k_x$:
    {\small \begin{displaymath}
        \begin{aligned}
            \frac{\partial \phi(v_x;k_x)}{\partial k_x}= \frac{-F^{k_x}(v_x)+1+k_x\ln F(v_x)}{k_x^2f(v_x)F^{k_x-1}(v_x)} \leq \frac{-F(v_x)+1+\ln F(v_x)}{k_x^2f(v_x)F^{k_x-1}(v_x)} \leq 0.
        \end{aligned}
    \end{displaymath}}
    The first simplification involves bounding the numerator of the derivative, denoted as the function $\alpha(k_x)$. Differentiating $\alpha(k_x)$ with respect to $k_x$ gives $\alpha'(k_x) = \ln F(v_x)(1 - F^{k_x}(v_x)) < 0$, indicating that $\alpha(k_x)$ is monotonically decreasing with respect to the subtree size $k_x$. Thus, for all $k_x$, the inequality $\alpha(k_x) \leq \alpha(1) = 1 + \ln F(v_x) - F(v_x)$ holds. The second simplification follows from a fundamental mathematical inequality, which states that for any $x \in (0,1]$, $ 1 + \ln x - x \leq 0$ always holds. Consequently, we obtain $\alpha(1) = 1 + \ln F(v_x) - F(v_x) \leq 0$. Therefore, the function $\phi(v_x;k_x) = v_x - \frac{1-H_x(v_x)}{h_x(v_x)}$ is monotonically decreasing with respect to the subtree size $k_x$.
\end{proof}

\begin{proposition}
\label{theorem_crit_mono}  
    If the buyers' valuation distribution $F$ satisfies regularity, the subtree critical value function $\phi(v_x;k_x) = v_x - \frac{1-H_x(v_x)}{h_x(v_x)}$ is monotonically increasing with respect to the subtree value $v_x$ and monotonically decreasing with respect to the subtree size $k_x$.
\end{proposition}

Based on the property that the subtree critical value is monotonically decreasing with respect to subtree size, consider two subtrees $T_{\alpha}$ and $T_{\beta}$ with sizes $k_{\alpha}$ and $k_{\beta}$, respectively, where both subtrees have the same subtree value $v_s$, and assume that $k_{\alpha} < k_{\beta}$. In this case, despite the smaller size of subtree $T_{\alpha}$, its virtual valuation is much higher, making it more likely to win the item. This result appears counterintuitive because, intuitively, a larger subtree should have a higher probability of containing high-valuation buyers, thereby giving it a competitive advantage. However, under our auction mechanism, buyers with critical diffusion behavior are rewarded, which may lead to scenarios where the seller's revenue is not necessarily maximized by allocating the item to the subtree with more buyers. Consequently, when subtree values are identical, a larger subtree does not always have an advantage. Instead, the higher expected propagation rewards associated with the larger subtree may result in the smaller subtree achieving higher effective revenues, ultimately making them more likely to win the auction.

This phenomenon can also be interpreted through price discrimination in traditional auction mechanisms, particularly in a way that favors "weaker" buyers. Although all buyers’ valuations are all drawn from the distribution $F$ independently, the distributions of  subtree values for $T_{\alpha}$ and $T_{\beta}$ are given by $H_{\alpha} = F^{k_{\alpha}}$ and $H_{\beta} = F^{k_{\beta}}$, respectively. Since the subtree sizes satisfy $k_{\alpha} < k_{\beta}$, it is evident that the hazard rates of these two distributions satisfy:
\begin{displaymath}
\begin{aligned}
    \lambda_{\alpha}(v_s) = \frac{h_{\alpha}(v_s)}{1-H_{\alpha}(v_s)} &=\frac{k_{\alpha}f(v_s)F^{k_{\alpha}-1}(v_s)}{1-F^{k_{\alpha}}(v_s)} \\
    &\geq \frac{k_{\beta}f(v_s)F^{k_{\beta}-1}(v_s)}{1-F^{k_{\beta}}(v_s)} =\frac{h_{\beta}(v_s)}{1-H_{\beta}(v_s)}= \lambda_{\beta}(v_s).
\end{aligned} 
\end{displaymath}
At this point, the distribution $F^{k_{\alpha}}$ is said to \textit{stochastically dominate} $F^{k_{\beta}}$ \citep{krishna2009auction}. The distribution of subtree value $v_{\alpha}$ in subtree $T_{\alpha}$, given by $H_{\alpha} = F^{k_{\alpha}}$, is relatively weaker due to the smaller number of buyers, which could lead to a lower subtree value. However, if both subtrees have the same subtree value $v_s$, the following inequality holds:
\begin{equation}\label{eq_other_dominates}
    \phi(v_s;k_{\alpha}) = v_s - \frac{1}{\lambda_{\alpha}(v_s)} \geq v_s - \frac{1}{\lambda_{\beta}(v_s)} = \phi(v_s;k_{\beta}),
\end{equation}
which implies that the subtree critical value of $T_{\alpha}$ is higher. Consequently, although $T_{\alpha}$ is disadvantaged in terms of the number of buyers, its higher critical value enables it to gain a competitive advantage in the auction.

\subsubsection{The Unique Solution of $r_{opt}$}

From Proposition \ref{theorem_crit_mono}, it follows that the solution to the optimal reserve price that maximizes the expected revenue of the \texttt{APX-R} mechanism is unique.

\begin{theorem}\label{theorem_only_ropt}
    When the buyers' valuation distribution $F$ satisfies regularity,  the solution of the optimal reserve price that maximizes the expected revenue of the \texttt{APX-R} mechanism is unique.
\end{theorem}

\begin{proof}
    Transforming the necessary condition for $r_{opt}$ in eq.\eqref{eq_reserve_opt}: 
    {\small \begin{displaymath}
        \begin{aligned}
            &\sum_{x=1}^m \left(1+k_x \cdot \frac{r_{opt} \cdot f(r_{opt})}{F(r_{opt})}-\left[{F^{k_x}(r_{opt})}\right]^{-1}\right)\\
            =& \frac{f(r_{opt})}{F(r_{opt})} \sum_{x=1}^m \left(k_x \cdot \left(r_{opt}- \frac{1-F^{k_x}(r_{opt})}{k_xf(r_{opt})F^{k_x-1}(r_{opt})}\right)\right) = 0.
        \end{aligned}
    \end{displaymath}}
    Through this simplification, the condition that the solution of $r_{opt}$ must satisfy can also be rewritten as:
     {\small \begin{equation}\label{eq_r_opt_form}
        \frac{f(r_{opt})}{F(r_{opt})} \sum_{x=1}^m \left(k_x \cdot \left(r_{opt}- \frac{1-F^{k_x}(r_{opt})}{k_xf(r_{opt})F^{k_x-1}(r_{opt})}\right)\right) = 0.
    \end{equation}}
    Define the function $\beta(r)$ as $\beta(r) =  \sum_{x=1}^m \left(k_x \cdot \left(r- \frac{1-F^{k_x}(r)}{k_xf(r)F^{k_x-1}(r)}\right)\right)$. From corollary \ref{col_mono_v}, it follows that each term in the summation of $\beta(r)$ is monotonically increasing with respect to $r$, which implies that $\beta(r)$ itself is also monotonically increasing. $r_{opt}$ is the unique solution to $\beta(r) = 0$. Moreover, since $\beta(0) < 0, \beta(\bar{v}) = k_x \cdot \bar{v} > 0$,  there exists a unique solution for $r_{opt}$.
\end{proof}

In traditional auction settings, when the reserve price is set higher than the optimal reserve price, the expected revenue declines rapidly. This result provides a guideline for designing an approximate optimal reserve price function, suggesting that it should be set to the left of $r_{opt}$. According to this theorem, we can search for an approximate optimal reserve price function within the range left of the unique optimal reserve price solution. Although the optimal reserve price solution is unique, in practice, this solution is influenced by buyers' diffusion strategies, meaning that it may fail to satisfy DSIC in terms of the propagation dimension of the mechanism.

\subsubsection{Properties of $r_{opt}$}

It can be observed that $r_{opt}$ is strongly correlated with the POT structure that emerges after buyers engage in propagation behavior. Specifically, it is highly dependent on the number of direct neighbors of the seller in the POT (denoted as $m$) and the number of buyers in each POST subtree (denoted as $k_x$). Therefore, whether $r_{opt}$ satisfies the incentive compatibility property in the propagation dimension of the mechanism requires further investigation.

To begin, consider an arbitrary buyer $i$. Assume that all other buyers truthfully propagate the auction information to all of their neighbors, except for one buyer $j$ in $e_i$, who does not receive the auction information from $i$. That is, buyer $i$ propagates the auction information only to $e'_i$, where $|e_i| - |e'_i| = 1$ and $e_i - e'_i = \{ j \}$. 

The key question is: How does this affect the structure of the POT compared to the scenario where all buyers truthfully propagate the auction information? This analysis will provide insight into how deviations from full propagation influence the resulting auction structure and the feasibility of using the $r_{opt}$ while maintaining DSIC in the propagation dimension.

\begin{enumerate}[leftmargin=*, labelsep=0.5em, label=\arabic*)]
\item Buyer $j$ cannot participate in the auction: This implies that the total number of buyers in the auction market will decrease. Moreover, buyer $i$ is the DCN of buyer $j$, then the number of direct neighbors of the seller in the POT remains unchanged. However, the number of buyers in $T_x$ which buyer $i$ belongs to will decrease. This structural change is illustrated in figure \ref{figure_diffusion_ic_1}.

\tikzset{global scale/.style={
    scale=#1,
    every node/.append style={scale=#1}
        }}
    \begin{figure}[ht]
        \centering
        \begin{tikzpicture}[global scale=1.2, box/.style={circle, draw}]
            \node[box,fill=gray,very thin, scale=0.6](s1) at(0,0){{$s$}};
            \node[box,fill={rgb:red,1;green,180;blue,80},scale=0.6](A1) at(-0.5,-0.6){$\mathbf{i}$};
            \node[box,fill={rgb:red,1;green,180;blue,80},scale=1](B1) at(0.5,-0.6){\quad};
            \node[box,fill=white, scale=0.6](C1) at(-0.5,-1.2){$\mathbf{j}$};
            \node[box,fill=white, scale=1](D1) at(-0.5,-1.8){\quad};
            \node at (0, -2.5) {\scriptsize{Social network}};
            \draw[line width=0.6pt] (s1) --(A1);
            \draw[line width=0.6pt] (s1) --(B1);
            \draw[line width=0.6pt] (A1) --(C1);
            \draw[line width=0.6pt] (C1) --(D1);

            \node[box,fill=gray,very thin, scale=0.6](s2) at(3,0){$s$};
            \node[box,fill={rgb:red,1;green,180;blue,80},scale=0.6](A2) at(2.5,-0.6){$\mathbf{i}$};
            \node[box,fill={rgb:red,1;green,180;blue,80},scale=1](B2) at(3.5,-0.6){\quad};
            \node[box,fill=white, scale=0.6](C2) at(2.5,-1.2){$\mathbf{j}$};
            \node[box,fill=white, scale=1](D2) at(2.5,-1.8){\quad};
            \node at (3, -2.5) {\scriptsize{POT of full diffusion}};
            \draw[->,  line width=0.6pt] (s2) --(A2);
            \draw[->,  line width=0.6pt] (s2) --(B2);
            \draw[->,  line width=0.6pt] (A2) --(C2);
            \draw[->,  line width=0.6pt] (C2) --(D2);

            \node[box,fill=gray,very thin, scale=0.6](s3) at(6,0){$s$};
            \node[box,fill={rgb:red,1;green,180;blue,80},scale=0.6](A3) at(5.5,-0.6){$\mathbf{i}$};
            \node[box,fill={rgb:red,1;green,180;blue,80},scale=1](B3) at(6.5,-0.6){\quad};
            \node at (6, -2.5) {\scriptsize{POT of $j \notin e'_i$}};
            \draw[->,  line width=0.4pt] (s3) --(A3);
            \draw[->,  line width=0.4pt] (s3) --(B3);
            \end{tikzpicture}
        \caption{POT changes when $j \notin e'_i$ - Case 1.}
        \label{figure_diffusion_ic_1}
\end{figure}
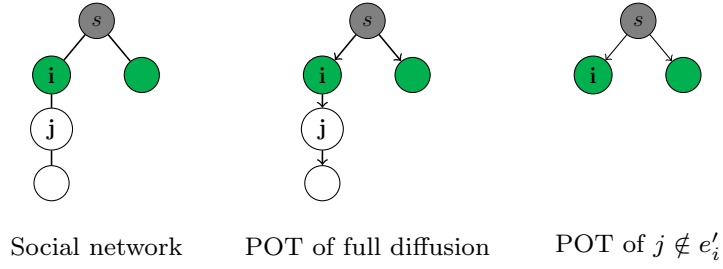

\item Buyer $j$ can still participate in the auction, and under the scenario where all buyers truthfully propagate the auction information, the buyers responsible for propagating the auction info to $j$ share a common DCN. This implies that the total number of buyers in the auction remains unchanged, and buyer $i$ is not the DCN of buyer $j$. Thus, when $j \notin e'_i$, both the value of $m$ and all $k_x$ remain unchanged. This structural change is illustrated in figure \ref{figure_diffusion_ic_2}.

\tikzset{global scale/.style={
    scale=#1,
    every node/.append style={scale=#1}
        }}
    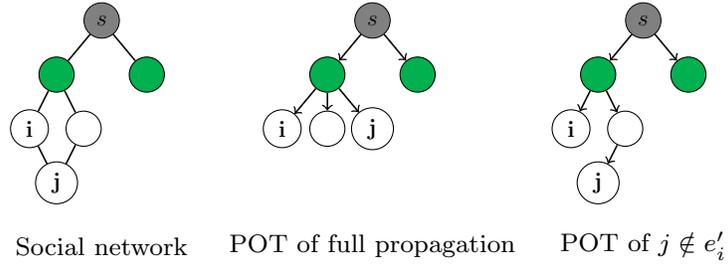
\begin{figure}[ht]
        \centering
        \begin{tikzpicture}[global scale=1.2, box/.style={circle, draw}]
            \node[box,fill=gray,very thin, scale=0.6](s1) at(0,0){{$s$}};
            \node[box,fill={rgb:red,1;green,180;blue,80},scale=1](A1) at(-0.5,-0.6){\quad};
            \node[box,fill={rgb:red,1;green,180;blue,80},scale=1](B1) at(0.5,-0.6){\quad};
            \node[box,fill=white, scale=0.6](C1) at(-0.8,-1.2){$\mathbf{i}$};
            \node[box,fill=white, scale=1](D1) at(-0.2,-1.2){\quad};
            \node[box,fill=white, scale=0.6](E1) at(-0.5,-1.8){$\mathbf{j}$};
            \node at (0, -2.5) {\scriptsize{Social network}};
            \draw[line width=0.6pt] (s1) --(A1);
            \draw[line width=0.6pt] (s1) --(B1);
            \draw[line width=0.6pt] (A1) --(C1);
            \draw[line width=0.6pt] (A1) --(D1);
            \draw[line width=0.6pt] (D1) --(E1);
            \draw[line width=0.6pt] (C1) --(E1);

            \node[box,fill=gray,very thin, scale=0.6](s2) at(3,0){{$s$}};
            \node[box,fill={rgb:red,1;green,180;blue,80},scale=1](A2) at(2.5,-0.6){\quad};
            \node[box,fill={rgb:red,1;green,180;blue,80},scale=1](B2) at(3.5,-0.6){\quad};
            \node[box,fill=white, scale=0.6](C2) at(2.0,-1.2){$\mathbf{i}$};
            \node[box,fill=white, scale=1](D2) at(2.5,-1.2){\quad};
            \node[box,fill=white, scale=0.6](E2) at(3,-1.2){$\mathbf{j}$};
            \node at (3, -2.5) {\scriptsize{POT of full propagation}};
            \draw[->,  line width=0.6pt] (s2) --(A2);
            \draw[->,  line width=0.6pt] (s2) --(B2);
            \draw[->,  line width=0.6pt] (A2) --(C2);
            \draw[->,  line width=0.6pt] (A2) --(D2);
            \draw[->,  line width=0.6pt] (A2) --(E2);

            \node[box,fill=gray,very thin, scale=0.6](s3) at(6,0){{$s$}};
            \node[box,fill={rgb:red,1;green,180;blue,80},scale=1](A3) at(5.5,-0.6){\quad};
            \node[box,fill={rgb:red,1;green,180;blue,80},scale=1](B3) at(6.5,-0.6){\quad};
            \node[box,fill=white, scale=0.6](C3) at(5.2,-1.2){$\mathbf{i}$};
            \node[box,fill=white, scale=1](D3) at(5.8,-1.2){\quad};
            \node[box,fill=white, scale=0.6](E3) at(5.5,-1.8){$\mathbf{j}$};
            \node at (6, -2.5) {\scriptsize{POT of $j \notin e'_i$}};
            \draw[->,  line width=0.6pt] (s3) --(A3);
            \draw[->,  line width=0.6pt] (s3) --(B3);
            \draw[->,  line width=0.6pt] (A3) --(C3);
            \draw[->,  line width=0.6pt] (A3) --(D3);
            \draw[->,  line width=0.6pt] (D3) --(E3);
            \end{tikzpicture}
        \caption{POT changes when $j \notin e'_i$ - Case 2.}
        \label{figure_diffusion_ic_2}
\end{figure}

\item Buyer $j$ can still participate in the auction, and under the scenario where all buyers truthfully propagate the auction information, none of the buyers responsible for propagating the auction info to $j$ share a common DCN. This implies that the total number of buyers in the auction remains unchanged, and buyer $i$ is not the DCN of buyer $j$. If there is only one path from the seller $s$ to buyer $j$, then when buyer $i$ does not propagate to buyer $j$, the value of $m$ decreases. Meanwhile, the POST subtree containing the only remaining path from the seller to buyer $j$ increases in size, as illustrated in figure \ref{figure_diffusion_ic_3}. If there are at least two paths from $s$ to buyer $j$, then $j$ remains a direct neighbor of the seller in the POT, and both $m$ and the sizes of all $k_x$ remain unchanged, as illustrated in figure \ref{figure_diffusion_ic_4}.

\tikzset{global scale/.style={
    scale=#1,
    every node/.append style={scale=#1}
        }}
    \begin{figure}[ht]
        \centering
        \begin{tikzpicture}[global scale=1.2, box/.style={circle, draw}]
            \node[box,fill=gray,very thin, scale=0.6](s1) at(0,0){{$s$}};
            \node[box,fill={rgb:red,1;green,180;blue,80},scale=1](A1) at(-0.6,-0.6){\quad};
            \node[box,fill={rgb:red,1;green,180;blue,80},scale=0.6](B1) at(0,-0.6){$\mathbf{i}$};
            \node[box,fill={rgb:red,1;green,180;blue,80},scale=1](C1) at(0.6,-0.6){\quad};
            \node[box,fill=white, scale=0.6](D1) at(0.3,-1.2){$\mathbf{j}$};
            \node at (0, -2) {\scriptsize{Social network}};
            \draw[line width=0.6pt] (s1) --(A1);
            \draw[line width=0.6pt] (s1) --(B1);
            \draw[line width=0.6pt] (s1) --(C1);
            \draw[line width=0.6pt] (B1) --(D1);
            \draw[line width=0.6pt] (C1) --(D1);

            \node[box,fill=gray,very thin, scale=0.6](s2) at(3,0){{$s$}};
            \node[box,fill={rgb:red,1;green,180;blue,80},scale=1](A2) at(2.1,-0.6){\quad};
            \node[box,fill={rgb:red,1;green,180;blue,80},scale=0.6](B2) at(2.7,-0.6){$\mathbf{i}$};
            \node[box,fill={rgb:red,1;green,180;blue,80},scale=1](C2) at(3.3,-0.6){\quad}; 
            \node[box,fill=white, scale=0.6](D2) at(3.9,-0.6){$\mathbf{j}$};
            \node at (3, -2) {\scriptsize{POT of full propagation}};
            \draw[->,  line width=0.6pt] (s2) --(A2);
            \draw[->,  line width=0.6pt] (s2) --(B2);
            \draw[->,  line width=0.6pt] (s2) --(C2);
            \draw[->,  line width=0.6pt] (s2) --(D2);

            \node[box,fill=gray,very thin, scale=0.6](s3) at(6,0){{$s$}};
            \node[box,fill={rgb:red,1;green,180;blue,80},scale=1](A3) at(5.4,-0.6){\quad};
            \node[box,fill={rgb:red,1;green,180;blue,80},scale=0.6](B3) at(6,-0.6){$\mathbf{i}$};
            \node[box,fill={rgb:red,1;green,180;blue,80},scale=1](C3) at(6.6,-0.6){\quad};
            \node[box,fill=white, scale=0.6](D3) at(6.6,-1.2){$\mathbf{j}$};
            \node at (6, -2) {\scriptsize{POT of $j \notin e'_i$}};
            \draw[->,  line width=0.6pt] (s3) --(A3);
            \draw[->,  line width=0.6pt] (s3) --(B3);
            \draw[->,  line width=0.6pt] (s3) --(C3);
            \draw[->,  line width=0.6pt] (C3) --(D3);
            
            \end{tikzpicture}
        \caption{POT changes when $j \notin e'_i$ - Case 3.}
        \label{figure_diffusion_ic_3}
\end{figure}
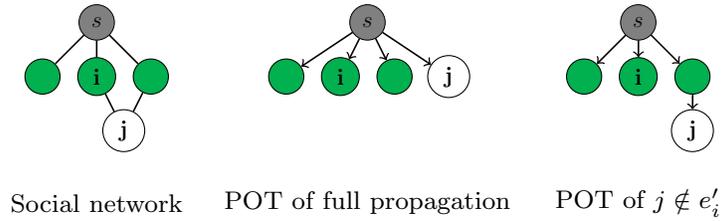
    \tikzset{global scale/.style={
    scale=#1,
    every node/.append style={scale=#1}
        }}
    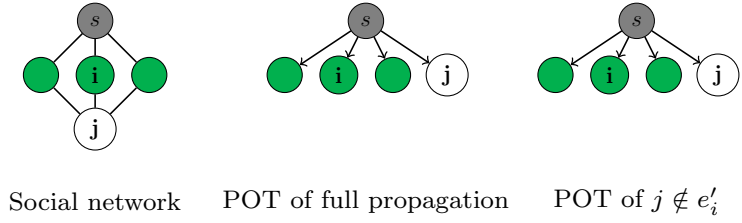
\begin{figure}[ht]
        \centering
        \begin{tikzpicture}[global scale=1.2, box/.style={circle, draw}]
            \node[box,fill=gray,very thin, scale=0.6](s1) at(0,0){{$s$}};
            \node[box,fill={rgb:red,1;green,180;blue,80},scale=1](A1) at(-0.6,-0.6){\quad};
            \node[box,fill={rgb:red,1;green,180;blue,80},scale=0.6](B1) at(0,-0.6){$\mathbf{i}$};
            \node[box,fill={rgb:red,1;green,180;blue,80},scale=1](C1) at(0.6,-0.6){\quad};
            \node[box,fill=white, scale=0.6](D1) at(0,-1.2){$\mathbf{j}$};
            \node at (0, -2) {\scriptsize{Social network}};
            \draw[line width=0.6pt] (s1) --(A1);
            \draw[line width=0.6pt] (s1) --(B1);
            \draw[line width=0.6pt] (s1) --(C1);
            \draw[line width=0.6pt] (A1) --(D1);
            \draw[line width=0.6pt] (B1) --(D1);
            \draw[line width=0.6pt] (C1) --(D1);

            \node[box,fill=gray,very thin, scale=0.6](s2) at(3,0){{$s$}};
            \node[box,fill={rgb:red,1;green,180;blue,80},scale=1](A2) at(2.1,-0.6){\quad};
            \node[box,fill={rgb:red,1;green,180;blue,80},scale=0.6](B2) at(2.7,-0.6){$\mathbf{i}$};
            \node[box,fill={rgb:red,1;green,180;blue,80},scale=1](C2) at(3.3,-0.6){\quad}; 
            \node[box,fill=white, scale=0.6](D2) at(3.9,-0.6){$\mathbf{j}$};
            \node at (3, -2) {\scriptsize{POT of full propagation}};
            \draw[->,  line width=0.6pt] (s2) --(A2);
            \draw[->,  line width=0.6pt] (s2) --(B2);
            \draw[->,  line width=0.6pt] (s2) --(C2);
            \draw[->,  line width=0.6pt] (s2) --(D2);

            \node[box,fill=gray,very thin, scale=0.6](s3) at(6,0){{$s$}};
            \node[box,fill={rgb:red,1;green,180;blue,80},scale=1](A3) at(5.1,-0.6){\quad};
            \node[box,fill={rgb:red,1;green,180;blue,80},scale=0.6](B3) at(5.7,-0.6){$\mathbf{i}$};
            \node[box,fill={rgb:red,1;green,180;blue,80},scale=1](C3) at(6.3,-0.6){\quad}; 
            \node[box,fill=white, scale=0.6](D3) at(6.9,-0.6){$\mathbf{j}$};
            \node at (6, -2) {\scriptsize{POT of $j \notin e'_i$}};
            \draw[->,  line width=0.6pt] (s3) --(A3);
            \draw[->,  line width=0.6pt] (s3) --(B3);
            \draw[->,  line width=0.6pt] (s3) --(C3);
            \draw[->,  line width=0.6pt] (s3) --(D3);
            
            \end{tikzpicture}
        \caption{POT changes when $j \notin e'_i$ - Case 4.}
        \label{figure_diffusion_ic_4}
\end{figure}

\end{enumerate}

In summary, when all other buyers truthfully propagate the auction information but buyer $i$ does not, the resulting POT structure may change in the following ways: 1) $m$ remains unchanged, while the size $k_x$ of one certain POST subtree $T_x$ decreases.  2) $m$ remains unchanged, and the sizes of all POST subtrees remain unchanged.  3) $m$ decreases, while the size $k_x$ of one certain POST subtree increases.  

We now analyze how the solution of $r_{opt}$ changes under each of these cases relative to the scenario where all buyers truthfully propagate auction info. That is, assuming all other buyers propagate truthfully, we examine how a single buyer reducing propagation affects the solution of $r_{opt}$. From Proposition \ref{theorem_crit_mono}, we know that the subtree critical value function $\phi(v_x; k_x)$ is monotonically increasing with respect to the subtree value $v_x$ and monotonically decreasing with respect to the subtree size $k_x$. By straightforward derivations, it follows that the left-hand side of the eq.\eqref{eq_r_opt_form} is also monotonically increasing with respect to the subtree value $v_x$ for any $T_x$ and monotonically decreasing with respect to its subtree size $k_x$. Define the function $\kappa(r; k_x) = k_x \cdot \phi(r; k_x)$.

\begin{enumerate}[leftmargin=*, labelsep=0.5em, label=\arabic*)]

\item  $m$ remains unchanged, and the size of one certain $T_x$ decreases: Assume that under truthful propagation, the optimal reserve price solution is $r_{opt}$, and the size of $T_x$ is $k_x$. When buyer $i$ does not propagate the auction info to buyer $j$, the size of the corresponding POST subtree in the new POT becomes $k'_x$. From the definition of $\kappa(r; k_x)$, $\sum_{x=1}^m \kappa(r_{opt}; k'_x) \geq \sum_{x=1}^m \kappa(r_{opt}; k_x) = 0$ holds.
Since the function $\kappa(r; k_x)$ is monotonically increasing with respect to $r$, when a buyer reduces propagation, causing $m$ to remain unchanged but the size of one $T_x$ to decrease, the new optimal reserve price solution becomes smaller compared to the original solution under truthful propagation.

\item $m$ remains Unchanged, and all POST subtree sizes remain unchanged: The optimal reserve price solution remains unchanged.

\item $m$ decreases, and the size of one certain $T_x$ increases: Assume that under truthful propagation, the optimal reserve price solution is $r_1$, the number of the seller's direct neighbors in the POT is $m_1$, the size of $T_j$ containing buyer $j$ is $k_j$, and the size of $T_l$ which contains the buyers along a path from the seller $s$ to buyer $j$ excluding buyer $i$ is $k_l$.  When buyer $i$ does not propagate the auction info to buyer $j$, the new optimal reserve price solution becomes $r_2$, and the number of the seller's direct neighbors in the POT changes to $m_2 = m_1 - 1$.  At this point, buyer $j$ is now part of $T_l$, whose size increases to $k_l + k_j$, while the sizes of all other POST subtrees remain unchanged. From eq.\eqref{eq_r_opt_form}:
\begin{displaymath}
    \begin{cases}
    \sum_{x=1,x \neq j,l}^{m_1} \kappa(r_1;k_x)+\kappa(r_1;k_j) + \kappa(r_1;k_l) = 0\\
    \sum_{x=1,x \neq l}^{m_2} \kappa(r_2;k_x)+\kappa(r_2;k_j+k_l) = 0
\end{cases}
    \end{displaymath}
Assuming that $r_2 \leq r_1$ holds:
{\small \begin{displaymath}
        \begin{aligned}
             &\sum_{x=1,x \neq l}^{m_2} \kappa(r_2;k_x)+\kappa(r_2;k_j+k_l) \\
            <& \sum_{x=1,x \neq l}^{m_2} \kappa(r_1;k_x)+\kappa(r_1;k_j+k_l) \\
             =&\kappa(r_1;k_j+k_l)-\kappa(r_1;k_j)-\kappa(r_1;k_l) \\
             =& \frac{\left[1-F^{k_j}(r_1)\right]\left[F^{k_l}(r_1)-1\right]}{f(r_1)F^{k_j+k_l-1}(r_1)} < 0
        \end{aligned}
    \end{displaymath}}
Thus the assumption does not hold, implying that $r_2 > r_1$. That is, in this case, when buyer $i$ does not propagate the auction information to buyer $j$, the optimal reserve price solution increases.
\end{enumerate}
Since the optimal reserve price solution is directly influenced by all buyers' propagation strategies, some buyers may strategically withhold propagation to decrease the optimal reserve price solution, thereby lowering the threshold to win the auction or reducing their payment. Similarly, some buyers may strategically withhold propagation to increase their expected propagation rewards by decreasing the reserve price.  However, for buyers whose non-truthful propagation results in an increase in the optimal reserve price solution, their utility will always be no higher than under truthful propagation.  Therefore, using the optimal reserve price cannot guarantee that all buyers truthfully propagate the auction information to all their neighbors. In other words, using $r_{opt}$ fails to ensure the DSIC of the \texttt{APX-R} mechanism.

\begin{proposition}\label{proposition_other_opt_notic}
    When buyers' valuations follow a regular distribution, the optimal reserve price that maximizes the expected revenue of the \texttt{APX-R} mechanism fails to satisfy DSIC.
\end{proposition}

According to proposition \ref{proposition_other_opt_notic}, although $r_{opt}$ maximizes the expected revenue of the \texttt{APX-R} mechanism, its failure to satisfy DSIC in the propagation dimension may prevent certain high-valuation buyers from participating in the auction, thereby affecting overall revenue. Thus, when designing the reserve price, it is possible to relax the strict requirement of revenue optimality and instead seek a simple reserve price function that both approximates $r_{opt}$ and ensures incentive compatibility. Notably, the search space for an approximate optimal reserve price is extremely large, and the goal is to find a simple yet effective reserve price function $\gamma$ that can significantly improve the seller’s expected revenue. According to theorem \ref{proposition_APX_DSIC}, if the reserve price function is independent of buyers' strategies, then the \texttt{APX-R} mechanism can be ensured to satisfy DSIC. Therefore, special attention should be given to simple reserve price functions that do not depend on buyers' bidding strategies and propagation strategies. However, within this space of such reserve price functions, the key challenge is optimizing the seller's expected revenue under this constraint. That is, selecting a reserve price function that optimizes revenue while ensuring DSIC. By constructing a reasonable approximate optimal reserve price, it is possible to maintain incentive compatibility while effectively optimizing the seller’s revenue.

\subsection{Near-Optimal Reserve Price}

\subsubsection{Secure Revenue Constraints}
Similarly, when selecting the reserve price function in the \texttt{APX-R} mechanism, the primary requirement is to ensure that the reserve price function $\gamma$ does not result in lower revenue for the seller compared to the case where no reserve price is set. That is, $\operatorname{APX}(\mathbf{t'};\gamma) \geq \operatorname{APX}(\mathbf{t'};0)$ must always hold. By substituting the expected revenue formula of the \texttt{APX-R} mechanism from eq.\eqref{eq_revenue_graph}, we obtain  $\sum_{x=1}^m [\operatorname{APX}_{T_x}(\mathbf{t'};\gamma)-\operatorname{APX}_{T_x}(\mathbf{t'};0)] \geq 0$. This constraint is referred to as the feasibility constraint for baseline revenue of the mechanism’s reserve price function $\gamma$, ensuring that the expected revenue when using $\gamma$ is not lower than when no reserve price is set. If each term in the summation is non-negative, i.e.,  $\forall T_x, \operatorname{APX}_{T_x}(\mathbf{t'};\gamma)-\operatorname{APX}_{T_x}(\mathbf{t'};0) \geq 0$, then the secure revenue guarantee is satisfied.  Thus we obtain a sufficient condition for the reserve price function $\gamma$ to satisfy the secure revenue constraint.

\begin{proposition}\label{proposition_secure_revenue_regular}
    If the reserve price function $\gamma$ satisfies the following inequality for any $T_x$:
    \begin{equation}\label{eq_secure_revenue_regular}
        \int_0^{\gamma} \left[ \frac{k_x}{n}F^n(v_x)-F^n(v_x)+F^{n-k_x}(v_x)\right]dv_x \geq \frac{k_x}{n}\gamma F^n(\gamma).
    \end{equation}
    then the reserve price $\gamma$ necessarily satisfies the secure revenue constraint, ensuring that the expected revenue of the \texttt{APX-R} mechanism with reserve price $\gamma$ is not lower than when no reserve price is applied.
\end{proposition}

The inequality \eqref{eq_secure_revenue_regular} indicates that different POST subtrees impose different constraints to achieve secure revenue, which reflects a form of price discrimination at the POST subtree level \citep{celis2014buy}, applying different constraints to different buyer groups. Since the entire POT tree consists of $m$ POST subtrees, satisfying the secure revenue constraint requires fulfilling $m$ inequality constraints. Thus, in designing an approximate optimal reserve price function, it is crucial to ensure that all POST subtrees satisfy the secure revenue constraint, guaranteeing that the expected revenue of the \texttt{APX-R} mechanism with the approximate reserve price is not lower than when no reserve price is used.

\subsubsection{Near-Optimal Reserve Price}

Since the optimal reserve price that maximizes the expected revenue of the \texttt{APX-R} mechanism does not satisfy DSIC. According to proposition \ref{proposition_APX_DSIC}, using a reserve price function independent of both buyers' bidding strategies and propagation strategies can ensure the DSIC property of the mechanism. Therefore, for various regular distributions of valuations, we aim to identify a reserve price function that ensures incentive compatibility while optimizing the seller’s expected revenue as much as possible. 

First, we discuss the tree-wise optimal reserve price, which maximizes the expected revenue of each POST subtree. Let the optimal reserve price for subtree $T_x$ be denoted as $r_x^*$. According to eq.\eqref{eq_optimal_r_Tx}, given a buyer valuation distribution, the solution $r_x^*$ depends only on the subtree size $k_x$. Now we investigate how the subtree-wise optimal reserve price $r_x^*$ is constrained by subtree size $k_x$. Since this section is based on the assumption that buyers' valuation distributions satisfy regularity, all subsequent analyses adhere to this assumption. By the definition \ref{definition_subtree_critical_value}, the subtree-wise optimal reserve price $r_x^*$ is the solution to the equation where the subtree critical value equals zero. According to Proposition \ref{theorem_crit_mono}, when the buyers' valuation distribution $F$ satisfies regularity, the subtree critical value function $\phi(v_x; k_x)$ is monotonically increasing with respect to the subtree value $v_x$ and monotonically decreasing with respect to the subtree size $k_x$.

Figure \ref{fig_monotonicity} illustrates the variation of subtree critical values $\phi(v_x; k_1)$ and $\phi(v_x; k_2)$ as functions of subtree value $v_x$ when subtree sizes satisfy $k_1 < k_2$. From the figure, it is evident that due to the monotonicity of the subtree critical value function concerning both subtree value and subtree size, the solution to the subtree-wise optimal reserve price $r_x^*$ (i.e., the solution to the equation $\phi(r; k_x) = r - (1 - H_x(r)) / h_x(r) = 0$) is monotonically increasing with respect to subtree size $k_x$. In other words, POST subtrees with more buyers tend to have higher subtree-wise optimal reserve price solutions $r_x^*$.
\begin{figure}[h]
\centering
\includegraphics[scale=0.6]{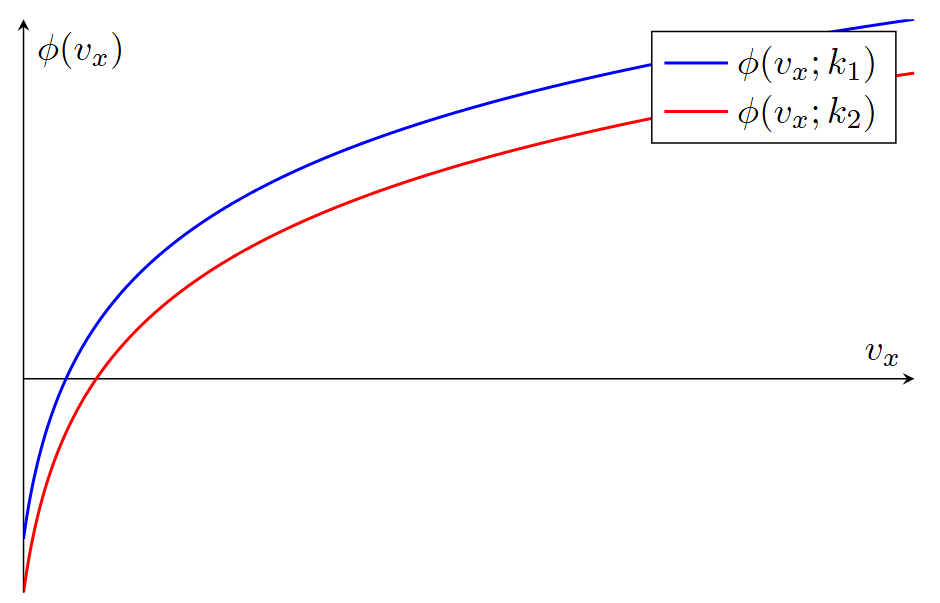}
\caption{Variation of subtree critical values $\phi(v_x; k_1)$ and $\phi(v_x; k_2)$ when $k_1 < k_2$.}
\label{fig_monotonicity}
\end{figure}

Figure \ref{fig_tx_rev_monotonicity} visually illustrates how the expected revenue of POST subtrees of different sizes varies with the reserve price. This figure is based on the assumption that buyers' valuations follow an exponential distribution $Exp(0.08)$. The triangle markers in the figure represent the subtree-wise optimal reserve price corresponding to each subtree size $k_x$. It is evident that as the subtree size $k_x$ increases, the subtree-wise optimal reserve price also increases.
\begin{figure}[h]
\centering
\includegraphics[scale=0.43]{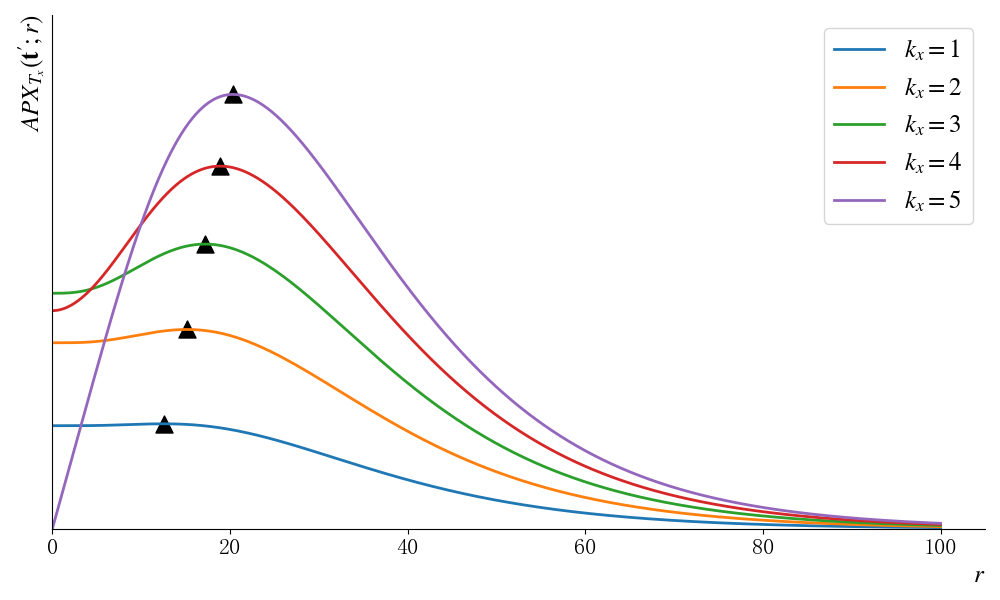}
\caption{Expected revenue of POST subtrees of different sizes as a function of the reserve price.}
\label{fig_tx_rev_monotonicity}
\end{figure}

\begin{lemma}\label{lemma_subtree_reserve_mono}
    The subtree-wise optimal reserve price $r_x^*$ of the \texttt{APX-R} mechanism, which maximizes the expected revenue $\operatorname{APX}_{T_x}(\mathbf{t}';r)$ of subtree $T_x$, is monotonically increasing with subtree size $k_x$.
\end{lemma}

\begin{theorem}\label{theorem_gamma_extend}
    Let $\underline{k}$ represent the smallest size of POST subtrees among all possible POT trees generated from historical auction instances. If buyers' valuations follow a regular distribution, the \texttt{APX-R} mechanism using a reserve price function $\gamma$ such that
    \begin{equation}\label{eq_gamma_extend}
        \gamma = \frac{1-F^{\underline{k}}(\gamma)}{\underline{k}f(\gamma)F^{\underline{k}-1}(\gamma)} 
    \end{equation}
    is secure and approximates the global optimal reserve price $r_{opt}$ from the left. 
\end{theorem}

\begin{proof}
    Since for any $T_x$, it holds that $\underline{k} \leq k_x$, the reserve price function $\gamma$ essentially corresponds to the optimal reserve price of a subtree of size $\underline{k}$. By Lemma \ref{lemma_subtree_reserve_mono}, we have $\gamma \leq r_x^*$ for all subtrees $T_x$.
    Additionally, since the expected revenue function $\operatorname{APX}_{T_x}(\mathbf{t}';r)$ is monotonically increasing in $r$ over $(0, r_x^*)$ and monotonically decreasing over $(r_x^*, \bar{v})$, the total expected revenue $\operatorname{APX}(\mathbf{t}';r)$ is monotonically increasing in $r$ over $(0, \gamma)$. 
    Thus, it necessarily holds that $\operatorname{APX}(\mathbf{t}';\gamma) > \operatorname{APX}(\mathbf{t}';0)$, meaning that using the reserve price function $\gamma$ in \texttt{APX-R} yields a higher expected revenue compared to not setting any reserve price. Furthermore, since $r_{opt}$ maximizes the total expected revenue $\operatorname{APX}(\mathbf{t}';r)$, we obtain $\gamma \leq r_{opt}$, which confirms that $\gamma$ approximates $r_{opt}$ from the left.
\end{proof}

The definition of $\underline{k}$ remains consistent with the previous section, where it is an a prior estimation of the network structure based on historical auction data, and is independent of the current auction process. Therefore, it is not influenced by any buyer's propagation strategy. Similarly, if no prior auction has been conducted on the social network and the seller lacks prior knowledge about $\underline{k}$, they can adopt a conservative estimate for $\underline{k}$ (e.g., setting $\underline{k} = 1$ and using the reserve price from Myerson's optimal auction) to initialize the \texttt{APX-R} mechanism. As the mechanism ensures DSIC, multiple auctions will allow buyers to reveal a more accurate estimate of $\underline{k}$, ultimately improving revenue.

Regarding the individual rationality (IR) and weak budget balance (WBB) properties of the mechanism, according to propositions \ref{proposition_WBB} and \ref{proposition_IR}, this approximate optimal reserve price function also preserves the IR and WBB properties of the mechanism.

\section{Simulations}
In this section, we explore the performance of the \texttt{APX-R} mechanism under different assumptions about buyers' valuation distributions. We assume that all buyers' valuations are independently and identically distributed drawn from the uniform, normal, and exponential distributions, and conduct a series of experiments based on these assumptions.  

For consistency, we set the upper bound of all valuation distributions to $\bar{v} = 100$, meaning that valuations are defined over the interval $[0,100]$. In the experiments, each auction mechanism is simulated one million times, with all buyers’ valuations independently drawn from the same distribution. To ensure that the vast majority of valuations fall within $[0,100]$, the following distribution settings are used:

\begin{enumerate}[leftmargin=*, labelsep=0.5em, label=\arabic*)]
    \item Uniform Distribution: Buyers' valuations are i.i.d. according to the Uniform distribution $U[0,100]$. Under this setting, the expected value of buyers' valuations is 50, and the variance is 833.33.
    \item Normal Distribution: To ensure that most valuation samples remain within $[0,100]$, buyers' valuations are i.i.d. following a Normal distribution $N(50,16.67^2)$, where the mean is $\mu = 50$ and the standard deviation is $\sigma = 16.67$. This ensures that $F(100) = 0.9987 \approx 1$, meaning nearly all values are within $[0,100]$. Under this setting, the expected value of buyers' valuations is 50, and the variance is 277.89.
    \item Exponential Distribution: To ensure that most valuation samples fall within $[0,100]$, the rate parameter of the Exponential distribution is set to $\lambda = 0.08$, ensuring that $F(100) = 0.9997 \approx 1$. Thus, buyers' valuations are i.i.d. following $Exp(0.08)$. Under this setting, the expected value of buyers' valuations is 12.5, and the variance is 156.25.
\end{enumerate}

\subsection{\texttt{APX-R} Under Different Estimates of $\underline{k}$}

To analyze the revenue performance of the \texttt{APX-R} mechanism under different prior estimates of the network structure parameter $\underline{k}$, we test in simulation experiments using the POT tree structure shown in figure \ref{fig_different_gamma}. We consider the following four cases:  (1) Without reserve price, where the mechanism degenerates into the \texttt{IDM}mechanism \citep{li2017mechanism, li2022diffusion} (i.e., $\gamma=0$);  (2) Using the reserve price from Myerson’s optimal auction, which corresponds to $\underline{k} = 1$;  (3) A conservatively estimated reserve price, obtained by setting $\underline{k} = 2$, referred to as the cautious reserve;  (4) A reserve price based on an accurate estimate of $\underline{k}$, obtained by setting $\underline{k} = 3$, referred to as the accurate reserve. 

\tikzset{global scale/.style={
    scale=#1,
    every node/.append style={scale=#1}
  }}
  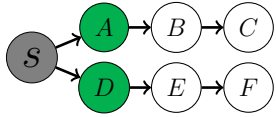
\begin{figure}[ht]
 \centering
        \centering
        \begin{tikzpicture}[global scale=0.8, box/.style={circle, draw}]
            \node[box,fill=gray,very thin](s) at(0,-0.5){{\LARGE$s$}};
            \node[box,fill={rgb:red,1;green,180;blue,80}](A) at(1.2,-1){$D$};
            \node[box,fill=white](H) at(2.4,-1){$E$};
            \node[box,fill=white](I) at(3.6,-1){$F$};
            
            \node[box,fill={rgb:red,1;green,180;blue,80}](B) at(1.2,0){$A$};
            \node[box,fill=white](C) at(2.4,0){$B$};
            \node[box,fill=white](D) at(3.6,0){$C$};
            
            \draw[->,  line width=1pt] (s) --(A);
            \draw[->,  line width=1pt] (A) --(H);
            \draw[->,  line width=1pt] (H) --(I);
            \draw[->,  line width=1pt]  (s) --(B);
            \draw[->,  line width=1pt]  (B) --(C);
            \draw[->,  line width=1pt]  (C) --(D);
        \end{tikzpicture}
        \caption{POT tree structure used to analyze \texttt{APX-R} under different reserve price settings.}
        \label{fig_different_gamma}
\end{figure}

Figure \ref{fig_different_gamma_rev} presents the revenue distributions of the \texttt{APX-R} mechanism under these four reserve price settings. Each subfigure illustrates the frequency of different revenue values across simulation runs. {The four gradient-colored histograms in the figure represent the revenue distributions of the \texttt{APX-R} mechanism under different estimations of $\underline{k}$. Due to the reserve price setting, there may be instances where the revenue is zero, this occurs when all buyers' valuations fall below the reserve price, resulting in an auction failure. Consequently, for certain valuation distributions, setting different reserve prices leads to histograms appearing at zero revenue, with varying heights depending on the frequency of auction failures caused by different reserve prices.} From the figure, it is evident that whether valuations follow a uniform or normal, exponential distribution, the revenue distribution shifts to the right as the reserve price estimation improves from no reserve price ($\gamma=0$), to the conservative estimate ($\underline{k}=2$), and finally to the accurate estimate ($\underline{k}=3$).  This observation indicates that a more precise estimate of $\underline{k}$ leads to significantly higher revenue, demonstrating that optimally setting the reserve price can improve expected revenue. More specifically, as the accuracy of the reserve price estimation increases, the expected revenue also increases, with the revenue distribution concentrating more on higher values.
\begin{figure}[ht]
    \centering
    \subfigure[$U{{[}}0,100{{]}}$]{
        \includegraphics[width=0.24\textwidth]{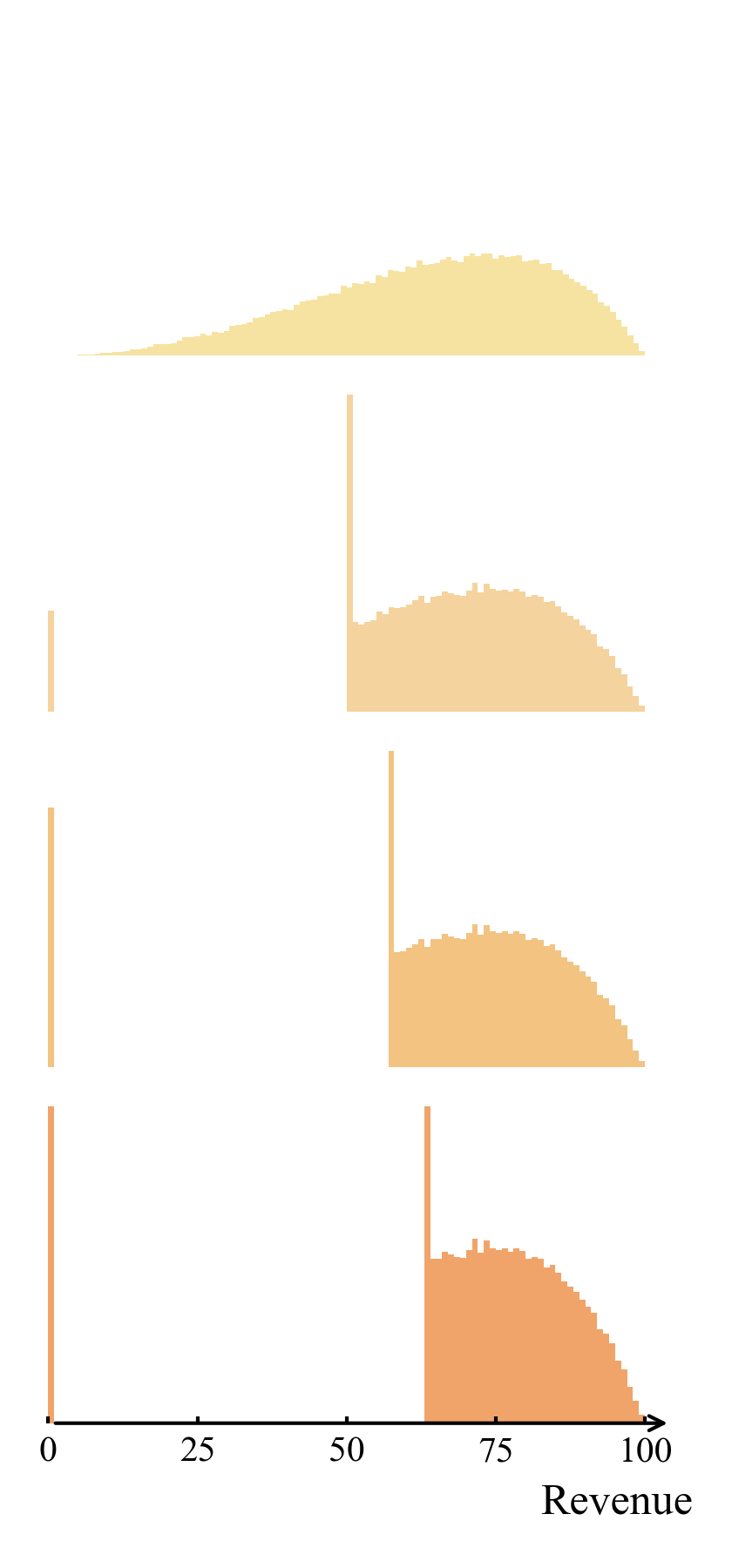}
        \label{fig_uniform_distribution2}
    }
    \subfigure[$N(50,16.67^2)$]{
        \includegraphics[width=0.24\textwidth]{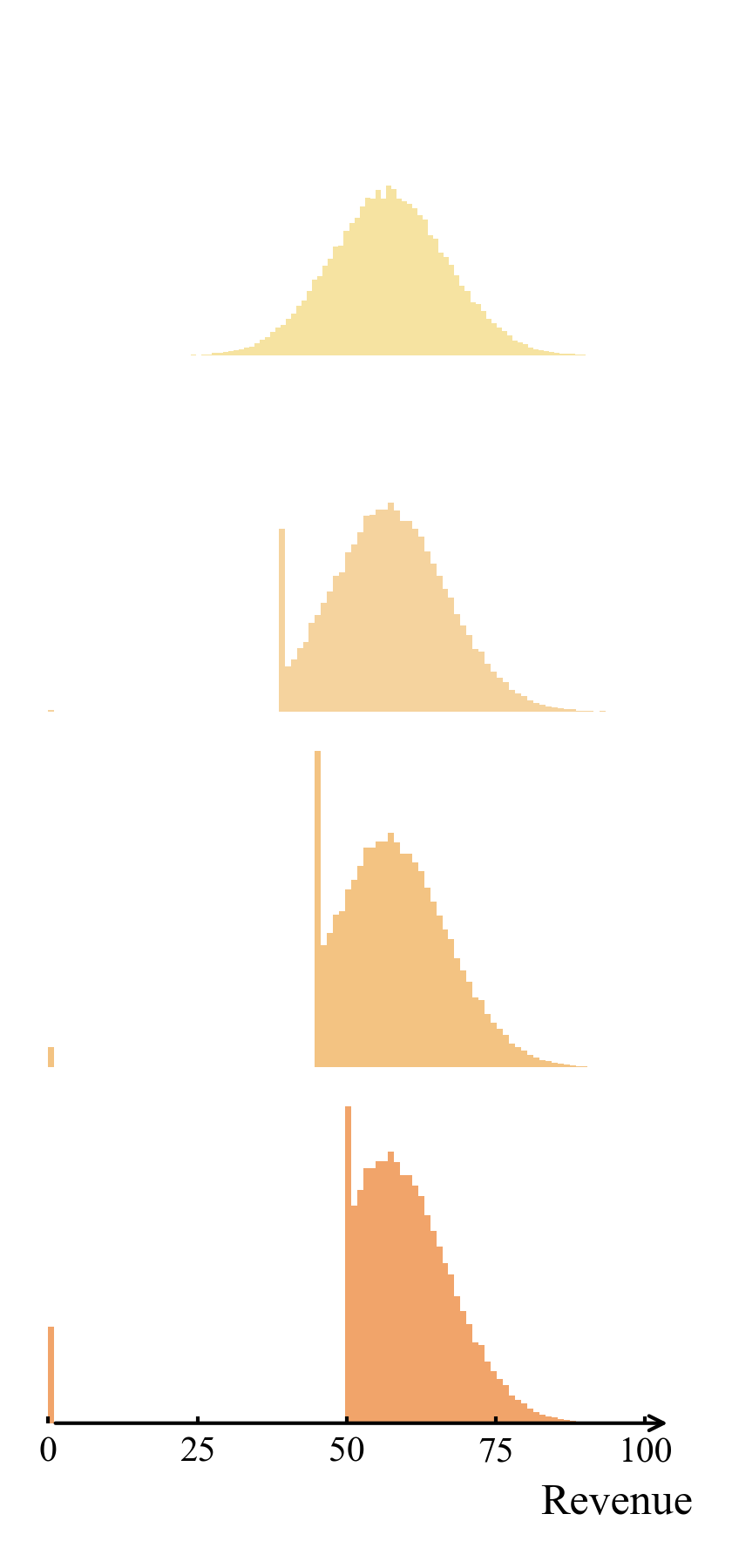}
        \label{fig_norm_distribution2}
    }
    \subfigure[$Exp(0.08)$]{
        \includegraphics[width=0.24\textwidth]{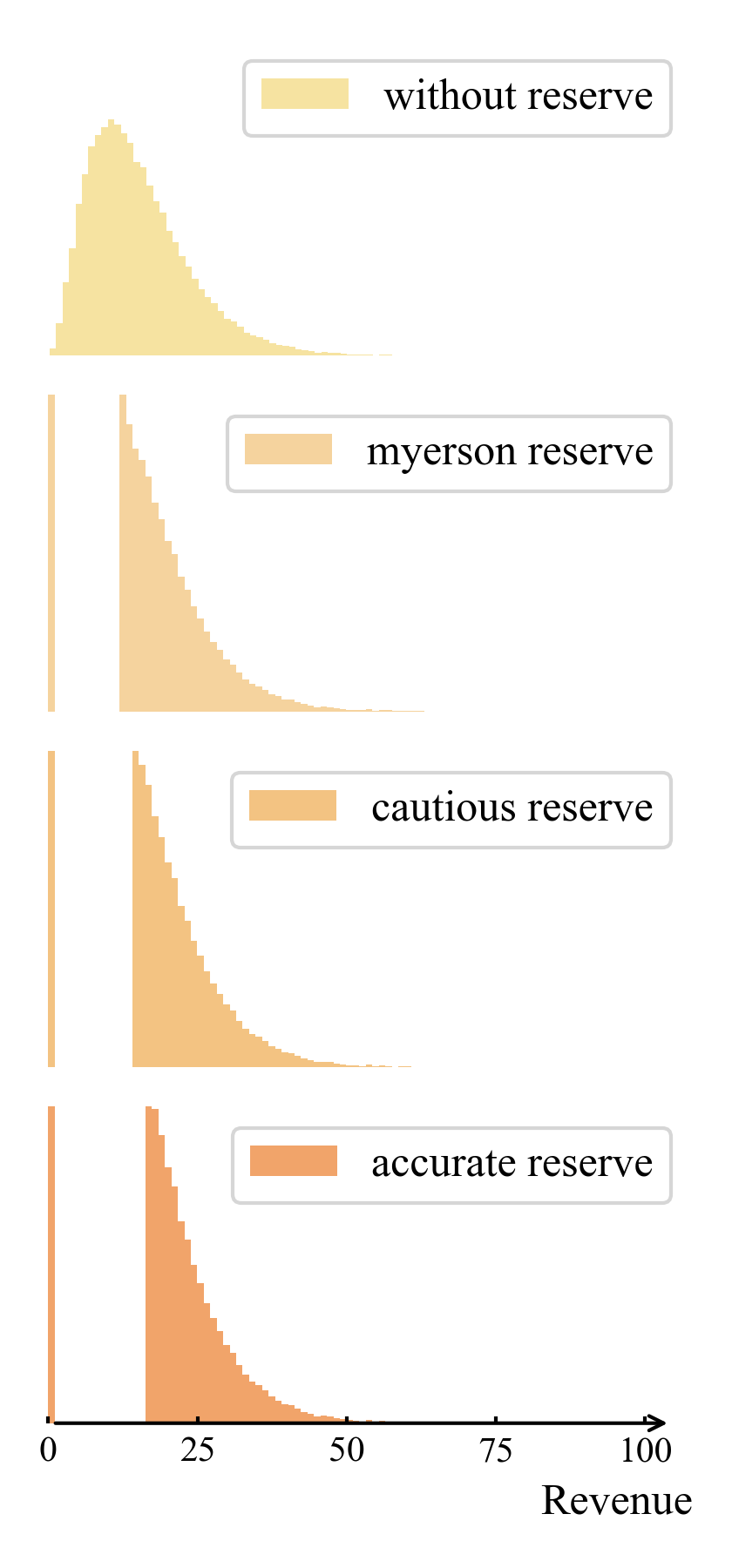}
        \label{fig_exp_distribution2}
    }
    \caption{Revenue distributions of \texttt{APX-R} under different reserves.}
    \label{fig_different_gamma_rev}
\end{figure}

\subsection{Impact of POT Tree Symmetry on \texttt{APX-R}}

To examine how the symmetry of POT trees affects the revenue performance of the \texttt{APX-R} mechanism, we conducted simulations using three POT tree structures with different levels of symmetry, as shown in Figure \ref{fig_symmetric}.

According to proposition \ref{proposition_WBB} in Chapter 3 and the seller’s expected revenue formula \eqref{eq_APX_r}, the seller’s expected revenue depends only on the number of buyers in each POST subtree, rather than the specific structure of each subtree. Therefore, when analyzing POT tree symmetry, we only need to focus on subtree size combinations, rather than their exact structural configurations. In all three POT tree structures, the seller's number of direct neighbors is $m = 2$, and the total number of buyers is $n = 6$, but the size combinations of the POST subtrees vary: In POT tree (a), the subtree size combination is $1+5$; in POT tree (b), the subtree size combination is $2+4$; In POT tree (c), the subtree size combination is $3+3$. Since the smallest POST subtree size in Figure \ref{fig_symmetric} is $1$, we use the approximate reserve price function $\gamma$ for $\underline{k} = 1$ in the \texttt{APX-R} simulations.

\tikzset{global scale/.style={
    scale=#1,
    every node/.append style={scale=#1}
  }}
  \begin{figure}[ht]
 \centering
        \centering
        \begin{tikzpicture}[global scale=0.8, box/.style={circle, draw}]
            \node[box,fill=gray,very thin](s1) at(-15,-0.5){{\LARGE$s$}};
            \node[box,fill={rgb:red,1;green,180;blue,80}](F1) at(-13.8,-1){$F$};
            
            \node[box,fill={rgb:red,1;green,180;blue,80}](A1) at(-13.8,0){$A$};
            \node[box,fill=white](B1) at(-12.6,0){$B$};
            \node[box,fill=white](C1) at(-11.4,0){$C$};
            \node[box,fill=white](D1) at(-10.2,0){$D$};
            \node[box,fill=white](E1) at(-9,0){$E$};
            \node at (-11.8,-2.5) {{\large POT structure (a)}};
            \draw[->,  line width=1pt] (s1) --(A1);
            \draw[->,  line width=1pt] (A1) --(B1);
            \draw[->,  line width=1pt] (B1) --(C1);
            \draw[->,  line width=1pt]  (C1) --(D1);
            \draw[->,  line width=1pt]  (D1) --(E1);
            \draw[->,  line width=1pt]  (s1) --(F1);

            \node[box,fill=gray,very thin](s2) at(-7,-0.5){{\LARGE$s$}};
            \node[box,fill={rgb:red,1;green,180;blue,80}](E2) at(-5.8,-1){$E$};
            \node[box,fill=white](F2) at(-4.6,-1){$F$};
            \node[box,fill={rgb:red,1;green,180;blue,80}](A2) at(-5.8,0){$A$};
            \node[box,fill=white](B2) at(-4.6,0){$B$};
            \node[box,fill=white](C2) at(-3.4,0){$C$};
            \node[box,fill=white](D2) at(-2.2,0){$D$};
            \node at (-5.2,-2.5) {{\large POT structure (b)}};
            
            \draw[->,  line width=1pt] (s2) --(A2);
            \draw[->,  line width=1pt] (A2) --(B2);
            \draw[->,  line width=1pt] (B2) --(C2);
            \draw[->,  line width=1pt]  (C2) --(D2);
            \draw[->,  line width=1pt]  (s2) --(E2);
            \draw[->,  line width=1pt]  (E2) --(F2);
            
            \node[box,fill=gray,very thin](s) at(0,-0.5){{\LARGE$s$}};
            \node[box,fill={rgb:red,1;green,180;blue,80}](D) at(1.2,-1){$D$};
            \node[box,fill=white](E) at(2.4,-1){$E$};
            \node[box,fill=white](F) at(3.6,-1){$F$};
            
            \node[box,fill={rgb:red,1;green,180;blue,80}](A) at(1.2,0){$A$};
            \node[box,fill=white](B) at(2.4,0){$B$};
            \node[box,fill=white](C) at(3.6,0){$C$};
            \node at (2.2,-2.5) {{\large POT structure (c)}};
            
            \draw[->,  line width=1pt] (s) --(A);
            \draw[->,  line width=1pt] (A) --(B);
            \draw[->,  line width=1pt] (B) --(C);
            \draw[->,  line width=1pt]  (s) --(D);
            \draw[->,  line width=1pt]  (D) --(E);
            \draw[->,  line width=1pt]  (E) --(F);
        \end{tikzpicture}
        \caption{POT structure for exploring \texttt{APX-R} revenue distribution under different levels of symmetry.}
        \label{fig_symmetric}
\end{figure}
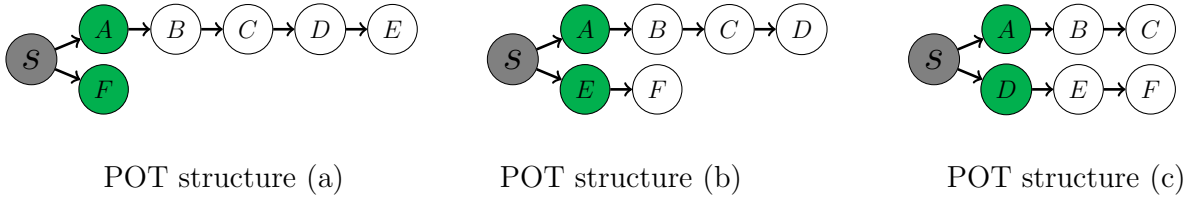

Figure \ref{fig_symmetric_rev} presents the revenue distributions of the \texttt{APX-R} mechanism using the approximate reserve price function $\gamma$ for $\underline{k} = 1$ in the three POT tree structures shown in Figure \ref{fig_symmetric}. Each histogram represents the frequency of the corresponding revenue values on the x-axis, illustrating how the symmetry of the POT tree structure affects \texttt{APX-R} revenue. {The three gradient-colored histograms in the figure represent the revenue distributions of the \texttt{APX-R} mechanism under different levels of symmetry in the POT structure. Due to the reserve price setting, a revenue of zero may still occur. Furthermore, for any given valuation distribution, since the reserve price is the same, the total number of buyers remains unchanged, and only the POT structure differs, the probability of auction failure remains consistent across cases. Consequently, the height of the corresponding histogram also remains the same.} It is evident that under all three different distributions, as the symmetry of the POT tree increases, the \texttt{APX-R} revenue distribution shifts to the right, indicating higher revenues. Thus, POT tree symmetry is positively correlated with the seller’s expected revenue. This trend suggests that as the POT tree becomes more symmetric, the \texttt{APX-R} mechanism’s revenue becomes more stable. Therefore, higher POT tree symmetry enhances the expected revenue of the \texttt{APX-R} mechanism.

\begin{figure}[ht]
    \centering
    \subfigure[$U{{[}}0,100{{]}}$]{
        \includegraphics[width=0.24\textwidth]{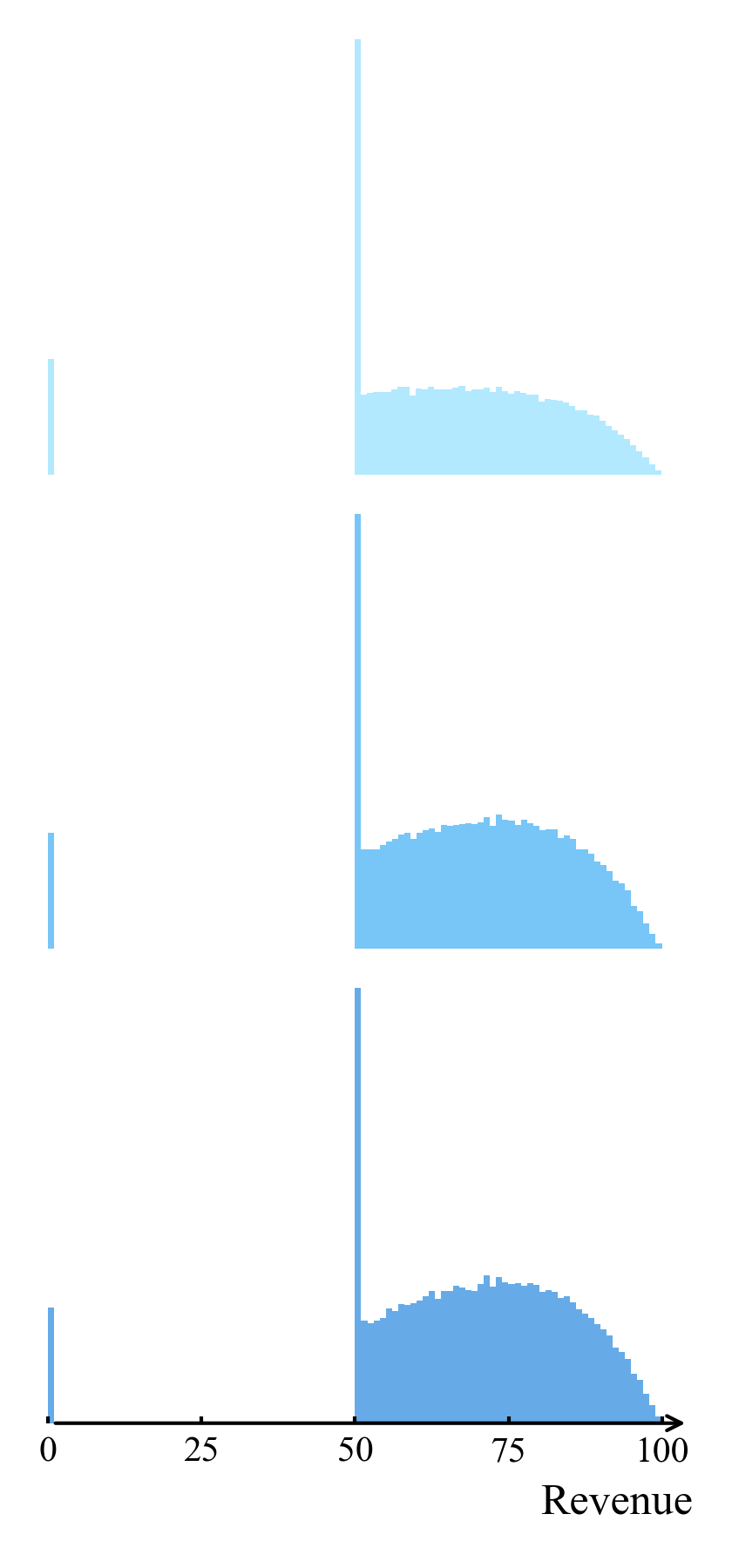}
        \label{fig_uniform_distribution1}
    }
    \subfigure[$N(50,16.67^2)$]{
        \includegraphics[width=0.24\textwidth]{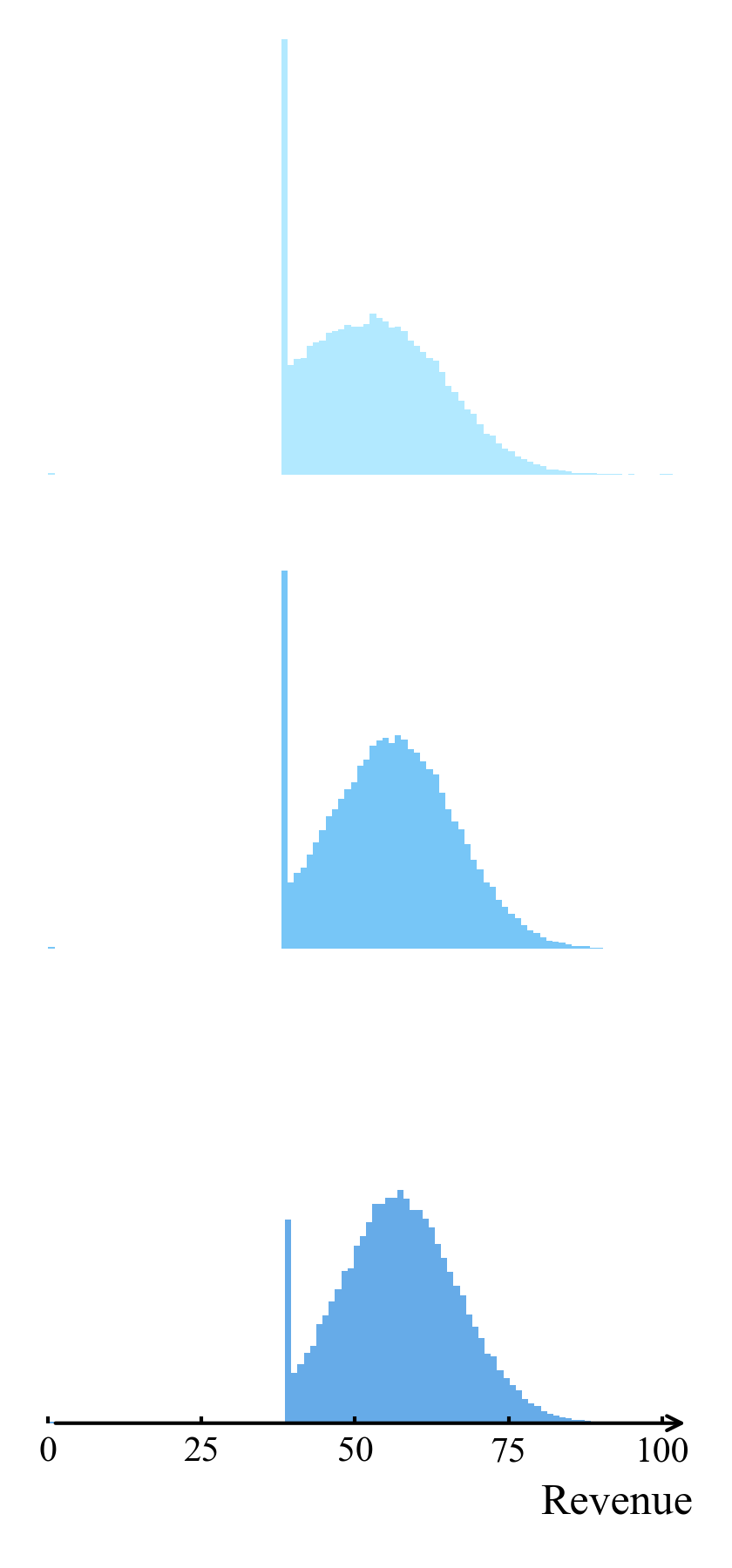}
        \label{fig_norm_distribution1}
    }
    \subfigure[$Exp(0.08)$]{
        \includegraphics[width=0.24\textwidth]{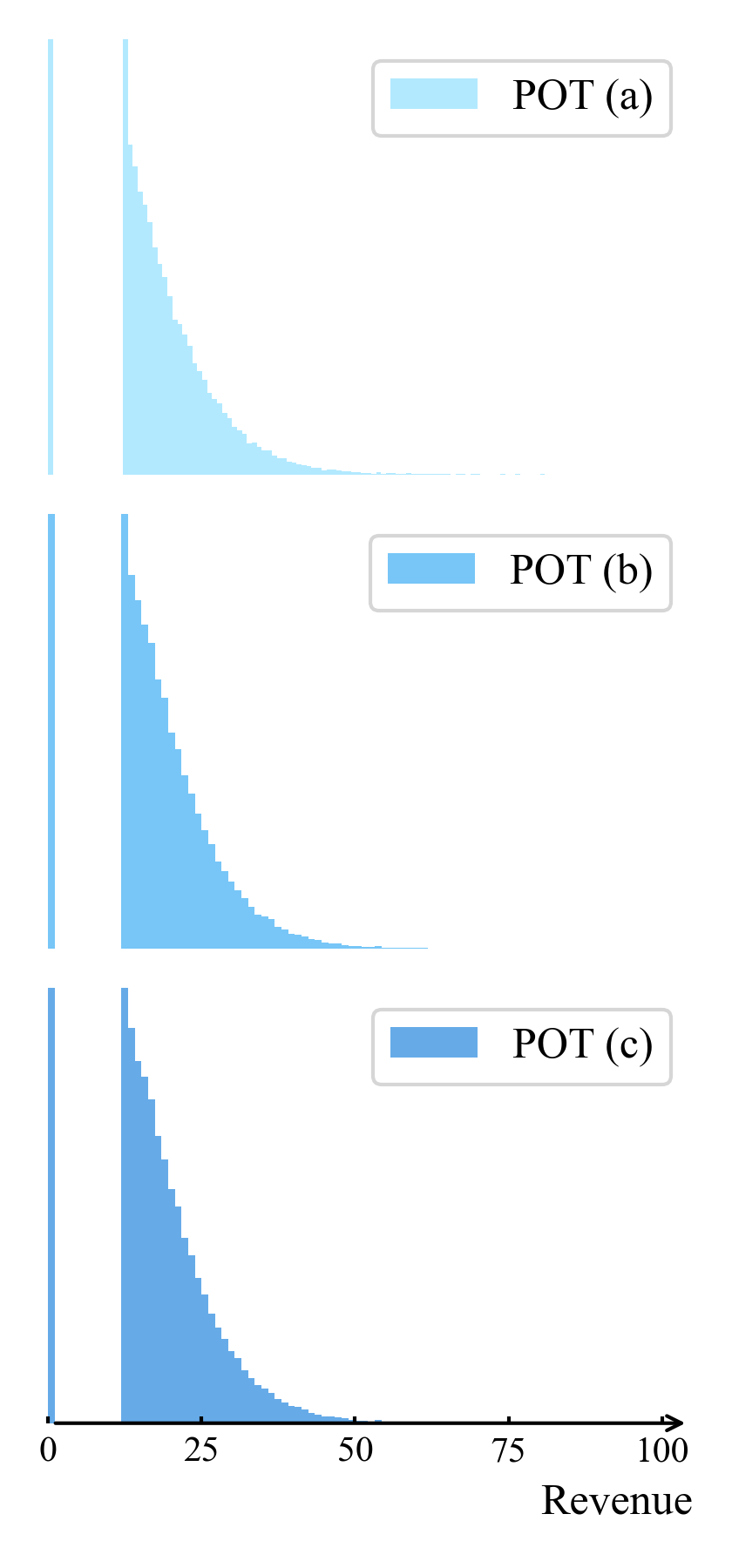}
        \label{fig_exp_distribution1}
    }
    \caption{Revenue distributions of \texttt{APX-R} in POT trees with different levels of symmetry.}
    \label{fig_symmetric_rev}
\end{figure}

Table \ref{table_symmetric} presents the expected revenues of the \texttt{APX-R} mechanism using Myerson’s reserve price in the three POT tree structures. It shows that as POT tree symmetry increases, expected revenue steadily rises. Therefore, POT tree symmetry plays a crucial role in determining the expected revenue of auction mechanisms, and more symmetric POT structures significantly improve the seller’s expected revenue.

\linespread{1.}
{   
    \begin{table}[htbp]
		\centering
        \caption{Expected revenue of the \texttt{APX-R} mechanism in POT tree structures with different levels of symmetry.}
        \addtolength{\leftskip} {-2cm}
        \addtolength{\rightskip}{-2cm}
		\begin{tabular}{lrrr}
		  \hline\hline
		  Distributions & \cellcolor{gray!10}POT (a) & \cellcolor{gray!20}POT (b) & \cellcolor{gray!30}POT (c)   \\
            \hline
            $U[0,100]$ & \cellcolor{gray!10}$59.4687$ &\cellcolor{gray!20}$64.8785$ & \cellcolor{gray!30}\textbf{66.5007} \\
		  \hline
		  $N(50,16.67^2)$ & \cellcolor{gray!10}$50.8081$ & \cellcolor{gray!20}$55.7984$ & \cellcolor{gray!30}\textbf{57.1788}  \\ 
            \hline
            $Exp(0.08)$ & \cellcolor{gray!10}$14.3518$ & \cellcolor{gray!20}$15.8711$ & \cellcolor{gray!30}\textbf{16.3513} \\
            \hline\hline
		\end{tabular}
		\label{table_symmetric}
	\end{table}}

\subsection{Classic Small-Scale Network Structure}

The network illustrated below represents a simple small-scale social network that exemplifies the theory of \textit{Six Degrees of Separation} that all people are six or fewer social connections away from each other \citep{amaral2000classes}. Although it is a toy network, this network highlights the high connectivity of social structures, which plays a crucial role in network expansion and information diffusion. For the \texttt{APX-R} mechanism, this small network presents a significant challenge. It extends a traditional two-buyer auction setting, which is often considered a highly limited market with low competition. The primary difficulty lies in incentivizing buyers to propagate auction info and invite additional participants, thereby improving seller revenue. Although the market is relatively small, the network’s longest path poses a critical constraint on auction info diffusion. The efficiency of information diffuses within the network is crucial for market expansion.

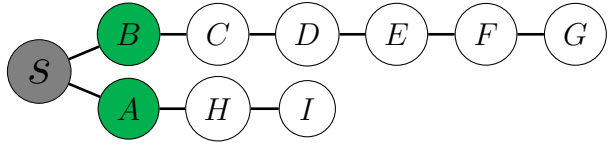
\begin{figure}[ht]
    \centering
    \begin{tikzpicture}[scale=0.99, box/.style={circle, draw}]
        \node[box,fill=gray,very thin](s) at(0,-0.5){{\LARGE$s$}};
        \node[box,fill={rgb:red,1;green,180;blue,80}](A) at(1.2,-1){$A$};
        \node[box,fill=white](H) at(2.4,-1){$H$};
        \node[box,fill=white](I) at(3.6,-1){$I$};

        \node[box,fill={rgb:red,1;green,180;blue,80}](B) at(1.2,0){$B$};
        \node[box,fill=white](C) at(2.4,0){$C$};
        \node[box,fill=white](D) at(3.6,0){$D$};
        \node[box,fill=white](E) at(4.8,0){$E$};
        \node[box,fill=white](F) at(6.0,0){$F$};
        \node[box,fill=white](G) at(7.2,0){$G$};

        \draw[line width=1pt] (s) --(A);
        \draw[line width=1pt] (A) --(H);
        \draw[line width=1pt] (H) --(I);

        \draw[line width=1pt] (s) --(B);
        \draw[line width=1pt] (B) --(C);
        \draw[line width=1pt] (C) --(D);
        \draw[line width=1pt] (D) --(E);
        \draw[line width=1pt] (E) --(F);
        \draw[line width=1pt] (F) --(G);
    \end{tikzpicture}
    \caption{Classic small-scale network structure.}
    \label{fig_small_net}
\end{figure}

Experiments were conducted in the classic small-scale network to evaluate the performance of the \texttt{APX-R} mechanism under different reserve price settings. The expected revenue was compared with the theoretical upper bound $\operatorname{OPT}(N)$, the theoretical lower bound $\operatorname{MYS}(e_s)$, and two other diffusion auction mechanisms designed to optimize seller revenue: \texttt{CWM} \cite{zhang2023optimal} and \texttt{maxViVa} \cite{bhattacharyya2023optimal}.

Since the estimated $\underline{k}$ directly impacts the reserve price in \texttt{APX-R}, and the smallest subtree size in the POT tree is 3, four different reserve price settings were analyzed: (1) no reserve price, which degenerates to \texttt{IDM}\citep{li2017mechanism,li2022diffusion}; (2) reserve price for the classical Myerson optimal auction; (3) $\gamma$ resulting from a conservative setting of $\underline k$; and (4) $\gamma$ resulting from the accurate $\underline k$. 

For \texttt{CWM}, all losers have no reward, so they have very weak incentive to diffuse. For \texttt{maxViVa} it's proved in Appendix that it's not DSIC in terms of diffusion. In contrast, \texttt{APX-R} provides strong propagation incentives to all buyers, ensuring full information diffusion. We set that all \texttt{CWM} agents diffuse with probability $0.5$, and that each \texttt{maxViVa} agent diffuses to $80\%$ of her neighbors.

It can be seen that \texttt{APX-R} always performs well, while the accurate $\gamma$ produces a revenue that is closest to the maximum possible revenue from the table \ref{table_small_market_regular}. This is because the accurate $\underline k$ gives a best approximation to the optimal reserve price $r_{opt}$ while still preserving the incentive compatibility. In contexts beyond this toy network, \texttt{APX-R} will only perform better,
particularly in more realistic shallow networks and scenarios where the seller is connected to a greater number of neighbors.
Each cell represents the expected revenue obtained under the specified mechanism and valuation distribution.

\linespread{1.}
{   
    \begin{table}[htbp]
		\centering
        \caption{Expected revenue of different diffusion auction mechanisms in the classic small-scale network under various regular distributions.}
        \addtolength{\leftskip} {-2cm}
        \addtolength{\rightskip}{-2cm}
        \scalebox{0.85}{
		\begin{tabular}{l|lrrr}
		  \toprule
		  \multicolumn{2}{c}{\textbf{Mechanisms}}  & 
            \multicolumn{3}{c}{\textbf{Expected Revenue}} \\
            \hline
            \rowcolor{blue!10}
            \multicolumn{2}{c}{\textbf{Distributions}} & $U[0,100]$ &$N(50,16.67^2)$ & $Exp(0.08)$ \\
		  \hline\hline
		  \multirow{4}{*}{\texttt{APX-R}} & $\gamma=0$ {\small (No reserve, \texttt{IDM})} & $70.7041$ & $60.4643$ & $18.1708$ \\ 
            \cdashline{2-5}[1pt/1pt] & {\small$\underline{k}=1$ (Myerson reserve)} & $72.2705$ & $60.5352$ & $19.1526$\\
            \cdashline{2-5}[1pt/1pt] & \small{$\underline{k}=2$ (cautious estimate)} & $73.3333$ & $60.7732$  & $19.6865$\\ 
            \cdashline{2-5}[1pt/1pt] & 
            \cellcolor{gray!20}\small{$\underline{k}=3$ (accurate estimate)} & \cellcolor{gray!20}\textbf{74.1242} &  \cellcolor{gray!20} \textbf{61.0995} &\cellcolor{gray!20}  \textbf{20.0335}  \\ 
            \hline
            \multirow{2}{*}{\small{Others}} & \texttt{maxViVa} & $65.5138$ & $56.4292$ &$16.2480$ \\ 
            \cdashline{2-5}[1pt/1pt]
            & \texttt{CWM} & $63.1984$ & $53.6910$ & $15.8095$ \\
            \hline
            \multirow{2}{*}{\small{Benchmarks}} & \small{$\operatorname{MYS}(e_s)$} & $41.6667$ & $42.7858$ & $8.3508$ \\ 
            \cdashline{2-5}[1pt/1pt]
            & \small{$\operatorname{OPT}(N)$} & $80.0195$ & $65.5407$ & $22.8856$ \\ 
		  \bottomrule
		\end{tabular}}
		\label{table_small_market_regular}
	\end{table}%
}

\subsection{Different Market Metrics}

We have also carried out a more detailed simulation. We check the performance of the \texttt{\texttt{APX-R}} mechanism under two criteria: \textit{market expansion ratio} and \textit{market depth}. The simulations are all based on the network composed of $1$ seller and $9$ potential buyers. 

\subsubsection{Market Expansion Ratio}
The market expansion ratio serves as a metric for measuring the proportion between the size of the traditional market and that of the networked market. A smaller ratio indicates greater potential for market expansion through the network, whereas a larger ratio suggests that the original market is more similar to the networked market, implying lower expansion potential. This proportion is formally defined as the market expansion ratio.

\begin{definition}\label{definition_MER}
The \textit{market expansion ratio}, denoted as $mer$, measures market growth from auction information diffusion. It is the ratio of the initial market size to the final market size after full propagation: $mer = \frac{\rho+1}{n+1}$,
where $\rho$ is the number of initial buyers directly connected to the seller, and $n$ is the total number of buyers after propagation.
\end{definition}

Although all auction simulations assume a market with $1$ seller and $9$ buyers, the network structures differ based on market expansion ratios. For each scenario with a different $mer$, the selected network structure follows the most symmetric tree structure. Figure \ref{figure_market_mer} illustrates the simulation settings for markets considering different values of $mer$. The investigation starts from $mer = 30\%$, since according to the six degrees of separation theory, the maximum distance between the seller and any buyer does not exceed six hops. As a result, a market with 10 participants $mer = 20\%$ is not feasible within this framework.

\begin{figure}[h]
\centering
\includegraphics[scale = 0.9]{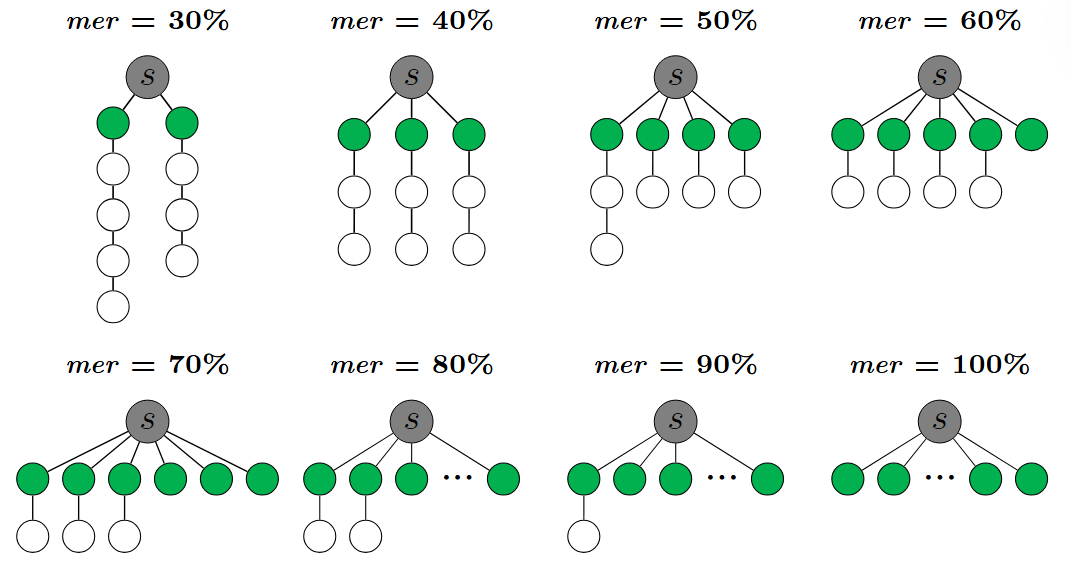}
\caption{Simulation scenarios for a market with $1$ seller and $9$ buyers under different market expansion ratios.}
\label{figure_market_mer}
\end{figure}

Table \ref{Table_market_portion_regular} shows expected revenues  of \texttt{APX-R} under different ratios of $mer$. Each decimal is the revenue of the mechanism (the row) running under the market scenario (the column). The empty grids in this table are where the corresponding $\gamma$ value is overestimated, which cannot happen with our method. 
The black numbers in each column are the highest revenues under each market scenario. We can see that the \texttt{APX-R} mechanism with an accurate reserve price $\gamma$ always performs best, even in the network market, which is not so different from the classical market (e.g. $mer=90\%$ or $100\%$). \texttt{APX-R} with conservative $\gamma$ still have a good performance against the currently existing optimal diffusion auctions \texttt{maxViVa} and \texttt{CWM}. More importantly, the most significant aspect of \texttt{APX-R} is that it has robust performance regardless of how the network market differs from the classical market. 

\begin{table}[htbp]
    \centering
    \addtolength{\leftskip} {-3cm}
    \addtolength{\rightskip}{-3cm}
    \scalebox{0.8}{
    \begin{tabular}{l|l|rrrrrrrr}
		\hline\hline
        \rowcolor{green!8}
        \multicolumn{2}{c|}{\quad} & \multicolumn{8}{c}{\textbf{Market Expansion Ratios}}   \\
		\hline
        \rowcolor{green!8}
		\multicolumn{2}{c|}{\textbf{Mechanisms}}         &30\% &40\% &50\% &\textbf{60\%} &\textbf{70\%} &\textbf{80\%} &\textbf{90\%} &\textbf{100\%}\\
		\hline\hline
        \rowcolor{blue!10}
        \multicolumn{10}{c}{\textbf{Uniform Dsitribution} $U[0,100]$} \\
        \multirow{4}{*}{\texttt{APX-R}} &$\gamma=0$      &$73.3333$ & $77.1418$ &$78.2043$ &$78.8778$ &$79.1567$ &$79.4344$ &$79.7112$ &$80.0000$\\ \cdashline{2-10}[1pt/1pt]
        &\small{$\gamma$ of $\underline{k}=1$} & $74.1113$ &$77.3605$ & $78.3454$ &\cellcolor{gray!20}\textbf{78.9692} &\cellcolor{gray!20}\textbf{79.2312} &\cellcolor{gray!20}\textbf{79.4944} &\cellcolor{gray!20}\textbf{79.7569} &\cellcolor{gray!20}\textbf{80.0195}\\ 
        \cdashline{2-10}[1pt/1pt]
        &\small{$\gamma$ of $\underline{k}=2$}  &$74.7809$ & $77.5655$ & \cellcolor{gray!20}\textbf{78.4597} &$-$ &$-$ &$-$ &$-$ &$-$\\ 
        \cdashline{2-10}[1pt/1pt]
        & \small{$\gamma$ of $\underline{k}=3$} & \cellcolor{gray!20}\textbf{75.2765} &\cellcolor{gray!20}\textbf{77.6491} &$-$ &$-$ &$-$ &$-$ &$-$ &$-$\\
        \hline
        \multirow{2}{*}{Others} &\texttt{maxViVa}           &$66.4026$ & $72.7807$ & $75.6555$ &$77.1544$ &$77.9283$ &$78.6896$ &$79.3832$ &$80.0195$\\ 
        \cdashline{2-10}[1pt/1pt]
        &\texttt{CWM} &$57.4405$ & $68.1743$ &$73.4132$ &$76.3581$ &$76.7892$ &$77.2172$ &$79.3694$ &$80.0195$\\
        \hline\hline
        \rowcolor{blue!10}
        \multicolumn{10}{c}{\textbf{Norm Distribution} $N(50,16.67^2)$} \\
		\multirow{4}{*}{\texttt{APX-R}} &$\gamma=0$      &$61.7839$ & $63.8564$ &$64.4753$ &$64.8691$ &$65.0370$ &$65.2040$ &$65.3713$ &$65.5406$\\ \cdashline{2-10}[1pt/1pt]
        &{\small $\gamma$ of $\underline{k}=1$} & $61.7990$ &$63.8581$ & $64.4762$ &\cellcolor{gray!20} \textbf{64.8698} &\cellcolor{gray!20} \textbf{65.0376} &\cellcolor{gray!20} \textbf{65.2044} &\cellcolor{gray!20} \textbf{65.3716} &\cellcolor{gray!20} \textbf{65.5407}\\ 
        \cdashline{2-10}[1pt/1pt]
        &{\small $\gamma$ of $\underline{k}=2$} &$61.8972$ & $63.8762$ & \cellcolor{gray!20} \textbf{64.4847} &$-$ &$-$ &$-$ &$-$ &$-$\\ \cdashline{2-10}[1pt/1pt]
        & {\small $\gamma$ of $\underline{k}=3$}          &\cellcolor{gray!20} \textbf{62.0203} &\cellcolor{gray!20} \textbf{63.8968} &$-$ &$-$ &$-$ &$-$ &$-$ &$-$\\
        \hline
        \multirow{2}{*}{Others} &\texttt{maxViVa}           &$56.8625$ & $61.1056$ & $62.8087$ &$63.8117$ &$64.2438$ &$64.7097$ &$65.1040$ &$65.5407$\\ 
        \cdashline{2-10}[1pt/1pt]
        &\texttt{CWM} &$48.8079$ & $57.7087$ &$61.3730$ &$63.2310$ &$63.4989$ &$63.7707$ &$63.8970$ &$65.5407$\\
		\hline\hline
        \rowcolor{blue!10}
        \multicolumn{10}{c}{\textbf{Exponential Distribution} $Exp(0.08)$} \\
		\multirow{4}{*}{\texttt{APX-R}} &$\gamma=0$      &$19.2069$ & $21.1406$ &$21.7613$ &$22.1559$ &$22.3330$ &$22.5040$ &$22.6795$ &$22.8534$\\ \cdashline{2-10}[1pt/1pt]
        &{\small $\gamma$ of $\underline{k}=1$} & $19.8243$ &$21.3753$ & $21.9135$ &\cellcolor{gray!20}\textbf{22.2580} &\cellcolor{gray!20}\textbf{22.4172} &\cellcolor{gray!20}\textbf{22.5709} &\cellcolor{gray!20}\textbf{22.7291} &\cellcolor{gray!20}\textbf{22.8856}\\ \cdashline{2-10}[1pt/1pt]
        &{\small $\gamma$ of $\underline{k}=2$} &$20.2147$ & $21.5223$ & \cellcolor{gray!20}\textbf{21.9908} &$-$ &$-$ &$-$ &$-$ &$-$\\ 
        \cdashline{2-10}[1pt/1pt]
        & {\small$\gamma$ of $\underline{k}=3$}          &\cellcolor{gray!20}\textbf{20.4610} &\cellcolor{gray!20}\textbf{21.5704} &$-$ &$-$ &$-$ &$-$ &$-$ &$-$\\
        \hline
        \multirow{2}{*}{Others} &\texttt{maxViVa}           &$16.6697$ & $19.2450$ & $20.4720$ &$21.2398$ &$21.6671$ &$22.1203$ &$22.5670$ &$22.8856$\\ 
        \cdashline{2-10}[1pt/1pt]
        &\texttt{CWM} &$14.1591$ & $17.4215$ &$19.4600$ &$20.8695$ &$21.0793$ &$21.3041$ &$22.1779$ &$22.8856$\\
		\hline\hline
		\end{tabular}}
        \caption{Impact of different $mer$ on revenue across mechanisms.}
        \label{Table_market_portion_regular}
	\end{table}

It is worth noting that in scenarios with higher market expansion ratios, the \texttt{APX-R} mechanism can still achieve significant revenue gains even without a reserve price (i.e., $\gamma = 0$). This is because in such cases, the seller already has access to a larger initial pool of buyers, making the impact of market expansion on revenue growth more substantial than the effect of setting a reserve price. Although the incremental benefit of a reserve price diminishes in market where the seller initially knows more buyers, accurately estimating the reserve price still allows \texttt{APX-R} to further enhance revenue beyond an already high baseline. Thus, the emphasis should be on optimizing seller revenue at high market expansion levels. Particularly in large-scale markets, refining the reserve price can further drive overall revenue growth.
\subsubsection{Market Depth}

Market depth $md$ refers to the longest path of information diffusion, equivalent to the graph depth. According to the \textit{Six Degrees of Separation}, real-world networks typically have $md \leq 6$. Thus, this experiment only considers cases where $md \leq 6$. In larger markets with greater $md$, the distance from the seller to the farthest buyer increases. When the winning buyer is at the maximum depth, more DCNs exist along the path, potentially leading to higher propagation rewards.

\begin{definition}\label{definition_MD}
    The market depth of an auction network, denoted as $md$, is defined as the longest distance from the seller to any buyer, representing the maximum information diffusion path.
\end{definition}
In this market depth simulation, the market still consists of $1$ seller and $9$ buyers. For each depth level, the most symmetric tree structure with the minimum number of direct neighbors for the seller is used. Figure \ref{figure_market_md} shows the simulated market structures considering different values of $md$.

\begin{figure}[h]
\centering
\includegraphics[scale = 0.9]{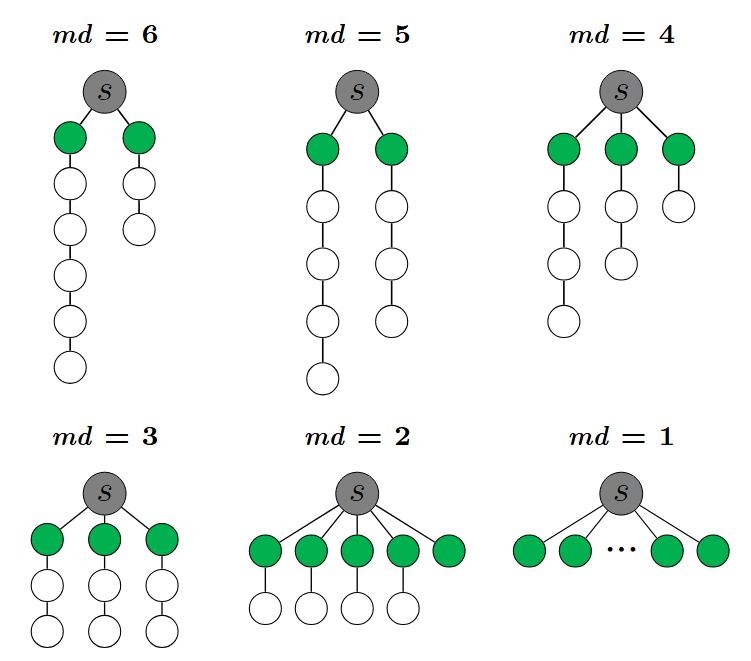}
\caption{Simulation scenarios for a market with $1$ seller and $9$ buyers under different market depths.}
\label{figure_market_md}
\end{figure}

Table \ref{Table_market_depth_regular} shows the performance of mechanisms in markets with different market depths. In this simulation of market depth, \texttt{APX-R} also performs well in all example settings (all columns).  From these two tables it can be seen that the \texttt{APX-R} mechanism shows its advantage and robustness when confronted with different possible variations of network structures. 

\begin{table}[htbp]
    \centering
    \addtolength{\leftskip} {-2cm}
    \addtolength{\rightskip}{-2cm}
    \scalebox{0.85}{
    \begin{tabular}{l|l|rrrrrr}
		\hline\hline  
        \rowcolor{green!8}
        \multicolumn{2}{c|}{\quad} & \multicolumn{6}{c}{\textbf{Market Depths}}   \\
        \hline
        \rowcolor{green!8}
		\multicolumn{2}{c|}{\textbf{Mechanisms}}  &$\mathbf{md}$\textbf{ = 6}&$\mathbf{md}$\textbf{ = 5} &$\mathbf{md}$\textbf{ = 4} &$\mathbf{md}$\textbf{ = 3}&$\mathbf{md}$\textbf{ = 2} &$\mathbf{md}$\textbf{ = 1}  \\
		\hline\hline
        \rowcolor{blue!10}
        \multicolumn{8}{c}{\textbf{Uniform Distribution} $U[0,100]$} \\
        \multirow{4}{*}{\texttt{APX-R}} &$\gamma=0$ &$70.7041$ & $73.3333$ &$76.5476$ &$77.1418$ &$78.8778$ &$80.0000$ \\ 
        \cdashline{2-8}[1pt/1pt]
        &\small{$\gamma$ of$\underline{k}=1$} &$72.2705$ &$74.1113$ & $76.8513$ &$77.3605$ &\cellcolor{gray!20}\textbf{78.9692} &\cellcolor{gray!20}\textbf{80.0195}\\  
        \cdashline{2-8}[1pt/1pt]
        & \small{$\gamma$ of$\underline{k}=2$} &$73.3333$ & $74.7809$ & \cellcolor{gray!20}\textbf{77.1301} &$77.5655$ &$-$ &$-$\\ \cdashline{2-8}[1pt/1pt]
        &\small{$\gamma$ of$\underline{k}=3$} &\cellcolor{gray!20}\textbf{74.1242} &\cellcolor{gray!20}\textbf{75.2765} &$-$ &\cellcolor{gray!20}\textbf{77.6491} &$-$ &$-$\\
        \hline
        \multirow{2}{*}{Others} &\texttt{maxViVa}  &$65.5138$ & $66.4026$ & $72.2899$ &$72.7807$ &$77.1544$ &$80.0195$\\ 
        \cdashline{2-8}[1pt/1pt]
        & \texttt{CWM}  &$63.1984$ & $57.740$ &$71.0953$ &$68.174$ &$76.358$ &$80.0195$\\
        \hline\hline
        \rowcolor{blue!10}
        \multicolumn{8}{c}{\textbf{Norm Distribution}  $ N(50,16.67^2)$} \\
		\multirow{4}{*}{\texttt{APX-R}} &$\gamma=0$ &$60.4643$ & $61.7839$ &$63.5359$ &$63.8564$ &$64.8694$ &$65.5406$ \\ 
        \cdashline{2-8}[1pt/1pt]
        &{\small$\gamma$ of $\underline{k}=1$} &$60.5352$ &$61.7990$ & $63.5391$ &$63.8581$ &\cellcolor{gray!20}\textbf{64.8698} &\cellcolor{gray!20}\textbf{65.5407}\\  \cdashline{2-8}[1pt/1pt]
        & {\small$\gamma$ of $\underline{k}=2$} &$60.7732$ & $61.8972$ &\cellcolor{gray!20} \textbf{63.5691} &$63.8762$ &$-$ &$-$\\ 
        \cdashline{2-8}[1pt/1pt]
        &{\small$\gamma$ of $\underline{k}=3$} &\cellcolor{gray!20}\textbf{61.0995} &\cellcolor{gray!20}\textbf{62.0203} &$-$ &\cellcolor{gray!20}\textbf{63.8968} &$-$ &$-$\\
        \hline
        \multirow{2}{*}{Others} &\texttt{maxViVa}  &$56.4292$ & $56.8625$ & $60.7927$ &$61.1056$ &$63.8117$ &$65.5407$\\ 
        \cdashline{2-8}[1pt/1pt]
        & \texttt{CWM}  &$53.6910$ & $48.8079$ &$59.8231$ &$57.7087$ &$63.2310$ &$65.5407$\\
		\hline\hline
        \rowcolor{blue!10}
        \multicolumn{8}{c}{\textbf{Exponential Distribution} $Exp(0.08)$} \\
        \multirow{4}{*}{\texttt{APX-R}} &$\gamma=0$ &$18.1708$ & $19.2069$ &$20.8474$ &$21.1406$ &$22.1559$ &$22.8534$ \\ 
        \cdashline{2-8}[1pt/1pt]
        &{\small$\gamma$ of $\underline{k}=1$} &$19.1526$ &$19.8243$ & $22.1414$ &$21.3753$ &\cellcolor{gray!20}\textbf{22.2580} &\cellcolor{gray!20}\textbf{22.8856}\\  \cdashline{2-8}[1pt/1pt]
        & {\small$\gamma$ of $\underline{k}=2$} &$19.6865$ & $20.2147$ & \cellcolor{gray!20}\textbf{21.3263} &$21.5223$ &$-$ &$-$\\ 
        \cdashline{2-8}[1pt/1pt]
        &{\small$\gamma$ of $\underline{k}=3$} &\cellcolor{gray!20}\textbf{20.0335} &\cellcolor{gray!20}\textbf{20.4610} &$-$ &\cellcolor{gray!20}\textbf{21.5704} &$-$ &$-$\\
        \hline
        \multirow{2}{*}{Others} &\texttt{maxViVa}  &$16.2480$ & $16.6697$ & $18.9024$ &$19.2450$ &$21.2398$ &$22.8856$\\ 
        \cdashline{2-8}[1pt/1pt]
        & \texttt{CWM}  &$15.8095$ & $14.1591$ &$18.5479$ &$17.4215$ &$20.8695$ &$22.8856$\\
	\hline\hline
		\end{tabular}}
        \caption{Impact of different $mer$ on revenue across mechanisms.}
        \label{Table_market_depth_regular}
	\end{table}

\subsection{Large Real Social Networks}

To demonstrate that under different values of $\rho$, the \texttt{APX-R} mechanism with the proposed reserve price function achieves higher expected revenue than directly running Myerson's optimal auction on the seller's initial neighbor set, i.e., $\operatorname{APX(\mathbf{t'};\gamma)} > \operatorname{MYS(e_s)}$, experiments were conducted using the FilmTrust social trust network dataset provided by KONECT\footnote{\url{http://www.konect.cc/networks/librec-filmtrust-trust/}} \citep{rossi2015network}. This dataset consists of 874 nodes and 1853 edges. For each $\rho$ value, a randomly selected node in the network satisfying the given $\rho$ was designated as the seller for the experiments. 

Figure \ref{figure_large_net_com_myerson} illustrates the expected revenue comparison between the \texttt{APX-R} mechanism and Myerson’s optimal auction in the original market as $\rho$ varies, under different regular valuation distributions. The results show that across all four valuation distributions, the \texttt{APX-R} mechanism (pink line) consistently outperforms the expected revenue of Myerson’s optimal auction in the original market $\operatorname{MYS}(e_s)$ (blue line). Notably, for larger values of $\rho$, the expected revenue of \texttt{APX-R} approaches the theoretical upper bound. This demonstrates that in large-scale real-world social networks, the \texttt{APX-R} mechanism significantly outperforms Myerson’s auction in maximizing the seller’s expected revenue.

\begin{figure}[h]
\centering
\includegraphics[scale = 0.5]{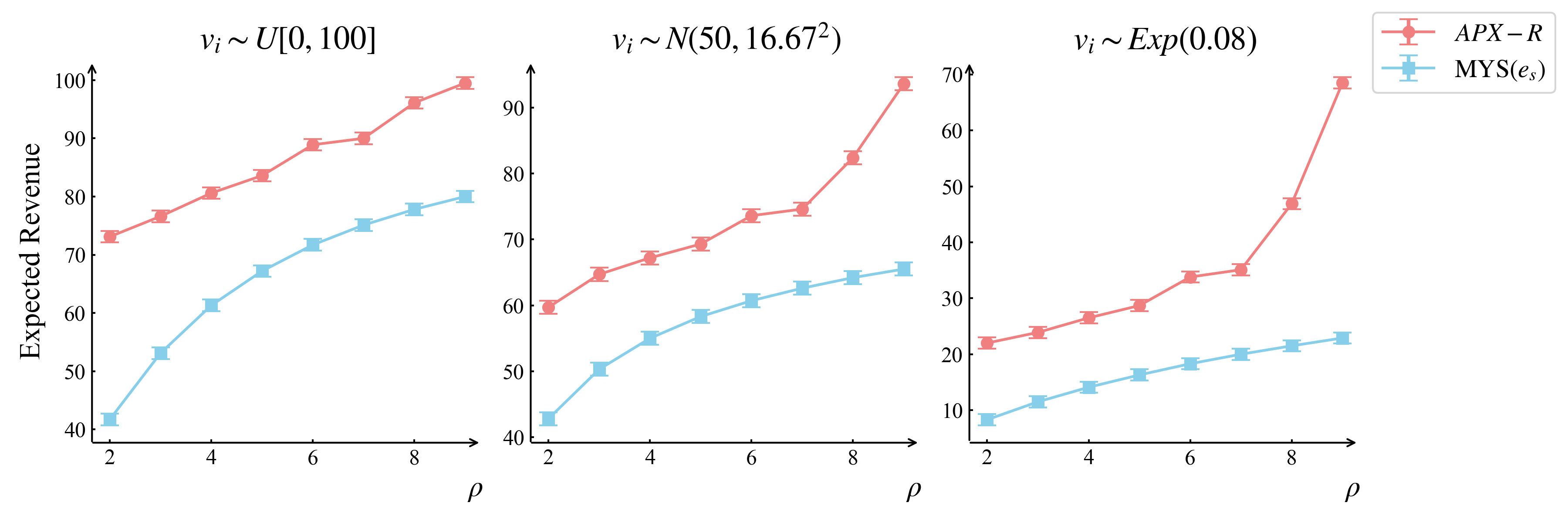}
\caption{Expected revenue comparison of \texttt{APX-R} and Myerson’s optimal auction in a large real-world network under different $\rho$ values.}
\label{figure_large_net_com_myerson}
\end{figure}

To examine the performance differences of the \texttt{APX-R} mechanism in sparse and dense social networks, simulations were conducted on two datasets: LastFM social network\footnote{\url{https://snap.stanford.edu/data/feather-lastfm-social.html}} and Facebook social network\footnote{\url{https://snap.stanford.edu/data/ego-Facebook.html}}. Both datasets are undirected graphs sourced from the SNAP (Stanford Large Network Dataset Collection) \citep{leskovec2014snap}. These two networks represent contrasting topological structures, with LastFM being relatively sparse and Facebook being more dense. Table \ref{table_bignet_property} summarizes key properties of both networks, where density is the ratio of actual edges to the maximum possible number of edges, ranging from $[0, 1]$. And transitivity is a measure of how likely adjacent nodes are to be interconnected, often used to evaluate clustering in social networks. High transitivity suggests that "friends of friends" are more likely to be directly connected.

\linespread{1.}
{   
    \begin{table}[htbp]
		\centering
        \addtolength{\leftskip} {-2cm}
        \addtolength{\rightskip}{-2cm}
        \scalebox{0.8}{
		\begin{tabular}{lrrrr}
		  \toprule
		  Datasets & Nodes & Edges & Density & Transitivity\\
            \midrule
            LastFM &  7624 & 27806 & 0.0001 & 0.179 \\
		  \midrule
		  Facebook &  4039 & 88234 & 0.0011 & 0.519 \\
          \bottomrule
		\end{tabular}}
        \caption{Parameter comparison of LastFM and Facebook network.}
		\label{table_bignet_property}
	\end{table}}

\begin{figure}[htbp]
    \centering
    \subfigure[$\rho =2$]{
        \centering
        \includegraphics[height=4cm]{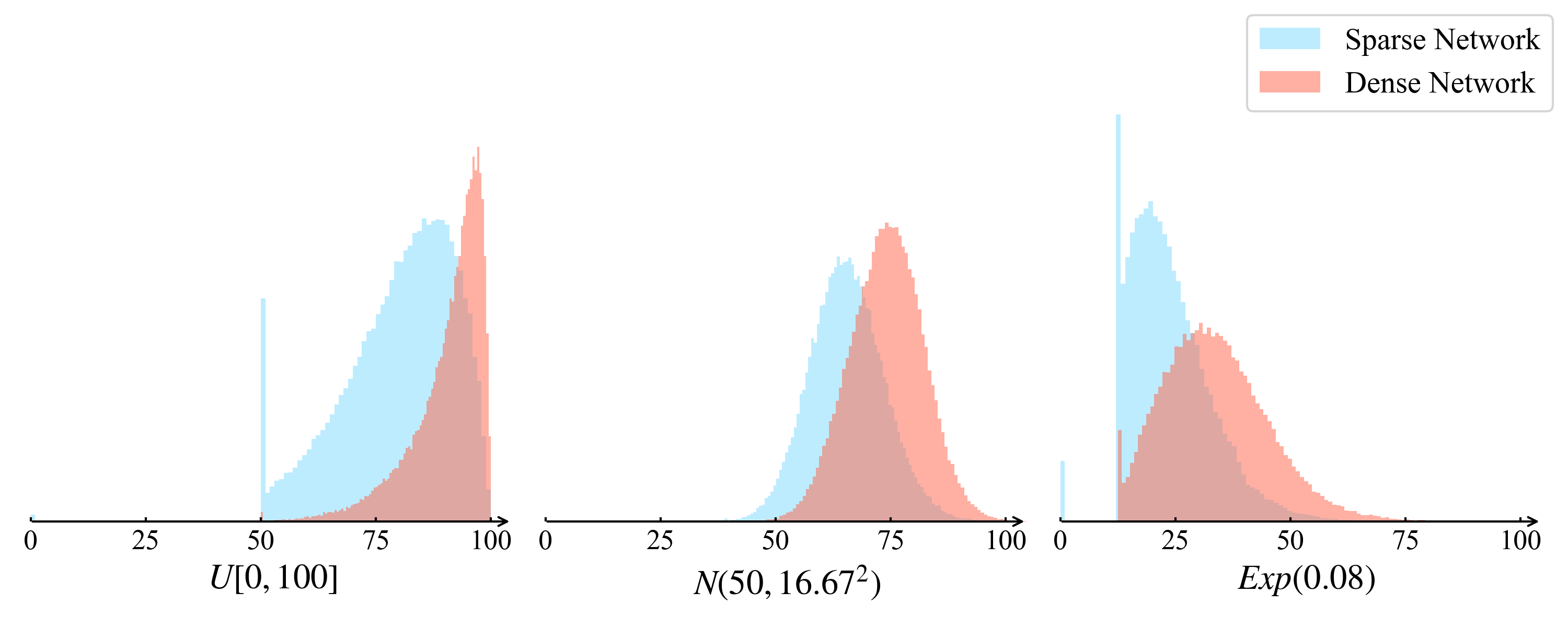}
        \label{fig_bignet_2deg}}

    \subfigure[$\rho =3$]{
        \centering
        \includegraphics[height=4cm]{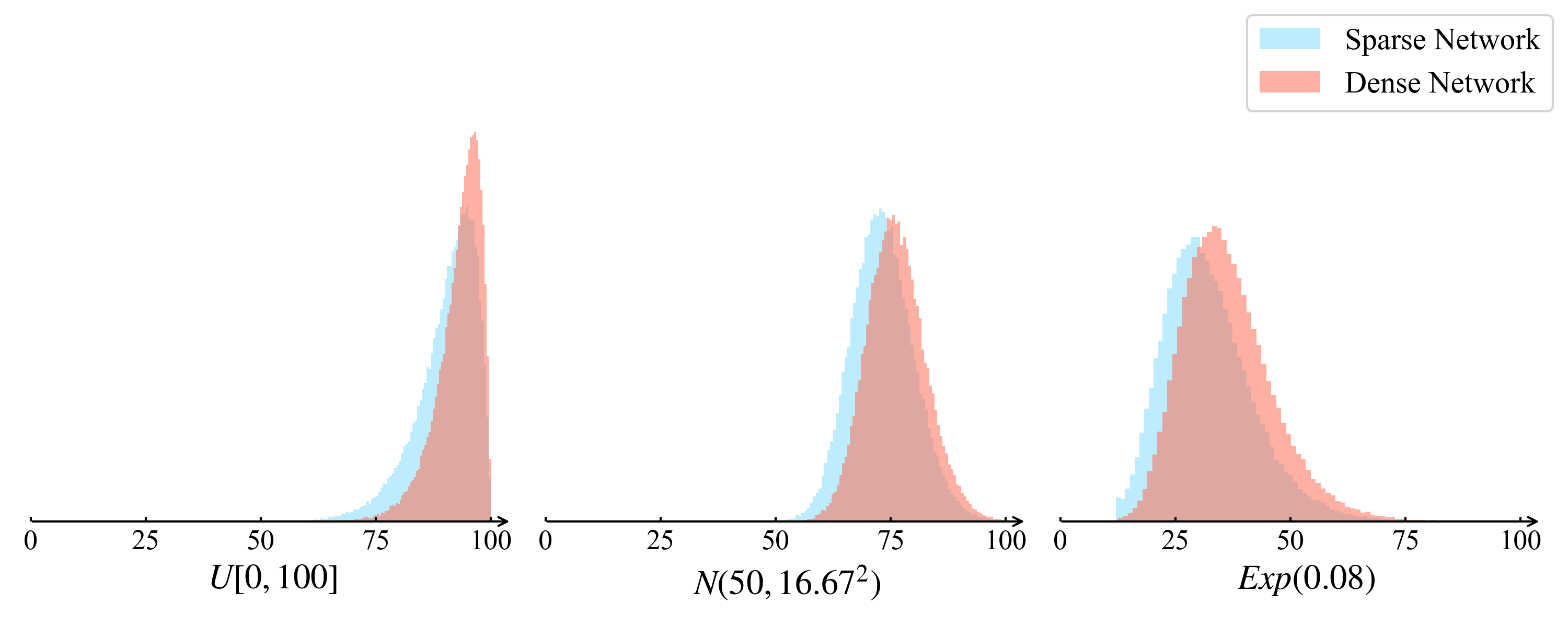}
        \label{fig_bignet_3deg}}

    \subfigure[$\rho =4$]{
        \centering
        \includegraphics[height=4cm]{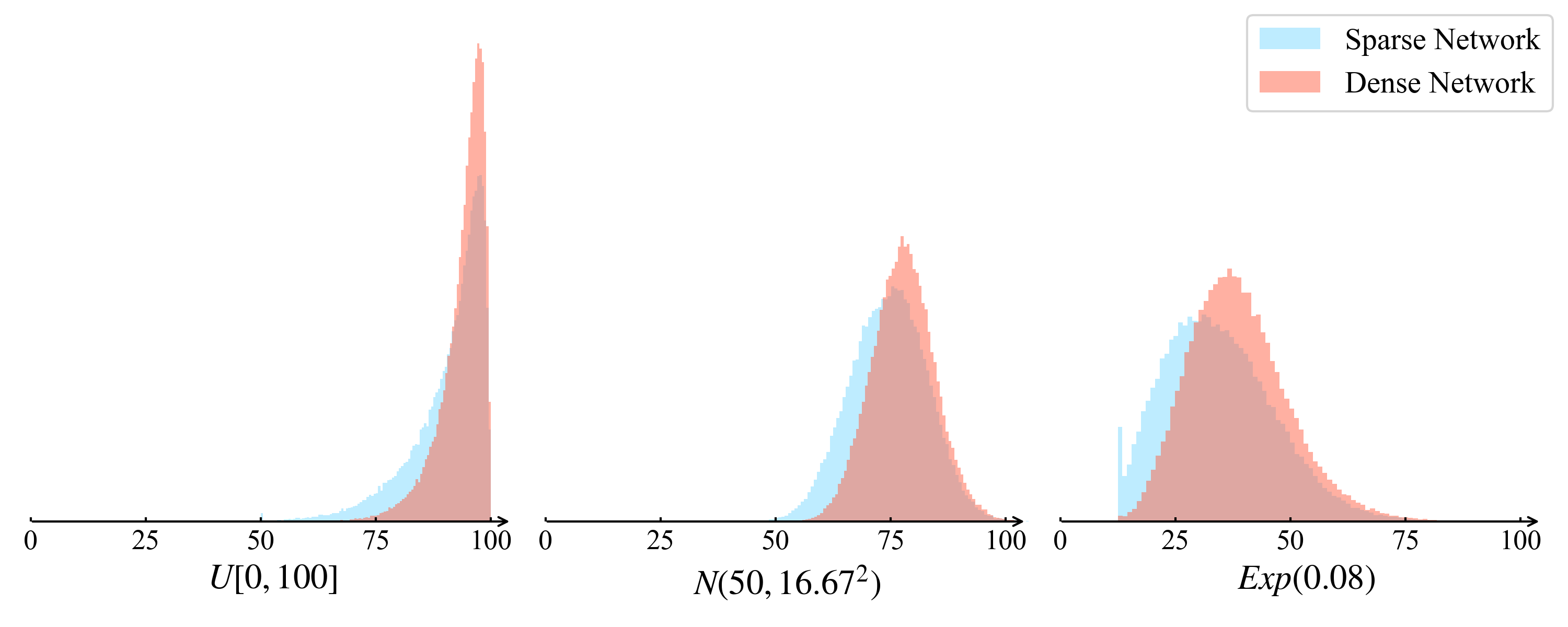}
        \label{fig_bignet_4deg}}

    \caption{Comparison of revenue distributions of \texttt{APX-R} in sparse and dense social networks.}
    \label{fig_bignet_xishumiji_rev}
\end{figure}
Figure \ref{fig_bignet_xishumiji_rev} illustrates the revenue distribution of the \texttt{APX-R} mechanism in sparse and dense social networks under different valuation distributions. The experiments consider cases where the number of seller's neighbors in the original market $\rho$ is 2, 3 and 4. The blue and red histograms represent revenue distributions in sparse and dense social networks, respectively. From Figure \ref{fig_bignet_xishumiji_rev}, it is evident that revenue distributions shift to the right in dense networks compared to sparse networks. That is, in dense networks, revenue is more frequently concentrated in higher ranges, indicating that sellers generally achieve higher revenues in denser social networks. Moreover, as the seller's number of neighbors $\rho$ increases, the revenue distribution in dense networks further concentrates in higher ranges. 

Table \ref{table_bignet_xishumiji_rev} presents the expected revenues of the \texttt{APX-R} mechanism in sparse and dense networks. It is evident that the expected revenue in the dense social network is consistently higher than in the sparse network. 

\linespread{1.}
{   
    \begin{table}[htbp]
		\centering
        \addtolength{\leftskip} {-2cm}
        \addtolength{\rightskip}{-2cm}
        \scalebox{0.8}{
		\begin{tabular}{lrrr|rrr|rrr}
		  \toprule
          \quad & \multicolumn{3}{c}{$U[0,100]$} & \multicolumn{3}{c}{$N(50,16.67^2)$} & \multicolumn{3}{c}{$Exp(0.08)$}  \\
            \hline
            Datasets & $\rho=2$ & $\rho=3$ & $\rho=4$ & $\rho=2$ & $\rho=3$ & $\rho=4$ & $\rho=2$ & $\rho=3$ & $\rho=4$ \\
            \hline
            LastFM &  79.9637 & 86.3822 & 90.6310 & 65.4908 & 70.2645 & 74.5560 &  22.9801 & 28.6986 & 34.5107\\ 
            \hline
		  Facebook & 90.4085 & 92.9271 & 93.8019 & 74.1225 & 76.2327 & 77.7093 & 33.7885 & 36.4504 & 38.8228\\
          \bottomrule
		\end{tabular}}
        \caption{Comparison of expected revenue of \texttt{APX-R} in sparse and dense social networks.}
		\label{table_bignet_xishumiji_rev}
	\end{table}}

\section{Conclusions and Discussion}
This work has answered two questions: (1) In order to optimize the seller's revenue  beyond what was achievable under Myerson optimal, by what means can the auctioneer exploit the economic networks, and attract more buyers from the economic network? (2) In turn, how is a reserve price determined in these new mechanisms such that the seller's revenue can be (approximately) maximized? 
We first designed the approximation mechanism for network auctions taking into consideration of a reserve price. Then we showed that using the format of an optimal reserve price in this mechanism is not applicable since it is complex to compute and also ruins DSIC. So we switched our target to  the reserve price which approaches the optimal reserve but still maintains DSIC.  The approximation ratio analysis shows  that this mechanism has very competitive performance on  seller's revenue. 

This approximation mechanism is particularly effective for sellers who already have a substantial audience at their disposal. In such scenarios, the mechanism will perform much more better than the Myerson optimal auction, and approximate the maximum possible revenue more closely. This indicates the result of this paper could be especially suitable for realistic large markets, such as internet auctions. The foundation of the reserve price optimization is eq.\eqref{eq_uniform_APX_Tx_r}, where the optimal reserves are different for each sub-tree. This is strongly related to the price discrimination \citep{bulow1989simple} and a more recent work investigates dynamic price discrimination against a changing number of bidders \citep{chaves2024auction}.
Our enhancement of  sellers' revenue can be even further improved taking into consideration of each agent's position in the network, in this case, it becomes a personalized price discrimination. 
Research about correlations of the values \citep{carrasco2018optimal, che2021robustly,he2022correlation,liu2023budget} could have underlying relation with our network setting.
Our immediate future work is to drop the i.i.d. and explore different valuation models in economic networks, for instance, independent private values (IPV), and more general value classes that allow for network-based value correlation. Extending the reserve price for economic networks to the combinatorial auction setting is also an interesting topic \citep{pekevc2003combinatorial}.

Another discussion is about a generalization of the allocation on the graph. The intuition of our allocation rule is:  allocate the item to the agent who has both reasonably high value (highest in $N-d_{j+1}$) and reasonably critical social position (leading the whole DCS). The ordering in  DCS  determines the agents' priority of winning. 
We allocate the item to the buyer with the highest priority who satisfies the two winning conditions.
This question is related to the core of all network auctions. The allocation in network auctions should consider both the value and the topology. We choose IDM-like allocation as the basis for setting reserve price and revenue optimization, since it is the simplest and most representative mechanism.
In addition, the problem of "selecting other agents (e.g., some middle node) in the DCS" is strongly related to the secretary problem and prophet inequality \citep{correa2024sample}. Beyond auction, using the power of crowds (combining many individual abilities to obtain
an aggregate decision-making) can be an effective technique for improving the decision system's performance.
These have been seen in the wisdom of crowds \citep{surowiecki2005wisdom}, and recent progresses considered shared information in this field \citep{palley2019extracting}. Thus exploiting the power of the networks underlying the crowds to make a more accurate wisdom of crowds is valuable future topic and can provide fundamental framework for crowdsourcing \citep{shi2022social}. Extending the network-based auction design to multi-unit item and even combinatorial auction scenarios \citep{pekevc2003combinatorial} and optimizing the revenue in these scenarios are challenging future topics.

\newpage



 
\appendix



\setcounter{equation}{0}
\renewcommand{\theequation}{A\arabic{equation}}

\setcounter{figure}{0}
\renewcommand{\thefigure}{A\arabic{figure}}

\renewcommand{\thelemma}{A\arabic{lemma}}

\setcounter{proposition}{0}
\renewcommand{\theproposition}{A\arabic{proposition}}

\setcounter{definition}{0}
\renewcommand{\thedefinition}{A\arabic{definition}}

\setcounter{table}{0}
\renewcommand{\thetable}{A\arabic{table}}

\section{Non-DSIC of \texttt{maxViVa}}
See \citep{bhattacharyya2023optimal}
for detailed definitions and notations, including transformed auction and \texttt{maxViVa}. 
According to the definition of transformed auction, each sub-tree $\hat{T_i}$, $i \in children(s)$ is replaced with a node having $v_{\ell}:=\max _{j \in \hat{T}_{\ell}} v_j$. Recall that the optimal referral auction is designed based on the assumption that the valuation of all buyers are i.i.d, assume that the cdf of buyers' valuations is $H$  and the corresponding pdf is $h$. Moreover, assume that the size of sub-tree $\hat{T_{\ell}}$ is $k_{\ell}$. Consequently, the cdf of $v_{\ell}$ is $F_{\ell}(v_{\ell})=H^{k_{\ell}}(v_{\ell})$, and the corresponding p.d.f is $f_{\ell}(v_{\ell})=k_{\ell}\cdot H^{k_{\ell}-1}(v_{\ell})\cdot h(v_{\ell})$. Then the \textit{virtual valuation} of agent $\ell$ in the transformed auction is given by:
{ {
\begin{equation}
{w_{\ell}(v_{\ell})=v_{\ell}-\frac{1-F_{\ell}(v_{\ell})}{f_{\ell}(v_{\ell})}=v_{\ell}-\frac{1-H^{k_{\ell}}(v_{\ell})}{k_{\ell}\cdot H^{k_{\ell}-1}(v_{\ell})}}
\end{equation}}}Notably, $k_{\ell}$ in the above equation is strongly correlated with the diffusion strategy of each buyer and thus can be influenced by buyers' actions. In essence, the \textit{virtual valuation} of agent $\ell$ is simultaneously determined by $v_{\ell}$ and $k_{\ell}$.

Taking the derivative of the virtual valuation $w_{\ell}(v_{\ell})$ with respect to $k_{\ell}$ yields:
{ {
\begin{equation}
{\frac{\partial{{w_{\ell}(v_{\ell})}}}{\partial {k_{\ell}}}=\frac{-H^{k_{\ell}}(v_{\ell})+1+k_{\ell} \ln H(v_{\ell})}{k_{\ell}^2 H^{k_{\ell}-1}(v_{\ell}) \cdot h(v_{\ell})}}
\end{equation}}}

Let $\varphi(k_{\ell})=-H^{k_{\ell}}(v_{\ell})+1+k_{\ell} \ln H(v_{\ell})$. The first-order derivative of $\varphi(k_{\ell})$ with respect to $k_{\ell}$ is $\varphi^{\prime}(k_{\ell})=\ln H(v_{\ell}) \cdot(1-H^{k_{\ell}}(v_{\ell}))<0$. It can be easily verified that $\varphi(k_{\ell})$ is monotonically decreasing with $k_{\ell}$. Therefore, $\varphi(k_{\ell}) \leq \varphi(1)=-H(v_{\ell})+1+\ln H(v_{\ell})$ holds for all $k_{\ell} \geq 1$, while $\varphi(1)$ must be non-positive due to $H(v_{\ell}) \in [0,1]$. Consequently,

\begin{equation}
{\frac{\partial{{w_{\ell}(v_{\ell})}}}{\partial {k_{\ell}}} \leq \frac{1+\ln H(v_{\ell})-H(v_{\ell})}{k_{\ell}^2 H^{k_{\ell}-1}(v_{\ell}) \cdot h(v_{\ell})} \leq 0}
\end{equation}

Thus, the \textit{virtual valuation} of agent $\ell$ in transformed auction is monotonically decreasing with respect to $k_{\ell}$, indicating that \textit{virtual valuation} may decrease with full propagation, leading to lower utility. In conclusion, the proposed revenue-optimal referral auction \texttt{maxViVa} in \citep{bhattacharyya2023optimal} is not DSIC.

Besides the Non-DSIC problem, in the paper, the characterization of the constraint of allocation and payment in Theorem $1$  also has an error. Firstly, given other buyers' reported types, buyer's utility is related not only to her reported type, but also to her true real type, but the utility function in the paper only covers the former. Moreover, buyer's utility function equals to her true valuation multiplied by her allocation minus her payment. In the proof, however, the equivalence inequality derived from inequality (3) defaults to a utility function that is the reported valuation multiplied by the allocation minus the payment. 

\section{Maximum of  Term $\sum_{x=1}^m \frac{k_x}{n-k_x+1}$ in the  Ratio $\frac{\operatorname{APX}(\mathbf{t'};\gamma)}{\operatorname{OPT}(N)}$ } \label{ratio_OPTN}
With a simple mathematical transformation, we know that
{\small{
\begin{displaymath}
\sum_{x=1}^m \frac{k_x}{n-k_x+1}=\sum_{x=1}^m \frac{k_x}{n-k_x+1}
=\sum_{x=1}^m \left(-1+\frac{n+1}{n-k_x+1}\right)=-m+\sum_{x=1}^m \frac{n+1}{n-k_x+1}
\end{displaymath}}}
When $m$ is a fixed value, $\big(-m+\sum_{x=1}^m \frac{n+1}{n-k_x+1}\big)$ has a maximum value when there are $(m-1)$ sub-trees with size $k_x=\underline{k}$ and one sub-tree with size $k_x=n-(m-1)\underline{k}$. And the corresponding maximum value is 
{ {
\begin{displaymath}
\begin{aligned}
q(m)=&\frac{(m-1)\underline{k}}{n-\underline{k}+1}+\frac{n-m\underline{k}+\underline{k}}{m\underline{k}-\underline{k}+1}\\
=&-1+\frac{(m-1)\underline{k}}{n-\underline{k}+1}+\frac{n+1}{m\underline{k}-\underline{k}+1}.
\end{aligned}
\end{displaymath}}}
Since $n\geq m\underline{k} \geq m\underline{k}-\underline{k}$, the first-order derivative $q'(m)$ satisfies 
{ {
\begin{displaymath}
\begin{aligned}
q'(m)=&\frac{\underline{k}}{n-\underline{k}+1}-\frac{(n+1)\underline{k}}{(m\underline{k}-\underline{k}+1)^2}\\
\leq &\frac{\underline{k}}{m\underline{k}-\underline{k}+1}\left(1-\frac{n+1}{m\underline{k}-\underline{k}+1}\right)\\
\leq &0.
\end{aligned}
\end{displaymath}}}
So $q(m)$ is decreasing on $m$. Consequently, $\sum_{x=1}^m \frac{k_x}{n-k_x+1}$ reaches the maximum value 
\begin{displaymath}
\frac{(\rho-1)\underline{k}}{n-\underline{k}+1}+\frac{n-\rho\underline{k}+\underline{k}}{\rho\underline{k}-\underline{k}+1},
\end{displaymath}
when $m=\rho$ and there are $(\rho-1)$ sub-trees with a size of $k_x=\underline{k}$ and one sub-tree with size  $k_x=n-(\rho-1)\underline{k}$.

\newpage




\bibliographystyle{elsarticle-num-names}


\bibliography{a-aaai24-all-recover}

\begin{thebibliography}{64}
\expandafter\ifx\csname natexlab\endcsname\relax\def\natexlab#1{#1}\fi
\providecommand{\url}[1]{\texttt{#1}}
\providecommand{\href}[2]{#2}
\providecommand{\path}[1]{#1}
\providecommand{\DOIprefix}{doi:}
\providecommand{\ArXivprefix}{arXiv:}
\providecommand{\URLprefix}{URL: }
\providecommand{\Pubmedprefix}{pmid:}
\providecommand{\doi}[1]{\href{http://dx.doi.org/#1}{\path{#1}}}
\providecommand{\Pubmed}[1]{\href{pmid:#1}{\path{#1}}}
\providecommand{\bibinfo}[2]{#2}
\ifx\xfnm\relax \def\xfnm[#1]{\unskip,\space#1}\fi
\bibitem[{Bulow and Klemperer(1996)}]{Bulow1996Auction}
\bibinfo{author}{J.~Bulow}, \bibinfo{author}{P.~Klemperer},
\newblock \bibinfo{title}{Auction versus negotiations},
\newblock \bibinfo{journal}{American Economic Review} \bibinfo{volume}{86} (\bibinfo{year}{1996}) \bibinfo{pages}{180--94}.
\bibitem[{Jackson et~al.(2008)}]{jackson2008social}
\bibinfo{author}{M.~O. Jackson}, et~al., \bibinfo{title}{Social and economic networks}, volume~\bibinfo{volume}{3}, \bibinfo{publisher}{Princeton University Press}, \bibinfo{year}{2008}.
\bibitem[{Borgatti et~al.(2009)Borgatti, Mehra, Brass, and Labianca}]{borgatti2009network}
\bibinfo{author}{S.~P. Borgatti}, \bibinfo{author}{A.~Mehra}, \bibinfo{author}{D.~J. Brass}, \bibinfo{author}{G.~Labianca},
\newblock \bibinfo{title}{Network analysis in the social sciences},
\newblock \bibinfo{journal}{Science} \bibinfo{volume}{323} (\bibinfo{year}{2009}) \bibinfo{pages}{892--895}.
\bibitem[{Myerson(1981)}]{myerson1981optimal}
\bibinfo{author}{R.~B. Myerson},
\newblock \bibinfo{title}{Optimal auction design},
\newblock \bibinfo{journal}{Mathematics of Operations Research} \bibinfo{volume}{6} (\bibinfo{year}{1981}) \bibinfo{pages}{58--73}.
\bibitem[{Riley and Samuelson(1981)}]{riley1981optimal}
\bibinfo{author}{J.~G. Riley}, \bibinfo{author}{W.~F. Samuelson},
\newblock \bibinfo{title}{Optimal auctions},
\newblock \bibinfo{journal}{The American Economic Review} \bibinfo{volume}{71} (\bibinfo{year}{1981}) \bibinfo{pages}{381--392}.
\bibitem[{Engelbrecht-Wiggans(1987)}]{engelbrecht1987optimal}
\bibinfo{author}{R.~Engelbrecht-Wiggans},
\newblock \bibinfo{title}{On optimal reservation prices in auctions},
\newblock \bibinfo{journal}{Management Science} \bibinfo{volume}{33} (\bibinfo{year}{1987}) \bibinfo{pages}{763--770}.
\bibitem[{Levin and Smith(1996)}]{levin1996optimal}
\bibinfo{author}{D.~Levin}, \bibinfo{author}{J.~L. Smith},
\newblock \bibinfo{title}{Optimal reservation prices in auctions},
\newblock \bibinfo{journal}{The Economic Journal} \bibinfo{volume}{106} (\bibinfo{year}{1996}) \bibinfo{pages}{1271--1283}.
\bibitem[{Wolfstetter(1996)}]{wolfstetter1996auctions}
\bibinfo{author}{E.~Wolfstetter},
\newblock \bibinfo{title}{Auctions: an introduction},
\newblock \bibinfo{journal}{Journal of Economic Surveys} \bibinfo{volume}{10} (\bibinfo{year}{1996}) \bibinfo{pages}{367--420}.
\bibitem[{Menicucci(2021)}]{menicucci2021basic}
\bibinfo{author}{D.~Menicucci},
\newblock \bibinfo{title}{In the basic auction model, the optimal reserve price may depend on the number of bidders},
\newblock \bibinfo{journal}{Journal of Economic Theory} \bibinfo{volume}{198} (\bibinfo{year}{2021}) \bibinfo{pages}{105371}.
\bibitem[{Eden et~al.(2016)Eden, Feldman, Friedler, Talgam-Cohen, and Weinberg}]{eden2016competition}
\bibinfo{author}{A.~Eden}, \bibinfo{author}{M.~Feldman}, \bibinfo{author}{O.~Friedler}, \bibinfo{author}{I.~Talgam-Cohen}, \bibinfo{author}{S.~M. Weinberg},
\newblock \bibinfo{title}{The competition complexity of auctions: A bulow-klemperer result for multi-dimensional bidders},
\newblock \bibinfo{journal}{arXiv preprint arXiv:1612.08821}  (\bibinfo{year}{2016}).
\bibitem[{Esfandiari et~al.(2019)Esfandiari, HajiAghayi, Lucier, and Mitzenmacher}]{esfandiari2019online}
\bibinfo{author}{H.~Esfandiari}, \bibinfo{author}{M.~HajiAghayi}, \bibinfo{author}{B.~Lucier}, \bibinfo{author}{M.~Mitzenmacher},
\newblock \bibinfo{title}{Online pandora’s boxes and bandits},
\newblock in: \bibinfo{booktitle}{Proceedings of the 33rd AAAI Conference on Artificial Intelligence}, volume~\bibinfo{volume}{33}, \bibinfo{year}{2019}, pp. \bibinfo{pages}{1885--1892}.
\bibitem[{Cai and Saxena(2021)}]{cai202199}
\bibinfo{author}{L.~Cai}, \bibinfo{author}{R.~R. Saxena},
\newblock \bibinfo{title}{99\% revenue with constant enhanced competition},
\newblock in: \bibinfo{booktitle}{Proceedings of the 22nd ACM Conference on Economics and Computation}, \bibinfo{year}{2021}, pp. \bibinfo{pages}{224--241}.
\bibitem[{Beyhaghi and Weinberg(2019)}]{beyhaghi2019optimal}
\bibinfo{author}{H.~Beyhaghi}, \bibinfo{author}{S.~M. Weinberg},
\newblock \bibinfo{title}{Optimal (and benchmark-optimal) competition complexity for additive buyers over independent items},
\newblock in: \bibinfo{booktitle}{Proceedings of the 51st Annual ACM SIGACT Symposium on Theory of Computing}, \bibinfo{year}{2019}, pp. \bibinfo{pages}{686--696}.
\bibitem[{Bei et~al.(2023)Bei, Gravin, Lu, and Tang}]{bei2023bidder}
\bibinfo{author}{X.~Bei}, \bibinfo{author}{N.~Gravin}, \bibinfo{author}{P.~Lu}, \bibinfo{author}{Z.~G. Tang},
\newblock \bibinfo{title}{Bidder subset selection problem in auction design},
\newblock in: \bibinfo{booktitle}{Proceedings of the 2023 Annual ACM-SIAM Symposium on Discrete Algorithms}, \bibinfo{year}{2023}, pp. \bibinfo{pages}{3788--3801}.
\bibitem[{Allouah and Besbes(2020)}]{allouah2020prior}
\bibinfo{author}{A.~Allouah}, \bibinfo{author}{O.~Besbes},
\newblock \bibinfo{title}{Prior-independent optimal auctions},
\newblock \bibinfo{journal}{Management Science} \bibinfo{volume}{66} (\bibinfo{year}{2020}) \bibinfo{pages}{4417--4432}.
\bibitem[{McAfee and McMillan(1987)}]{mcafee1987auctions}
\bibinfo{author}{R.~P. McAfee}, \bibinfo{author}{J.~McMillan},
\newblock \bibinfo{title}{Auctions with a stochastic number of bidders},
\newblock \bibinfo{journal}{Journal of economic theory} \bibinfo{volume}{43} (\bibinfo{year}{1987}) \bibinfo{pages}{1--19}.
\bibitem[{Levin and Ozdenoren(2004)}]{levin2004auctions}
\bibinfo{author}{D.~Levin}, \bibinfo{author}{E.~Ozdenoren},
\newblock \bibinfo{title}{Auctions with uncertain numbers of bidders},
\newblock \bibinfo{journal}{Journal of Economic Theory} \bibinfo{volume}{118} (\bibinfo{year}{2004}) \bibinfo{pages}{229--251}.
\bibitem[{Krishna(2009)}]{krishna2009auction}
\bibinfo{author}{V.~Krishna}, \bibinfo{title}{Auction Theory}, \bibinfo{publisher}{Academic press}, \bibinfo{year}{2009}.
\bibitem[{Fu et~al.(2015)Fu, Immorlica, Lucier, and Strack}]{fu2015randomization}
\bibinfo{author}{H.~Fu}, \bibinfo{author}{N.~Immorlica}, \bibinfo{author}{B.~Lucier}, \bibinfo{author}{P.~Strack},
\newblock \bibinfo{title}{Randomization beats second price as a prior-independent auction},
\newblock in: \bibinfo{booktitle}{Proceedings of the 16th ACM Conference on Economics and Computation}, \bibinfo{year}{2015}, pp. \bibinfo{pages}{323--323}.
\bibitem[{Ko{\c{c}}yi{\u{g}}it et~al.(2020)Ko{\c{c}}yi{\u{g}}it, Iyengar, Kuhn, and Wiesemann}]{koccyiugit2020distributionally}
\bibinfo{author}{{\c{C}}.~Ko{\c{c}}yi{\u{g}}it}, \bibinfo{author}{G.~Iyengar}, \bibinfo{author}{D.~Kuhn}, \bibinfo{author}{W.~Wiesemann},
\newblock \bibinfo{title}{Distributionally robust mechanism design},
\newblock \bibinfo{journal}{Management Science} \bibinfo{volume}{66} (\bibinfo{year}{2020}) \bibinfo{pages}{159--189}.
\bibitem[{Hartline(2013)}]{hartline2013mechanism}
\bibinfo{author}{J.~D. Hartline},
\newblock \bibinfo{title}{Mechanism design and approximation},
\newblock \bibinfo{journal}{Book draft. October} \bibinfo{volume}{122} (\bibinfo{year}{2013}).
\bibitem[{Roughgarden and Talgam-Cohen(2016)}]{roughgarden2016optimal}
\bibinfo{author}{T.~Roughgarden}, \bibinfo{author}{I.~Talgam-Cohen},
\newblock \bibinfo{title}{Optimal and robust mechanism design with interdependent values},
\newblock \bibinfo{journal}{ACM Transactions on Economics and Computation} \bibinfo{volume}{4} (\bibinfo{year}{2016}) \bibinfo{pages}{1--34}.
\bibitem[{Shah et~al.(2019)Shah, Johari, and Blanchet}]{shah2019semi}
\bibinfo{author}{V.~Shah}, \bibinfo{author}{R.~Johari}, \bibinfo{author}{J.~Blanchet},
\newblock \bibinfo{title}{Semi-parametric dynamic contextual pricing},
\newblock \bibinfo{journal}{Proceedings of the 33rd Conference on Neural Information Processing Systems} \bibinfo{volume}{32} (\bibinfo{year}{2019}) \bibinfo{pages}{2363--2373}.
\bibitem[{Castiglioni et~al.(2023)Castiglioni, Marchesi, Romano, and Gatti}]{castiglioni2023increasing}
\bibinfo{author}{M.~Castiglioni}, \bibinfo{author}{A.~Marchesi}, \bibinfo{author}{G.~Romano}, \bibinfo{author}{N.~Gatti},
\newblock \bibinfo{title}{Increasing revenue in bayesian posted price auctions through signaling},
\newblock \bibinfo{journal}{Artificial Intelligence}  (\bibinfo{year}{2023}) \bibinfo{pages}{103990}.
\bibitem[{Bajari and Horta{\c{c}}su(2003)}]{bajari2003winner}
\bibinfo{author}{P.~Bajari}, \bibinfo{author}{A.~Horta{\c{c}}su},
\newblock \bibinfo{title}{The winner's curse, reserve prices, and endogenous entry: Empirical insights from ebay auctions},
\newblock \bibinfo{journal}{RAND Journal of Economics}  (\bibinfo{year}{2003}) \bibinfo{pages}{329--355}.
\bibitem[{Lucking-Reiley et~al.(2007)Lucking-Reiley, Bryan, Prasad, and Reeves}]{lucking2007pennies}
\bibinfo{author}{D.~Lucking-Reiley}, \bibinfo{author}{D.~Bryan}, \bibinfo{author}{N.~Prasad}, \bibinfo{author}{D.~Reeves},
\newblock \bibinfo{title}{Pennies from ebay: The determinants of price in online auctions},
\newblock \bibinfo{journal}{The Journal of Industrial Economics} \bibinfo{volume}{55} (\bibinfo{year}{2007}) \bibinfo{pages}{223--233}.
\bibitem[{Barrymore and Raviv(2009)}]{barrymore2009effect}
\bibinfo{author}{N.~Barrymore}, \bibinfo{author}{Y.~Raviv},
\newblock \bibinfo{title}{The effect of different reserve prices on auction outcomes},
\newblock \bibinfo{journal}{Robert Day School of Economics and Finance Research Paper}  (\bibinfo{year}{2009}).
\bibitem[{Li et~al.(2017)Li, Hao, Zhao, and Zhou}]{li2017mechanism}
\bibinfo{author}{B.~Li}, \bibinfo{author}{D.~Hao}, \bibinfo{author}{D.~Zhao}, \bibinfo{author}{T.~Zhou},
\newblock \bibinfo{title}{Mechanism design in social networks},
\newblock in: \bibinfo{booktitle}{Proceedings of the 31st AAAI Conference on Artificial Intelligence}, \bibinfo{year}{2017}, pp. \bibinfo{pages}{586--592}.
\bibitem[{Lee(2017)}]{lee2017mechanisms}
\bibinfo{author}{J.~Lee},
\newblock \bibinfo{title}{Mechanisms with referrals: Vcg mechanisms and multilevel mechanisms}  (\bibinfo{year}{2017}).
\bibitem[{Guo and Hao(2021)}]{guo2021emerging}
\bibinfo{author}{Y.~Guo}, \bibinfo{author}{D.~Hao},
\newblock \bibinfo{title}{Emerging methods of auction design in social networks},
\newblock in: \bibinfo{booktitle}{Proceedings of the 30th International Joint Conference on Artificial Intelligence}, \bibinfo{year}{2021}, pp. \bibinfo{pages}{4434--4441}.
\bibitem[{Li et~al.(2022)Li, Hao, Gao, and Zhao}]{li2022diffusion}
\bibinfo{author}{B.~Li}, \bibinfo{author}{D.~Hao}, \bibinfo{author}{H.~Gao}, \bibinfo{author}{D.~Zhao},
\newblock \bibinfo{title}{Diffusion auction design},
\newblock \bibinfo{journal}{Artificial Intelligence} \bibinfo{volume}{303} (\bibinfo{year}{2022}) \bibinfo{pages}{103631}.
\bibitem[{Jeong and Lee(2024)}]{jeong2024groupwise}
\bibinfo{author}{S.~E. Jeong}, \bibinfo{author}{J.~Lee},
\newblock \bibinfo{title}{The groupwise-pivotal referral auction: Core-selecting referral strategy-proof mechanism},
\newblock \bibinfo{journal}{Games and Economic Behavior} \bibinfo{volume}{143} (\bibinfo{year}{2024}) \bibinfo{pages}{191--203}.
\bibitem[{Li et~al.(2020)Li, Hao, and Zhao}]{li2020incentive}
\bibinfo{author}{B.~Li}, \bibinfo{author}{D.~Hao}, \bibinfo{author}{D.~Zhao},
\newblock \bibinfo{title}{Incentive-compatible diffusion auctions},
\newblock in: \bibinfo{booktitle}{Proceedings of the 29th International Conference on International Joint Conferences on Artificial Intelligence}, \bibinfo{year}{2020}, pp. \bibinfo{pages}{231--237}.
\bibitem[{Lee(2017)}]{Lee2017MechanismsWR}
\bibinfo{author}{J.~Lee},
\newblock \bibinfo{title}{Mechanisms with referrals: Vcg mechanisms and multilevel mechanisms},
\newblock in: \bibinfo{booktitle}{Available at SSRN: https://ssrn.com/abstract=2987761}, \bibinfo{year}{2017}.
\bibitem[{Li et~al.(2018)Li, Hao, Zhao, and Zhou}]{Li2018CustomerSI}
\bibinfo{author}{B.~Li}, \bibinfo{author}{D.~Hao}, \bibinfo{author}{D.~Zhao}, \bibinfo{author}{T.~Zhou},
\newblock \bibinfo{title}{Customer sharing in economic networks with costs},
\newblock in: \bibinfo{booktitle}{Proceedings of the 27th International Joint Conference on Artificial Intelligence}, \bibinfo{year}{2018}, pp. \bibinfo{pages}{368--374}.
\bibitem[{{Li} et~al.(2019){Li}, {Hao}, {Zhao}, and {Yokoo}}]{li2019graph}
\bibinfo{author}{B.~{Li}}, \bibinfo{author}{D.~{Hao}}, \bibinfo{author}{D.~{Zhao}}, \bibinfo{author}{M.~{Yokoo}},
\newblock \bibinfo{title}{Diffusion and auction on graphs},
\newblock in: \bibinfo{booktitle}{Proceedings of the 28th International Joint Conference on Artificial Intelligence}, \bibinfo{year}{2019}, pp. \bibinfo{pages}{435--441}.
\bibitem[{Kawasaki et~al.(2021)Kawasaki, Wada, Todo, and Yokoo}]{kawasaki2021mechanism}
\bibinfo{author}{T.~Kawasaki}, \bibinfo{author}{R.~Wada}, \bibinfo{author}{T.~Todo}, \bibinfo{author}{M.~Yokoo},
\newblock \bibinfo{title}{Mechanism design for housing markets over social networks},
\newblock in: \bibinfo{booktitle}{Proceedings of the 20th International Conference on Autonomous Agents and MultiAgent Systems}, \bibinfo{year}{2021}, pp. \bibinfo{pages}{692--700}.
\bibitem[{Liu et~al.(2021)Liu, Wu, Li, and Wang}]{liu2021budget}
\bibinfo{author}{X.~Liu}, \bibinfo{author}{W.~Wu}, \bibinfo{author}{M.~Li}, \bibinfo{author}{W.~Wang},
\newblock \bibinfo{title}{Budget feasible mechanisms over graphs},
\newblock in: \bibinfo{booktitle}{Proceedings of the AAAI Conference on Artificial Intelligence}, \bibinfo{year}{2021}, pp. \bibinfo{pages}{5549--5556}.
\bibitem[{Guo et~al.(2023)Guo, Hao, Xiao, and Li}]{guo2022combinatorial}
\bibinfo{author}{Y.~Guo}, \bibinfo{author}{D.~Hao}, \bibinfo{author}{M.~Xiao}, \bibinfo{author}{B.~Li},
\newblock \bibinfo{title}{Networked combinatorial auction for crowdsourcing and crowdsensing},
\newblock \bibinfo{journal}{IEEE Internet of Things Journal}  (\bibinfo{year}{2023}).
\bibitem[{Shi and Hao(2022)}]{shi2022social}
\bibinfo{author}{Q.~Shi}, \bibinfo{author}{D.~Hao},
\newblock \bibinfo{title}{Social sourcing: Incorporating social networks into crowdsourcing contest design},
\newblock \bibinfo{journal}{IEEE/ACM Transactions on Networking}  (\bibinfo{year}{2022}) \bibinfo{pages}{1--15}.
\bibitem[{Li et~al.(2024)Li, Hao, and Zhao}]{li2024diffusion}
\bibinfo{author}{B.~Li}, \bibinfo{author}{D.~Hao}, \bibinfo{author}{D.~Zhao},
\newblock \bibinfo{title}{Diffusion auction design with transaction costs},
\newblock \bibinfo{journal}{Autonomous Agents and Multi-Agent Systems} \bibinfo{volume}{38} (\bibinfo{year}{2024}) \bibinfo{pages}{2}.
\bibitem[{Li et~al.(2019)Li, Hao, Zhao, and Yokoo}]{li2019diffusion}
\bibinfo{author}{B.~Li}, \bibinfo{author}{D.~Hao}, \bibinfo{author}{D.~Zhao}, \bibinfo{author}{M.~Yokoo},
\newblock \bibinfo{title}{Diffusion and auction on graphs},
\newblock in: \bibinfo{booktitle}{Proceedings of the 28th International Joint Conference on Artificial Intelligence}, \bibinfo{year}{2019}, pp. \bibinfo{pages}{435--441}.
\bibitem[{Zhang et~al.(2020{\natexlab{a}})Zhang, Zhao, and Chen}]{zhang2020redistribution}
\bibinfo{author}{W.~Zhang}, \bibinfo{author}{D.~Zhao}, \bibinfo{author}{H.~Chen},
\newblock \bibinfo{title}{Redistribution mechanism on networks},
\newblock in: \bibinfo{booktitle}{Proceedings of the 19th International Conference on Autonomous Agents and MultiAgent Systems}, \bibinfo{year}{2020}{\natexlab{a}}, pp. \bibinfo{pages}{1620--1628}.
\bibitem[{Zhang et~al.(2020{\natexlab{b}})Zhang, Zhao, and Zhang}]{zhangECAIincentivize}
\bibinfo{author}{W.~Zhang}, \bibinfo{author}{D.~Zhao}, \bibinfo{author}{Y.~Zhang},
\newblock \bibinfo{title}{Incentivize diffusion with fair rewards},
\newblock in: \bibinfo{booktitle}{24th European Conference on Artificial Intelligence, ECAI 2020, including 10th Conference on Prestigious Applications of Artificial Intelligence, PAIS 2020}, volume \bibinfo{volume}{325}, \bibinfo{year}{2020}{\natexlab{b}}, pp. \bibinfo{pages}{251--258}.
\bibitem[{Bhattacharyya et~al.(2023)Bhattacharyya, Dave, Dey, and Nath}]{bhattacharyya2023optimal}
\bibinfo{author}{R.~Bhattacharyya}, \bibinfo{author}{P.~Dave}, \bibinfo{author}{P.~Dey}, \bibinfo{author}{S.~Nath}, \bibinfo{title}{Optimal Referral Auction Design}, \bibinfo{type}{Papers}, arXiv.org, \bibinfo{year}{2023}.
\bibitem[{Zhang et~al.(2023)Zhang, Zheng, and Zhao}]{zhang2023optimal}
\bibinfo{author}{Y.~Zhang}, \bibinfo{author}{S.~Zheng}, \bibinfo{author}{D.~Zhao},
\newblock \bibinfo{title}{Optimal diffusion auctions},
\newblock \bibinfo{journal}{arXiv preprint arXiv:2302.02580}  (\bibinfo{year}{2023}).
\bibitem[{Yu(2024)}]{yu2024diffusion}
\bibinfo{author}{F.~Yu},
\newblock \bibinfo{title}{Diffusion mechanism design in tree-structured social network},
\newblock \bibinfo{journal}{arXiv preprint arXiv:2407.21143}  (\bibinfo{year}{2024}).
\bibitem[{Jeong and Lee(2023)}]{jeong2023groupwise}
\bibinfo{author}{S.~E. Jeong}, \bibinfo{author}{J.~Lee},
\newblock \bibinfo{title}{The groupwise-pivotal referral auction: Core-selecting referral strategy-proof mechanism},
\newblock \bibinfo{journal}{Games and Economic Behavior}  (\bibinfo{year}{2023}).
\bibitem[{Zhang and Tang(2023)}]{zhang2023truthful}
\bibinfo{author}{Y.~Zhang}, \bibinfo{author}{P.~Tang},
\newblock \bibinfo{title}{A truthful referral auction over networks},
\newblock \bibinfo{journal}{arXiv preprint arXiv:2302.08135}  (\bibinfo{year}{2023}).
\bibitem[{Alizamir et~al.(2022)Alizamir, Chen, Kim, and Manshadi}]{alizamir2022impact}
\bibinfo{author}{S.~Alizamir}, \bibinfo{author}{N.~Chen}, \bibinfo{author}{S.-H. Kim}, \bibinfo{author}{V.~Manshadi},
\newblock \bibinfo{title}{Impact of network structure on new service pricing},
\newblock \bibinfo{journal}{Mathematics of Operations Research} \bibinfo{volume}{47} (\bibinfo{year}{2022}) \bibinfo{pages}{1999--2033}.
\bibitem[{Celis et~al.(2014)Celis, Lewis, Mobius, and Nazerzadeh}]{celis2014buy}
\bibinfo{author}{L.~E. Celis}, \bibinfo{author}{G.~Lewis}, \bibinfo{author}{M.~Mobius}, \bibinfo{author}{H.~Nazerzadeh},
\newblock \bibinfo{title}{Buy-it-now or take-a-chance: Price discrimination through randomized auctions},
\newblock \bibinfo{journal}{Management Science} \bibinfo{volume}{60} (\bibinfo{year}{2014}) \bibinfo{pages}{2927--2948}.
\bibitem[{Amaral et~al.(2000)Amaral, Scala, Barthelemy, and Stanley}]{amaral2000classes}
\bibinfo{author}{L.~A.~N. Amaral}, \bibinfo{author}{A.~Scala}, \bibinfo{author}{M.~Barthelemy}, \bibinfo{author}{H.~E. Stanley},
\newblock \bibinfo{title}{Classes of small-world networks},
\newblock \bibinfo{journal}{Proceedings of the national academy of sciences} \bibinfo{volume}{97} (\bibinfo{year}{2000}) \bibinfo{pages}{11149--11152}.
\bibitem[{Rossi and Ahmed(2015)}]{rossi2015network}
\bibinfo{author}{R.~Rossi}, \bibinfo{author}{N.~Ahmed},
\newblock \bibinfo{title}{The network data repository with interactive graph analytics and visualization},
\newblock in: \bibinfo{booktitle}{Proceedings of the AAAI conference on artificial intelligence}, \bibinfo{year}{2015}, p. \bibinfo{pages}{4292–4293}.
\bibitem[{Leskovec and Krevl(2014)}]{leskovec2014snap}
\bibinfo{author}{J.~Leskovec}, \bibinfo{author}{A.~Krevl}, \bibinfo{title}{Snap datasets: Stanford large network dataset collection}, \bibinfo{year}{2014}.
\bibitem[{Bulow and Roberts(1989)}]{bulow1989simple}
\bibinfo{author}{J.~Bulow}, \bibinfo{author}{J.~Roberts},
\newblock \bibinfo{title}{The simple economics of optimal auctions},
\newblock \bibinfo{journal}{Journal of Political Economy} \bibinfo{volume}{97} (\bibinfo{year}{1989}) \bibinfo{pages}{1060--1090}.
\bibitem[{Chaves and Ichihashi(2024)}]{chaves2024auction}
\bibinfo{author}{I.~N. Chaves}, \bibinfo{author}{S.~Ichihashi},
\newblock \bibinfo{title}{Auction timing and market thickness},
\newblock \bibinfo{journal}{Games and Economic Behavior} \bibinfo{volume}{143} (\bibinfo{year}{2024}) \bibinfo{pages}{161--178}.
\bibitem[{Carrasco et~al.(2018)Carrasco, Luz, Kos, Messner, Monteiro, and Moreira}]{carrasco2018optimal}
\bibinfo{author}{V.~Carrasco}, \bibinfo{author}{V.~F. Luz}, \bibinfo{author}{N.~Kos}, \bibinfo{author}{M.~Messner}, \bibinfo{author}{P.~Monteiro}, \bibinfo{author}{H.~Moreira},
\newblock \bibinfo{title}{Optimal selling mechanisms under moment conditions},
\newblock \bibinfo{journal}{Journal of Economic Theory} \bibinfo{volume}{177} (\bibinfo{year}{2018}) \bibinfo{pages}{245--279}.
\bibitem[{Che and Zhong(2021)}]{che2021robustly}
\bibinfo{author}{Y.-K. Che}, \bibinfo{author}{W.~Zhong},
\newblock \bibinfo{title}{Robustly-optimal mechanism for selling multiple goods},
\newblock in: \bibinfo{booktitle}{Proceedings of the 22nd ACM Conference on Economics and Computation}, \bibinfo{year}{2021}, pp. \bibinfo{pages}{314--315}.
\bibitem[{He and Li(2022)}]{he2022correlation}
\bibinfo{author}{W.~He}, \bibinfo{author}{J.~Li},
\newblock \bibinfo{title}{Correlation-robust auction design},
\newblock \bibinfo{journal}{Journal of Economic Theory} \bibinfo{volume}{200} (\bibinfo{year}{2022}) \bibinfo{pages}{105403}.
\bibitem[{Liu et~al.(2023)Liu, Chan, Li, Wu, and Zhao}]{liu2023budget}
\bibinfo{author}{X.~Liu}, \bibinfo{author}{H.~Chan}, \bibinfo{author}{M.~Li}, \bibinfo{author}{W.~Wu}, \bibinfo{author}{Y.~Zhao},
\newblock \bibinfo{title}{Budget-feasible mechanisms for proportionally selecting agents from groups},
\newblock \bibinfo{journal}{Artificial Intelligence}  (\bibinfo{year}{2023}) \bibinfo{pages}{103975}.
\bibitem[{Peke{\v{c}} and Rothkopf(2003)}]{pekevc2003combinatorial}
\bibinfo{author}{A.~Peke{\v{c}}}, \bibinfo{author}{M.~H. Rothkopf},
\newblock \bibinfo{title}{Combinatorial auction design},
\newblock \bibinfo{journal}{Management Science} \bibinfo{volume}{49} (\bibinfo{year}{2003}) \bibinfo{pages}{1485--1503}.
\bibitem[{Correa et~al.(2024)Correa, Cristi, Epstein, and Soto}]{correa2024sample}
\bibinfo{author}{J.~Correa}, \bibinfo{author}{A.~Cristi}, \bibinfo{author}{B.~Epstein}, \bibinfo{author}{J.~A. Soto},
\newblock \bibinfo{title}{Sample-driven optimal stopping: From the secretary problem to the iid prophet inequality},
\newblock \bibinfo{journal}{Mathematics of Operations Research} \bibinfo{volume}{49} (\bibinfo{year}{2024}) \bibinfo{pages}{441--475}.
\bibitem[{Surowiecki(2005)}]{surowiecki2005wisdom}
\bibinfo{author}{J.~Surowiecki}, \bibinfo{title}{The wisdom of crowds}, \bibinfo{publisher}{Anchor}, \bibinfo{year}{2005}.
\bibitem[{Palley and Soll(2019)}]{palley2019extracting}
\bibinfo{author}{A.~B. Palley}, \bibinfo{author}{J.~B. Soll},
\newblock \bibinfo{title}{Extracting the wisdom of crowds when information is shared},
\newblock \bibinfo{journal}{Management Science} \bibinfo{volume}{65} (\bibinfo{year}{2019}) \bibinfo{pages}{2291--2309}.

\end{thebibliography}



\end{document}